\documentclass{CSML}

\def\dOi{12(3:3)2016}
\lmcsheading%
{\dOi}
{1--43}
{}
{}
{Oct.~31, 2015}
{Aug.~17, 2016}
{}

\ACMCCS{[{\bf Theory of computation}]: Models of computation; Semantics and reasoning---Program semantics}

\usepackage{amssymb,stmaryrd,amsmath}

\usepackage{pstricks}
\usepackage{pst-node}

\usepackage[all]{xy}
\usepackage{wrapfig}

\theoremstyle{definition}\newtheorem{definition}[thm]{Definition}
\theoremstyle{definition}\newtheorem{example}[thm]{Example}
\theoremstyle{definition}\newtheorem{proposition}[thm]{Proposition}
\theoremstyle{definition}\newtheorem{lemma}[thm]{Lemma}
\theoremstyle{definition}\newtheorem{theorem}[thm]{Theorem}
\theoremstyle{definition}\newtheorem{corollary}[thm]{Corollary}
\theoremstyle{definition}\newtheorem{remark}[thm]{Remark}

\newcommand\nt[1]{#1}
\newcommand\am[1]{#1}

\newcommand\ev{\mathsf{ev}}
\newcommand\app{\ev}
\newcommand\Scat{\mathcal{S}}
\newcommand\Scatinn{\Scat_{\rm inn}}
\newcommand\Scatinntot{\Scatinn^{t}}
\newcommand\lland{\,\land\,}
\newcommand\qwe{\ \ }
\newcommand\qweq{\ }
\renewcommand\Sigma{\varSigma}
\newcommand\recur{\mathcal{R}}
\newcommand{\nada}{$\text{}$}
\newcommand\justf[3][]{\nccurve[nodesep=.5pt,linewidth=0.4pt,angleA=110,angleB=45,linecolor=darkgray#1]{->}{#2}{#3}}
\newcommand\Tau{T}
\newcommand\ov[1]{\llcorner{#1}\lrcorner}
\newcommand\defn{\triangleq}
\newcommand\Erase[1]{\mathsf{erase}(#1)}

\newcommand\plays[1]{P_{#1}}
\newcommand\arr{\rightarrow}
\newcommand\Arr{\Rightarrow}
\newcommand\nb{\beta}

\newcommand\rest{\upharpoonright}
\newcommand\iseq{\mathop{\|}}
\newcommand\mix{\,\raisebox{.25ex}{\scalebox{.5}{$\bullet$}}\,}
\newcommand\ee\epsilon
\newcommand\remv{\setminus}
\newcommand\nice{\boldsymbol\Phi}
\newcommand\actn{\cdot}
\renewcommand\int{\mathsf{int}}
\newcommand\var{\mathsf{var}}
\newcommand\List[1]{\mathsf{List}(#1)}
\newcommand\remove[3]{\mathsf{remove}\ (#1,#2)\ \mathsf{in}\ #3}
\newcommand\add[4]{\mathsf{insert}\ (#1,#2,#3)\ \mathsf{in}\ #4}
\newcommand\Hextend[3]{\mathsf{let}\ #1=\mathsf{cons}\, #2\, #1\ \mathsf{in}\ #3}
\newcommand\Textend[3]{\mathsf{let}\ #1=\mathsf{snoc}\, #1\, #2\ \mathsf{in}\ #3}
\newcommand\indx{\mathsf{indx}}
\newcommand\deindx[1]{\mathsf{deindx}\ #1}

\newcommand\code[1]{\lceil#1\rceil}
\newcommand\ang[1]{\langle#1\rangle}

\newcommand\one{1}
\newcommand\id{\mathsf{id}}
\newcommand\cell{\mathsf{cell}}

\newcommand\ok{\mathsf{ok}}
\newcommand\rd{\mathsf{read}}
\newcommand\viewf{\mathsf{viewf}}
\newcommand\strat{\mathsf{strat}}
\newcommand\sskip{\mathsf{skip}}
\newcommand\sqleq\sqsubseteq
\newcommand\sqle\sqsubset
\newcommand\preplays[1]{\mathit{PP}_{#1}}
\newcommand\Splays[1]{\mathit{SP}_{#1}}
\newcommand\st[1]{\mathsf{st}(#1)}
\newcommand\Sinter[1]{\mathit{SInt}(#1)}

\newcommand{\HY}{\hspace{2pt}---\hspace{2pt}}
\newcommand\clg[1]{\mathcal{#1}}
\newcommand\boldemph[1]{\emph{\textbf{#1}}}
\newcommand{\BI}[1]{\overline{I}_{#1}}

\newcommand{\impl}{\mathbin{\Rightarrow}}
\newcommand{\permg}{{\rm PERM}}
\newcommand\A{\mathbb{A}}

\newcommand\Z{\mathbb{Z}}

\newcommand\dom[1]{\mathsf{dom}\,#1}

\newcommand\langp{\mathcal{L}'}
\newcommand\eqmod[1]{\cong_{#1}}

\newcommand\mwrite[1]{\mathsf{write}(#1)}
\newcommand\mread{\mathsf{read}}
\newcommand\mok{\mathsf{ok}}
\newcommand\compi[2]{\mathcal{C}_{#1}^{#2}}

\newcommand\comp[1]{\mathcal{C}_{#1}}
\newcommand\lrarr\longrightarrow
\newcommand\ctype{\mathit{ctype}}
\newcommand\ttype{\mathit{ttype}}
\newcommand\spinal[1]{P_A^{sp}}
\newcommand\ialoop{\mathsf{IA}_{\circlearrowright}}
\newcommand\iatwo{\ialoop^{2+}}
\newcommand\pcfplus{\mathsf{PCF}^+}
\newcommand\pcf{\mathsf{PCF}}
\newcommand\iacbv{\mathsf{IA}_{\mathsf{cbv}}}
\newcommand\rml{\mathsf{RML}}

\newcommand\ds{\displaystyle}

\newcommand\letin[2]{\mathsf{let}\ #1\ \mathsf{in}\ #2}

\newcommand\badvar[2]{\mathsf{mkvar}(#1,#2)}
\newcommand\mkvar{\mathsf{mkvar}}

\newcommand\substore\leq
\newcommand\Substore{\leq_p}
\newcommand\substorE{\leq_s}
\newcommand\prefix{\sqsubseteq_p}
\newcommand\suffix{\sqsubseteq_s}
\newcommand\subseq\sqsubseteq
\newcommand\fix[1]{\mathsf{Y}(#1)}

\newcommand{\aasg}{\,\raisebox{0.065ex}{:}{=}\,}
\newcommand\comt{\mathsf{unit}}
\newcommand\Rarr{\Rightarrow}
\newcommand\expt{\mathsf{int}}
\newcommand\vart{\mathsf{var}}
\renewcommand\red[2]{#1 \Downarrow #2}

\newcommand\while[2]{\mathsf{while}\,#1\,\mathsf{do}\,#2}
\newcommand\cond[3]{\mathsf{if}\,#1\,\mathsf{then}\,#2\,\mathsf{else}\,#3}

\newcommand{\rarr}{\rightarrow}
\newcommand\pview[1]{\ulcorner{#1}\urcorner}
\newcommand\pv[1]{\pview{#1}}
\newcommand\oview[1]{\llcorner{#1}\lrcorner}

\newcommand\makeset[1]{\{\,#1\,\}}
\newcommand\cutout[1]{}
\newcommand\eff{\mathsf{eff}}
\newcommand\proj[2]{{#1\restriction#2}}

\newcommand\erase[1]{\Erase{#1}}

\newcommand\intnum{\mathbb Z}

\newcommand\sem[1]{\llbracket #1 \rrbracket}
\newcommand\ssem[1]{\sem{#1}_{\mathrm{S}}}

\newcommand\seq[2]{{#1} \vdash {#2}}

\newcommand\lang{\mathcal{L}}

\newcommand\new[2]{\mathsf{new}\,#1\,\mathsf{in}\,#2}
\newcommand\newc{\mathsf{ref}}
\newcommand\loc{\mathsf{Loc}}
\newcommand\na\alpha
\newcommand\can{\mathbb{C}}
\newcommand\splay[1]{\mathit{SP}_{#1}}
\newcommand\knowing[1]{\mathcal{K}_{#1}}
\newcommand\Sbinno[1]{\Scatinn}
\newcommand\binno[1]{\mathcal{B}_{#1}}
\newcommand\inno[1]{\mathcal{I}_{#1}}
\newcommand\ord[1]{\mathsf{ord}(#1)}
\newcommand\xar[2]{\ar@{-}@/_#1mm/[#2]}

\begin{document}

\title[Block structure vs scope extrusion]%
{Block structure vs scope extrusion: between innocence and omniscience\rsuper*}

\author[A.S.~Murawski]{Andrzej S.\ Murawski\rsuper a}
\address{{\lsuper a}University of Warwick}
\email{A.Murawski@warwick.ac.uk}
\thanks{{\lsuper a}Supported by the 
Engineering and Physical Sciences Research Council (EP/C539753/1)}

\author[N.~Tzevelekos]{Nikos Tzevelekos\rsuper b}
\address{{\lsuper b}Queen Mary University of London}
\email{nikos.tzevelekos@qmul.ac.uk}
\thanks{{\lsuper b}Supported by the 
Engineering and Physical Sciences Research Council (EP/F067607/1)
and the Royal Academy of Engineering}

\titlecomment{{\lsuper*}Extended abstract appeared in FOSSACS~\cite{MT10}.}


\begin{abstract}
We study the semantic meaning of block structure using game semantics.
To that end, we introduce the notion of block-innocent strategies and
characterise call-by-value computation with block-allocated storage
through soundness, finite definability and universality results.
This puts us in a good position to conduct a comparative study of
purely functional computation, computation with block storage as well as
that with dynamic memory allocation.
For example, we can show that dynamic variable allocation
can be replaced with block-allocated variables exactly when the
term involved (open or closed) is of base type and that block-allocated
storage can be replaced with purely functional computation
when types of order two are involved.
To illustrate the restrictive nature of block structure further,
we prove a decidability result for a finitary fragment of call-by-value
Idealized Algol for which it is known that allowing for dynamic memory
allocation leads to undecidability.
\end{abstract}

\maketitle


\section{Introduction}

Most programming languages manage memory
by employing a stack for local variables and heap storage
for data that are supposed to live beyond their initial context.
A prototypical example of the former mechanism is Reynolds's
Idealized Algol~\cite{Rey81}, in which local variables can only
be introduced inside blocks of ground type. Memory is then allocated
on entry to the block and deallocated on exit.
In contrast, languages such as ML permit variables to escape
from their current context under the guise of pointers or references.
In this case, after memory is allocated at the point of reference creation,
the variable must be allowed to persist indefinitely (in practice, garbage collection
or explicit deallocation can be used to put an end to its life).

In this paper we would like to compare the expressivity of the two paradigms.
As a simple example of heap-based memory allocation we consider
the language $\rml$, introduced by Abramsky and McCusker in~\cite{AM97b},
which is a fragment of ML featuring integer-valued references.  In op.~cit.\ the authors
also construct a fully abstract game model of $\rml$
based on strategies (referred to as \emph{knowing} strategies)
that allow the Proponent to base his decisions on the full history of play.
On the other hand, at around the same time Honda and Yoshida~\cite{HY97} showed that the purely functional
core of $\rml$, better known as call-by-value $\pcf$~\cite{Plo77}, corresponds to
\emph{innocent} strategies~\cite{HO00}, i.e.\ those that can only
rely on a restricted view of the play when deciding on the next move.
Since block-structured storage of Idealized Algol seems
less expressive than dynamic memory allocation of ML
and more expressive than $\pcf$, it is natural to ask about
its exact position in the spectrum of strategies
between innocence and omniscience.
Our first result is an answer to this question. 
We introduce the family of \emph{block-innocent}
strategies, situated strictly between innocent and knowing strategies, and exhibit a series of results
relating such strategies to a call-by-value variant $\iacbv$ of Idealized Algol.

Block-innocence captures the particular kind of uniformity exhibited by
strategies originating from block-structured programs, akin to innocence
yet strictly weaker. In fact, we shall define block-innocence through innocence in
a setting enriched with explicit store annotations added to standard moves.
For instance, in the play shown below\footnote{For the sake of clarity, we only include pointers pointing more
than one move ahead.},
if P follows a block-innocent strategy, P is free to use different moves as the fourth (1) and sixth (2) moves,
but the tenth one (0) and the twelfth one (0) have to be the same.
\[ \]
\[
\rnode{A}{q}\qwe \rnode{B}{q}\qwe \rnode{C}{q}\qwe \rnode{D}{1}\qwe
\rnode{E}{q}\qwe \rnode{F}{2}\qwe \rnode{G}{a}\qwe \rnode{H}{a}\qwe \rnode{I}{q}\qwe \rnode{J}{0}\qwe \rnode{K}{q}\qwe \rnode{L}{0}
\justf{E}{B}\justf{G}{B}\justf{H}{A}\justf{K}{H}
\]
The above play is present in the strategy representing the term
\[\begin{array}{rcl}
{f:(\comt\rarr\expt)\rarr (\comt\rarr\expt)} &\vdash & \left(\new{x}{(\letin{g=f(\lambda y^\comt. (x\aasg !x+1); !x)}{()}})\right); \\
&&\lambda y^\comt.0: \comt\rarr\int
\end{array}\]
and the necessity to play the same value ($0$ in this case) in the twelfth move once $0$ has been played in the tenth one
stems from the fact that variables in $\iacbv$ can only be allocated in blocks of ground type. For example, the block in which $x$ was allocated cannot extend over $\lambda y^\comt.0$.

Additionally, our framework can detect ``storage violations" resulting
from an attempt to access a variable from outside of its block.
For instance, no $\iacbv$-term will ever produce the following play.
\[ \]
\[
\rnode{A}{q}\qwe \rnode{B}{q}\qwe \rnode{C}{q}\qwe\rnode{D}{1}\qwe\rnode{E}{q}\qwe\rnode{F}{2}\qwe\rnode{G}{a}\qwe\rnode{H}{a}\qwe\rnode{I}{q}\qwe
\rnode{J}{q}\justf{E}{B}\justf{G}{B}\justf{H}{A}\justf{J}{G}
\]
The last move is the offending one: for the term given above, it would amount to trying to use $g$ after deallocation of the block for $x$.
Note, though, that the very similar play drawn below does originate from an $\iacbv$-term. 
\[ \]
\[
\rnode{A}{q}\qwe \rnode{B}{q}\qwe \rnode{C}{q}\qwe\rnode{D}{1}\qwe\rnode{E}{q}\qwe\rnode{F}{1}\qwe\rnode{G}{a}\qwe\rnode{H}{a}\qwe\rnode{I}{q}\qwe
\rnode{J}{q}\justf{E}{B}\justf{G}{B}\justf{H}{A}\justf{J}{G}
\]
Take, e.g.\
$\seq{f:(\comt\rarr\expt)\rarr (\comt\rarr\expt)}{\letin{g=f(\lambda y^\comt. 1)}{\lambda y^\comt. g()}}:\comt\to\int$.

The notion of block-innocence provides us with a systematic methodology
to address expressivity questions  related to block structure such as
\emph{``Does a given strategy originate from a stack-based memory discipline?''}
or \emph{``Can a given program using dynamic memory allocation be replaced
with an equivalent program featuring stack-based storage?''}.
To illustrate the approach we conduct a complete study of the relationship
between the three classes of strategies (innocent, block-innocent and knowing respectively) 
according to the underpinning type shape.
We find that knowingness implies block-innocence when terms of base types
(open or closed) are involved, that block-innocence implies innocence
exactly for types of order at most two, and that knowingness implies
innocence if the term is of base type and its free identifiers are of order $1$.
\nt{The fact that knowingness and innocence coincide at terms of base types implies, in particular, that $\rml$ and $\iacbv$ contexts have the same expressive power: two $\rml$-terms can be distinguished by an $\rml$-context if, and only if, they can be distinguished by an $\iacbv$-context.}

As a further confirmation of the restrictive nature of the stack discipline of $\iacbv$,
we prove that program equivalence is decidable for a finitary variant of $\iacbv$
which properly contains all second-order types as well as some third-order types
(interestingly, this type discipline covers the available higher-order types in PASCAL).
In contrast, the corresponding restriction of $\rml$ is known to be undecidable~\cite{Mur04b}.

\nada\\
\noindent
{\bf Related work.} The stack discipline has always been regarded as part of the essence
of Algol~\cite{Rey81}. \nt{The first languages introduced in that lineage, i.e.\ Algol 58 and 60, 
featured both call-by-name and call-by-value parameters. Call-by-name was abandoned in Algol 68, 
though, and was absent from subsequent designs, such as Pascal and C.
}

On the semantic front, finding models embodying stack-oriented storage management
has always been an important goal of research into Algol-like languages.
In this spirit, in the early 1980s, Reynolds~\cite{Rey81} and Oles~\cite{Ole85} devised a semantic
model of \nt{(call-by-name)} Algol-like languages using a category of functors from a category of store shapes
to the category of predomains. Perhaps surprisingly, in the 1990s, Pitts and Stark~\cite{PS93,Sta95}
managed to adapt the techniques to \nt{(call-by-value)} languages with dynamic allocation.
This would appear to create a common platform suitable for a comparative study
such as ours.  However, despite the valuable structural insights, the relative
imprecision of the functor category semantics (failure of definability and full abstraction)
makes it unlikely that the results obtained by us can be proved via this route.
The semantics of local effects has also been investigated from the category-theoretic
point of view in~\cite{Pow06}.

As for the game semantics literature, Ong's work~\cite{Ong02} based on strategies-with-state
is the work closest to ours. His paper defines a compositional framework that is proved sound
for the third-order fragment of \emph{call-by-name} Idealized Algol. Adapting the results to call-by-value
and all types is far from immediate, though. For a start, to handle higher-order types, we note that
the state of O-moves is no longer determined by their justifier and the preceding move.
Instead, the right state has to be computed globally using the whole history
of play. However, the obvious adaptation of this idea to call-by-value does not capture
the block structure of $\iacbv$. Quite the opposite: it seems to be more compatible with $\rml$.
Consequently,  further changes are needed to characterize $\iacbv$. Firstly,
to restore definability, the explicit stores have to become lists instead of sets.
Secondly, conditions controlling state changes must be tightened. In particular,
P must be forbidden from introducing fresh variables at any step and,  in a similar vein,
must be forced to drop some variables from his moves in certain circumstances.

\nt{Another related paper is~\cite{AM99a}, in which Abramsky and McCusker introduce a model of
Idealized Algol with passive (side-effect-free) expressions~\cite{AM99a}.
Their framework is based on a distinction between \emph{active} and \emph{passive} moves, which correspond
to active and passive types respectively. Legal plays must then satisfy a novel correctness condition, called  \emph{activity}, 
and strategies must be \emph{$a/p$-innocent}. In contrast, our setting does not feature any type support for 
discovering the presence of storage. Moreover, as the sequences discussed in the Introduction demonstrate,
in order to understand legality in our setting, it is sometimes necessary to scrutinise values used in plays:
changing $1,1$ to $1,2$ may entail loss of correctness!
This is different from the activity condition (and other conditions used in game semantics), where it suffices
to consider the kind of moves involved (question/answer) or pointer patterns. 
Consequently, in order to capture the desired
shape of plays and associated notion of innocence in our setting,  
we felt it was necessary to  introduce moves explicitly decorated with stores.}

Our paper is also related to the efforts of finding decidable fragments of (finitary) $\rml$ as far as contextual equivalence is concerned.
Despite several papers in the area~\cite{Ghi01, Mur04b, HMO11,CBHMO15}, no full classification based on type shapes has emerged yet, even though
the corresponding call-by-name case has been fully mapped
out~\cite{MOW05}. We show that, for certain types, moving from $\rml$
to $\iacbv$ (thus weakening storage capabilities) can help to regain
decidability.


\section{Syntax\label{sec:syntax}}

To set a common ground for our investigations, we introduce a higher-order 
programming language that features syntactic constructs for both
block and dynamic memory allocation. 

\begin{definition}[The language~$\lang$]
We define types as generated by the grammar below,
where $\beta$ ranges over the ground types $\comt$ and $\expt$.
\[
\theta\;\; ::= \quad \beta \quad|\quad \vart\quad | \quad \theta\rarr\theta
\]
The syntax of $\lang$ is given in Figure~\ref{fig:syntax}. 
\end{definition}

Note in particular the first two rules concerning variables and the rule for the $\sf mkvar$ constructor: the latter allows us to build ``bad variables'' in accordance with Idealized Algol. 
\begin{figure}[t]
\renewcommand\arraystretch{2.5}
\[\begin{array}{c}
\frac{\ds\seq{\Gamma,x:\vart}{M:\beta}}{\ds \seq{\Gamma}{\new{x}{M}:\beta}}\qquad\frac{}{\ds \seq{\Gamma}{\newc:\vart}} \qquad
\frac{}{\ds \seq{\Gamma}{\mathsf{()}:\comt}} \qquad
\frac{\ds i\in \intnum}{\ds \seq{\Gamma}{i:\expt}}\qquad
\frac{\ds (x:\theta)\in \Gamma}{\ds \seq{\Gamma}{x:\theta}}
\\
\frac{\ds\seq{\Gamma}{M_1:\expt}\qquad \seq{\Gamma}{M_2:\expt}}{\ds\seq{\Gamma}{M_1\oplus M_2:\expt}}\qquad\,
\frac{\ds \seq{\Gamma}{M:\expt}\quad\,
          \seq{\Gamma}{N_0:\theta}\quad\,  \seq{\Gamma}{N_1:\theta}}
{\ds\seq{\Gamma}{\cond{M}{N_1}{N_0}:\theta}}\\
\frac{\ds\seq{\Gamma}{M:\vart}}{\ds\seq{\Gamma}{!M:\expt}}\quad\,\,
\frac{\ds \seq{\Gamma}{M:\vart}\quad\seq{\Gamma}{N:\expt}}{\ds \seq{\Gamma}{M\aasg N:\comt}}\quad\,\,
\frac{\ds\seq{\Gamma}{M:\comt\rarr\expt}\quad\seq{\Gamma}{N:\expt\rarr\comt}}{\ds
\seq{\Gamma}{\badvar{M}{N}:\vart}}\\
\frac{\ds\seq{\Gamma}{M:\theta\rarr\theta'}\quad\seq{\Gamma}{N:\theta}}{\ds 
\seq{\Gamma}{MN:\theta'}}\quad\,\,
\frac{\ds\seq{\Gamma,x:\theta}{M:\theta'}}{\ds \seq{\Gamma}{\lambda x^\theta.M:\theta\rarr\theta'}}\quad\,\,
\frac{\ds\seq{\Gamma}{M:(\theta\rarr\theta')\rarr(\theta\rarr\theta')}}{\ds\seq{\Gamma}{\fix{M}:\theta\rarr\theta'}}\\
\end{array}\]
\caption{Syntax of $\lang$\label{fig:syntax}}
\makebox[\textwidth][l]{\hrulefill}
\end{figure}
{The order of a type is defined as follows: 
\[\begin{array}{rcl}
\ord{\beta}&=&0\\
\ord{\vart}&=&1\\
\ord{\theta_1\rarr\theta_2}&=& \max (\ord{\theta_1}+1,\ord{\theta_2}).
\end{array}\]}
For any $i\ge 0$, 
terms that are typable using exclusively judgments of the
form 
\[ \seq{x_1:\theta_1,\cdots,x_n:\theta_n}{M:\theta}\] 
where $\ord{\theta_j}<i$ ($1\le j\le n$) and $\ord{\theta}\le i$,
are said to form the $i$th-order fragment.

To spell out the operational semantics of $\lang$, we need to assume a countable set $\loc$
of \emph{locations}, which are added to the syntax as auxiliary constants of type $\vart$.
We shall write $\alpha$ to range over them.
The semantics then takes the form of judgments $\red{s,M}{s',V}$, where $s,s'$ are
finite partial functions from $\loc$ to integers, $M$ is a \nt{closed} term and $V$ is a value.
Terms of the following shapes are values:  $()$, integer constants, elements of $\loc$, 
$\lambda$-abstractions or terms of the form $\badvar{\lambda x^\comt.M}{\lambda y^\expt.N}$.

The operational semantics is given via the large-step rules in Figure~\ref{fig:opsem}. 
Most of them take the form 
\[
\frac{\ds\red{M_1}{V_1}\quad \red{M_2}{V_2}\quad\cdots\quad \red{M_n}{V_n}}{\ds \red{M}{V}}
\]
which is meant to abbreviate:
\[
\frac{\ds 
\red{s_1,M_1}{s_2,V_1}\quad
\red{s_2,M_2}{s_3,V_2}\quad \cdots\quad \red{s_n,M_n}{s_{n+1}, V_n}}{\ds \red{s_1,M_1}{s_{n+1},V}}
\]
This is a common semantic convention, introduced in the Definition of Standard ML~\cite{MTH90}.
In particular, it means that the ordering of the hypotheses is significant.
The penultimate rule in the figure encapsulates the state within the newly created block,
while the last one creates a reference to a new memory cell that can
be passed around without restrictions on its scope.
Note that $s'\setminus\alpha$
is the restriction of $s'$ to $\dom{s'}\setminus\makeset{\alpha}$.

\begin{figure}[t]
\renewcommand\arraystretch{2.5}
\[\begin{array}{c}
\frac{\ds \textrm{$V$ is a value}}{\ds \red{s,V}{s,V}}\qquad
\frac{\ds \red{M}{0}\quad\, \red{N_0}{V}}{\ds
\red{\cond{M}{N_1}{N_0}}{V}}
\qquad
\frac{\ds i\neq 0\quad\, \red{M}{i}\quad\, \red{N_1}{V}}{\ds
\red{\cond{M}{N_1}{N_0}}{V}}
\\
\frac{\ds \red{M_1}{i_1}\quad\, \red{M_2}{i_2}}{\ds
\red{M_1\oplus M_2}{i_1\oplus i_2}}\qquad
\frac{\ds \red{M}{\lambda x.M'}\quad\,\red{N}{V'}\quad\red{M'[V'/x]}{V}}{\ds
\red{MN}{V}}\\
\frac{\ds \red{s,M}{s',\alpha}\quad\, s'(\alpha)=i}{\ds \red{s,!M}{s',i}}
\qquad
\frac{\ds \red{s,M}{s',\alpha}\quad\red{s',N}{s'',i}} {
\ds\red{s,M\aasg N}{s''(\alpha\mapsto i),()}}\\
\frac{\ds \red{M}{\badvar{V_1}{V_2}}\quad\red{V_1()}{i}}{\ds \red{!M}{i}}
\qquad
\frac{\ds \red{M}{\badvar{V_1}{V_2}}\quad\red{N}{i}\quad \red{V_2\, i}{()}} {
\ds\red{M\aasg N}{()}}\\
\frac{\ds \red{M}{V_1}\quad \red{N}{V_2}}{\ds \red{\badvar{M}{N}}{\badvar{V_1}{V_2}}}
\qquad \frac{\ds\red{M}{V}}{\ds \red{\fix{M}}{\lambda x^\theta. (V(\fix{V}))x} }\\
\frac{\ds\red{s\cup(\alpha\mapsto 0),M[\alpha/x]}{s',V}}{\ds\red{s,\new{x}{M}}{s'\setminus \na, V}}\,\,
\alpha\not\in\dom{s}
\qquad\,\,
\frac{}{\ds\red{s,\newc}{s\cup(\na\mapsto 0), \na}}\quad\alpha\not\in\dom{s}
\end{array}\]
\caption{Operational semantics of $\lang$\label{fig:opsem}}
\makebox[\textwidth][l]{\hrulefill}
\end{figure}

\cutout{
Here we only reproduce the two evaluation rules related to variable creation.
\[
\frac{\ds\red{s\cup(\alpha\mapsto 0),M[\alpha/x]}{s',V}}{\ds\red{s,\new{x}{M}}{s'\setminus \na, V}}
{\ \scriptstyle\na\notin\dom{s}}
\qquad
\frac{}{\ds\red{s,\newc}{s\cup(\na\mapsto 0), \na}}
{\ \scriptstyle\na\notin\dom{s}}
\]
}

Given a closed term $\seq{}{M:\comt}$, we write $M\Downarrow$ 
if there exists $s$ such that $\red{\emptyset,M}{s,()}$.
We shall call two programs equivalent if they behave identically 
in every context. This is captured by the following definition, 
parameterised by the kind of contexts that are considered,
to allow for testing of terms with contexts originating from a designated 
subset of the language.

\begin{definition}
Suppose $\langp$ is a subset of $\lang$.
We say that the terms-in-context $\seq{\Gamma}{M_1, M_2:\theta}$ 
are $\langp$-equivalent (written $\seq{\Gamma}{M_1\eqmod{\langp} M_2}$)
if, for any $\langp$-context ${C}$ such that $\seq{}{{C}[M_1],{C}[M_2]:\comt}$,
${C}[M_1]\Downarrow$ if and only if  ${C}[M_2]\Downarrow$.
\end{definition}

We shall study three sublanguages of $\lang$, called
$\pcfplus$, $\iacbv$ and $\rml$ respectively. The latter two
have appeared in the literature as paradigmatic examples of programming
languages with stack discipline and dynamic memory allocation respectively.
\begin{itemize}
\item 
$\pcfplus$ is a purely functional language obtained from $\lang$ by removing
$\new{x}{M}$ and $\newc$. It extends the language $\mathsf{PCF}$~\cite{Plo77}
with primitives for variable access, but not for memory allocation.
\item 
$\iacbv$ is $\lang$ without the $\newc$ constant.
It can be viewed as  a call-by-value variant of Idealized Algol~\cite{Rey81}.
Only block-allocated storage is available in $\iacbv$.
\item 
$\rml$ is $\lang$ save the construct $\new{x}{M}$. It is exactly the language
introduced in~\cite{AM97b} as a prototypical language for ML-like integer references.\footnote{\nt{In other words, $\rml$ is Reduced ML~\cite{Sta95} with the addition of the $\mkvar$ construct.}}
\end{itemize}
{We shall often use $\letin{x=M}{N}$  as shorthand for $(\lambda x.N)M$. 
Moreover, $\letin{x=M}{N}$, where $x$ does not occur in $N$, will be abbreviated to $M;N$.}
Note also that $\new{x}{M}$ is equivalent to $\letin{x=\newc}{M}$.
\nt{\begin{example}
The term $\seq{}{\letin{v=\newc}{\lambda x^\comt. (\cond{!v}{\Omega}{v\aasg !v+1}}):\comt\rarr\comt}$ is an example of an $\rml$-term 
that is not \am{$\rml$-equivalent} to any term from $\iacbv$.
On the other hand, 
\[
\seq{\!}{\lambda f^{(\comt\rarr\comt)\rarr\comt}.\new{v}{f(\lambda y^\comt.\cond{!v}{\Omega}{v\aasg !v+1})}:((\comt\!\rarr\!\comt)\!\rarr\!\comt)\!\rarr\!\comt}
\]
is an $\iacbv$-term that has no \am{$\rml$-equivalent} in $\pcfplus$. All of  the inequivalence claims will follow immediately from our results.
\end{example}}

\nt{
\begin{lemma}
Given any base type \am{$\lang$-term}\, $\seq{\Gamma,x:\vart}{M:\beta}$, we have $\seq{\Gamma}{\new{x}{M}\eqmod{\am{\lang}}}$ $\letin{x=\newc}{M}:\beta$.
\end{lemma}
\begin{proof}
The proof is based on the following two claims:
\begin{itemize}
\item if $s,M\Downarrow s',V$ and $\alpha\in\dom(s)$ does not appear in $M$, then $s\setminus\alpha,M\Downarrow s'\setminus\alpha,V$;
\item for any closed context $C$, value $V$ and $s,s'$, if
$s,C[\letin{x=\newc}{M}]\Downarrow s',V$ then there is a set of locations $S\subseteq\dom(s')$ such that 
$s,C[\new{x}{M}]\Downarrow s'\setminus S,V$ and $S$ contains no locations from $s$ or $V$;
\end{itemize}
which are proven by straightforward induction.
\cutout{ following fact without proof: if $s,M\Downarrow s',V$ and $\alpha\in\dom(s)$ does not appear in $M$, then $s\setminus\alpha,M\Downarrow s'\setminus\alpha,V$.
Let us write $M_1$ for $\letin{x=\newc}{M}$, and $M_2$ for $\new{x}{M}$.
We now show that, for any closed context $C$, value $V$ and $s,s'$, if
$s,C[M_1]\Downarrow s',V$ then there is a set of locations $S\subseteq\dom(s')$ such that 
$s,C[M_2]\Downarrow s'\setminus S,V$ and $S$ contains no locations from $s$ or $V$. We argue by induction on $C$. In the base case, $C[M_1]\equiv M_1$ and the claim clearly holds as $M$ is of base type. For the inductive case, consider $C\equiv C_1C_2$ and suppose $s,C_1[M_1]\Downarrow s_1,\lambda x.N$ and $s_1,C_2[M_1]\Downarrow s_2,V'$ and $s_2,N[V'/x]\Downarrow s',V$. Then, by IH, we have that $s,C_1[M_1]\Downarrow s_1\setminus S_1,\lambda x.N$ for some set of names $S_1$ fresh for $s,N$. By the above fact, and since $s$ is compatible with $C_2[M_1]$, we have 
$s_1\setminus S_1,C_2[M_1]\Downarrow s_2\setminus S_1,V'$ and, hence by the IH, 
$s_1\setminus S_1,C_2[M_2]\Downarrow s_2\setminus S,V'$ for some $S$ extending $S_1$ with names fresh for $s_1,V'$. We moreover have $s_2\setminus S,N[V'/x]\Downarrow s'\setminus S,V$, which yields the claim. The other inductive cases are treated similarly.

\nt{Mention the two facts in the proof, say proven by easy induction.}}
\end{proof}}

Hence, 
$\rml$ and $\lang$ merely differ on a syntactic level in that $\lang$
contains ``syntactic sugar" for blocks. In the opposite direction, 
our results will show that $\newc$ cannot in general be replaced
with an equivalent term that uses $\new{x}{M}$. Indeed, our paper provides 
a general methodology for identifying and studying scenarios 
in which this expressivity gap is real.


\section{Game semantics}
%

We next introduce the game models used throughout the paper,
which are based on the Honda-Yoshida approach to modelling call-by-value computation~\cite{HY97}.
%
\begin{definition}
An \boldemph{arena} $A=(M_A,I_A,\vdash_A,\lambda_A)$ is given by
\begin{itemize}
  \item a set $M_A$ of moves, and a subset $I_A\subseteq M_A$ of \emph{initial} moves,
  \item a justification relation ${\vdash_A}\subseteq M_A\times (M_A\setminus{I}_A)$, and
  \item a labelling function $\lambda_A:M_A\arr \{O,P\}\times\{Q,A\}$\,
\end{itemize}
such that $\lambda_A(I_A)\subseteq\{\mathit{PA}\}$ and, whenever $m\vdash_A m'$,
we have $(\pi_1\lambda_A)(m)\neq(\pi_2\lambda_A)(m')$ and $(\pi_2\lambda_A)(m')=A\implies (\pi_2\lambda_A)(m)=Q$. 
\end{definition}
The role of $\lambda_A$ is to label moves as \emph{Opponent} or \emph{Proponent} moves and
as \emph{Questions} or \emph{Answers}. We typically write them as $m,n,\dots$, or $o,p,q,a,q_P,q_O,\dots$ when we want to be specific about their kind.
\nt{Note that we abbreviate elements of the codomain of $\lambda_A$, e.g.\ $(P,A)$ above is written as $PA$.}

The simplest arena is $0=(\emptyset,\emptyset,\emptyset,\emptyset)$.
Other ``flat'' arenas are $1$ and $\Z$, defined by:
\[
M_{1}=I_{1}=\{*\}\,,\quad 
M_{\Z}=I_{\Z}=\Z\,.
\]
Below we recall two standard constructions on arenas, 
where $\bar{I}_A$ stands for $M_A\setminus I_A$,
the $OP$-complement of $\lambda_A$ is written as $\bar{\lambda}_A$, and $i_A, i_B$ range
over initial moves in the respective arenas.
\[
\begin{aligned}
M_{A\impl B}   &= I_{A\impl B}\uplus I_A \uplus \BI{A} \uplus M_B\\
I_{A\impl B}   &= \{*\} \\
\lambda_{A\impl B} &= [(*,{PA}),(i_A,{OQ}),\proj{\bar{\lambda}_A}{\BI{A}},\lambda_B] \\
\vdash_{A\impl B} &= \{(*,i_A),(i_A,i_B)\} \cup{}\vdash_A{}\cup{}\vdash_B
\end{aligned}
\]
\bigskip
\[
\begin{aligned}
M_{A\otimes B}   &= I_{A\otimes B}\uplus \BI{A} \uplus \BI{B}\\
I_{A\otimes B}   &= I_A \times I_B \\
\lambda_{A\otimes B} &= [((i_A,i_B),{PA}),\, \proj{\lambda_A}{\BI{A}},\,\proj{\lambda_B}{\BI{B}}] \\
\vdash_{A\otimes B} &= \{((i_A,i_B),m)\,\,|\,\, i_A\vdash_A m\,\lor\, i_B\vdash_B m\}\\
&\quad\cup(\proj{\vdash_A}{\BI{A}}^2)\cup(\proj{\vdash_B}{\BI{B}}^2)\\
\end{aligned}
\]
Types of $\mathcal{L}$ can now be interpreted with arenas in the following way.
\begin{align*}
\sem{\comt} &= 1\\
\sem{\expt} &=\Z\\
\sem{\vart} &=(1\Arr\Z)\otimes(\Z\Arr1)\\
\sem{\theta_1\arr\theta_2} &=\sem{\theta_1}\impl \sem{\theta_2}
\end{align*}
Note that the type $\vart$ is translated as a product arena the components of which represent its read and write methods.

Although arenas model types, the actual games will be played in \boldemph{prearenas},
which are defined in the same way as arenas with the exception that initial moves must
be O-questions. Given arenas $A$ and  $B$, we can construct the prearena $A\arr B$ by setting:
\[
\begin{aligned}
M_{A\arr B} &= M_A \uplus M_B\\
I_{A\arr B}&= I_A\\
\lambda_{A\arr B}&= [(i_A, OQ)\,\cup\, (\proj{\bar{\lambda}_A}{\BI{A}}) \;,\; \lambda_B]\\
\vdash_{A\arr B}&= \{(i_A,i_B)\}\,\cup\,\vdash_A\,\cup\,\vdash_B.
\end{aligned}
\]
For $\Gamma=\makeset{x_1:\theta_1,\cdots,x_n:\theta_n}$, typing judgments $\seq{\Gamma}{\theta}$ will eventually be interpreted by strategies for the prearena
$\sem{\theta_1}\otimes\cdots\otimes\sem{\theta_n}\arr \sem{\theta}$ (if $n=0$ we take the left-hand side to be $1$),
which we shall denote by $\sem{\seq{\Gamma}{\theta}}$ or $\sem{\seq{\theta_1,\cdots,\theta_n}{\theta}}$.

A \boldemph{justified sequence} in a prearena $A$ is a finite sequence $s$
of moves of $A$ satisfying the following condition: the first move must be initial,
but all other moves $m$ must be equipped with a pointer  to an earlier
occurrence of a move $m'$ such that $m'\vdash_A m$.
\nt{We then say that $m'$ \emph{justifies} $m$. If $m$ is an answer, we may also say that $m$ \emph{answers} $m'$.
If a question remains unanswered in $s$, it is \emph{open}; and the rightmost open question in $s$ is its \emph{pending question}.}

Given a justified sequence $s$, we define its \emph{O-view }$\oview{s}$ and its \emph{P-view} $\pview{s}$ inductively as follows.
\begin{itemize}
\item $\oview{\epsilon}=\epsilon$\,, \,
$\oview{s\, o}= \oview{s}\, o$\,, \,
$\oview{s\, \rnode{A}{o}\cdots \rnode{B}{p}}\justf{B}{A} =  \oview{s}\,o\,p$\,;
\item
$\pview{\epsilon} = \epsilon$\,,
$\pview{s\, p} = \pview{s}\, p$\,,
$\pview{s\, \rnode{A}{p}\cdots \rnode{B}{o}}\justf{B}{A} =   \pview{s}\,p\,o$\,.
\end{itemize}
Above, \nt{recall that} $o$ ranges over O-moves (i.e.\ moves $m$ such that $(\pi_1\lambda_A)(m)=O$), and $p$ ranges over P-moves.

\begin{definition}
A \boldemph{play} in a prearena $A$ is a justified sequence $s$ satisfying
the following conditions.
\begin{itemize}
\item If $s=\cdots m\,n\cdots$ then $\lambda_A^{OP}(m)=\bar{\lambda}_A^{OP}(n)$.\ (\emph{Alternation})
\item If $s=s_1\, \rnode{A}{q}\,s_2\,\rnode{B}{a}\cdots\justf{B}{A}$ then $q$ is the pending question in $s_1\,q\,s_2$.\
(\emph{Well-Bracketing})
\item
If $s=s_1\,\rnode{A}{o}\,s_2\rnode{B}{p}\cdots\justf{B}{A}$ then $o$ appears in $\pv{s_1\,o\,s_2}$\,;\\ 
if  $s=s_1\, \rnode{A}{p}\,s_2\rnode{B}{o}\cdots\justf{B}{A}$ then $p$ appears in $\ov{s_1\,p\,s_2}$\,.
(\emph{Visibility})
\end{itemize}
We write $\plays{A}$ to denote the set of plays in $A$.
\end{definition}


We are going to model terms-in-context $\Gamma\vdash M:\theta$ as sets of plays in $\sem{\Gamma\vdash\theta}$ subject to specific conditions.

\begin{definition}
A \boldemph{(knowing) strategy} $\sigma$ on a prearena $A$ is a {non-empty} prefix-closed set of plays from $A$ satisfying the first two conditions below. A strategy is \boldemph{innocent} if, in addition, the third condition holds.
\begin{itemize}
 \item If even-length $s\in\sigma$ and $sm\in\plays{A}$ then $sm\in\sigma$.\
(\emph{O-Closure})
 \item If even-length $sm_1,sm_2\in\sigma$ then $m_1=m_2$.\ (\emph{Determinacy})
 \item If $s_1m,s_2\in\sigma$ with odd-length $s_1,s_2$ and $\pv{s_1}=\pv{s_2}$ then $s_2m\in\sigma$.\ (\emph{Innocence})
\end{itemize}
We write $\sigma:A$ to denote that $\sigma$ is a strategy on $A$.
\end{definition}

Note, in particular, that every strategy $\sigma:A$ contains the empty sequence $\epsilon$ as well as the elements of $I_A$, the latter being the 1-move plays in $A$.
\nt{Moreover, in the last condition above, the move $m$ in $s_2m$ points at the same move it points inside $s_1m$: by visibility and the fact that $s_1$ and $s_2$ have the same view, this is always possible.}

{%
In previous work it has been shown that knowing strategies yield a fully abstract semantics for $\rml$ in the following sense\footnote{Perhaps it is worth noting that the presentation of the model in~\cite{AM97b} is in a different setting
  (we follow Honda-Yoshida call-by-value games, while~\cite{AM97b} applies the \emph{family construction} on call-by-name games) which, nonetheless, is equivalent to the one presented above.}.
  \nt{
  \begin{theorem}[\cite{AM97b}]
  Two $\rml$-terms are $\rml$-equivalent
  if and only if their interpretations contain the same complete plays\footnote{A play is called \emph{complete} if each question in that play has been answered.}. 
  \end{theorem}
Moreover, as an immediate consequence of the full abstraction result of~\cite{HY97}, we have that innocent strategies (quotiented by the intrinsic preorder)
yield full abstraction for $\pcfplus$.}
What remains open is the model for the intermediate language, $\iacbv$, which requires one to identify a family of strategies between the innocent and knowing ones.
This is the problem examined in the next section. We address it in two steps.
\begin{itemize}
\item First we introduce a category of strategies that are equipped with explicit stores for registering private variables. 
We show that this category, of so-called innocent \emph{S-strategies}, indeed models block allocation: terms of $\iacbv$ translate into innocent S-strategies (Proposition~\ref{prop:model}) and, moreover, in a complete manner (Propositions~\ref{prop:defin} and~\ref{prop:univ}). 
\item The strategies capturing $\iacbv$, called \emph{block-innocent} strategies, are then defined by deleting stores from innocent S-strategies.
\end{itemize}
}

\section{Games with stores and the model of \texorpdfstring{$\iacbv$}{IAcbv}}

We shall now extend the framework to allow moves to be decorated with
stores that contain name-integer pairs. 
This extension will be necessary for capturing block-allocated storage.
The names should be viewed as semantic analogues of locations.
The stores will be used for carrying the values of private, block-allocated variables.

\subsection{Names and stores in games}

When employing such moves-with-store, we are not interested in what exactly the names are,
but we would like to know how they relate to names that have already been in play. Hence,
the objects of study are rather the induced equivalence classes with respect to name-invariance, and all ensuing constructions and reasoning need to be compatible with it. This overhead can be dealt with robustly using the language of nominal set theory~\cite{GP02}.
Let us fix a countably infinite set $\A$, the set of \emph{names}, the elements of which we shall denote by $\na,\nb$ and variants. Consider the group $\permg(\A)$ of finite permutations of $\A$. 

\begin{definition}
A \boldemph{strong nominal set}~\cite{GP02,Tze09}
is a set equipped with a group action\footnote{A group action of $\permg(\A)$ on $X$ is a function\, $\_\cdot\_\ :\permg(\A)\times X\rightarrow X$\, such that, for all $x\in X$ and $\pi,\pi'\in\permg(\A)$, $\pi\cdot(\pi'\cdot x)=(\pi\circ\pi')\cdot x$  and $\mathsf{id}\cdot x = x$, where $\mathsf{id}$ is the identity permutation.}
of $\permg(\A)$
such that each of its elements has \emph{finite strong support}. That is to say, for any $x\in X$,
there exists a finite set $\nu(x)\subseteq\A$, called \boldemph{the support of $x$}, such that, for all permutations $\pi$, $(\forall{\na\in\nu(x)}.\,\pi(\na)=\na)\iff \pi\cdot x=x$.
\end{definition}

Intuitively, $\nu(x)$ is the set of names ``involved'' in $x$.
%
For example, the set $\A^\#$ of finite lists of distinct atoms with permutations acting elementwise is a strong nominal set.
If $X$ and $Y$ are strong nominal sets, then so is their cartesian product $X\times Y$ (with permutations acting componentwise) and their disjoint union $X\uplus Y$.
Name-invariance in a strong nominal set $X$ is represented by the relation: $x\sim x'$ if there exists $\pi$ such that $x=\pi\actn x'$.

We define a strong nominal set of \boldemph{stores}, the elements of which are finite sequences of
name-integer pairs. Formally,
\[ 
\Sigma,\Tau\; ::=\;\; \epsilon\ |\ (\na,i)::\Sigma 
\]
where $i\in\Z$ and $\na\in\A\setminus\nu(\Sigma)$. 
We view stores as finite functions from names to
integers, though their domains are lists rather than sets. Thus, we
define the \boldemph{domain} of a store to be the \emph{list} of
names obtained by applying the first projection to all of its elements. In particular, $\nu(\dom(\Sigma))=\nu(\Sigma)$. If $\na\in\nu(\Sigma)$ then we write $\Sigma(\na)$ for the unique $i$ such that $(\na,i)$ is an element of $\Sigma$. 
For stores $\Sigma,\Tau$ we write:
\begin{align*}
   \Sigma\substore\Tau\;\ &\text{for } \dom(\Sigma)\subseq\dom(\Tau),\\
   \Sigma\substore_X\Tau\;\ &\text{for } \dom(\Sigma)\subseq_X\dom(\Tau),
\end{align*}
where $X\in\{p,s\}$, and $\subseq,\prefix,\suffix$ denote the subsequence, prefix and suffix relations
respectively.
Note that $\Sigma\substore_X\Tau\substore_X\Sigma$ implies $\dom(\Sigma)=\dom(\Tau)$ but not $\Sigma=\Tau$.
Finally, let us write $\Sigma\remv\Tau$ for $\Sigma$ restricted to $\nu(\Sigma)\setminus\nu(\Tau)$.

An \boldemph{S-move} (or \emph{move-with-store}) in a prearena $A$ is a pair consisting of a move and a store. We typically write S-moves as $m^\Sigma,n^\Tau,o^\Sigma,p^\Tau,q^\Sigma,a^\Tau$.
The first projection function is viewed as \emph{store erasure} and denoted by $\Erase{\_\!\_\,}$.
Note that moves contain no names and therefore, for any $m^\Sigma$,\
$\nu(m^\Sigma)=\nu(\Sigma)=\nu(\dom(\Sigma))$\,.
%
A \boldemph{justified S-sequence} in $A$ is a sequence of S-moves equipped with justifiers, so that its erasure is a justified sequence. The notions of \emph{O-view} and \emph{P-view} are extended to S-sequences in the obvious manner.
We say that a name $\na$ \boldemph{is closed} in $s$
if there are no open questions in $s$ containing $\na$.

\begin{definition}\label{def:splay}
A justified S-sequence $s$ in a prearena $A$ is called an \boldemph{S-play} if it satisfies the following conditions, for all $\na\in\A$.
\begin{itemize}
  \item If $s=m^\Sigma\cdots$ then $\Sigma=\ee$.\ ({\em Init})
  \item If $s=\cdots \rnode{A}{o^{\Sigma}}\cdots \rnode{B}{p}^\Tau\cdots\justf[,angleB=20]{B}{A}$ then $\Sigma\Substore\Tau$. If $\lambda_A(p)=PA$ then $\Tau\Substore\Sigma$ too.\ ({\em Just-P})
 \item If $s=\cdots \rnode{A}{p^{\Sigma}} \cdots \rnode{B}{o}^\Tau\cdots\justf[,angleB=20]{B}{A}$ \ then $\Sigma\Substore\Tau\Substore\Sigma$.\ ({\em Just-O})
  \item If $s=s_1\, o^\Sigma q_P^\Tau\cdots$ then $\Sigma\remv\Tau\substorE\Sigma$ and $\Sigma\remv(\Sigma\remv\Tau)\Substore\Tau$ and 
  \begin{enumerate}[label=\({\alph*}]
    \item if $\na\in\nu(\Tau\remv\Sigma)$ then $\na\notin\nu(s_1o^\Sigma)$,
    \item if $\na\in\nu(\Sigma\remv\Tau)$ then $\na$ is closed in $s_1o^\Sigma$.
  \end{enumerate}({\em Prev-PQ})
  \item If $s=\cdots p^\Sigma s'o^\Tau\cdots$ \ and $\na\in(\nu(\Tau)\cap\nu(\Sigma))\setminus\nu(s')$ then $\Tau(\na)=\Sigma(\na)$.\ ({\em Val-O})
\end{itemize}
We write $\Splays{A}$ for the set of S-plays in $A$.
\end{definition}

Let us remark that, as stores have strong support, the set of S-plays $\Splays{A}$ is a strong nominal set.
The conditions we impose on S-plays reflect the restrictions pertaining to block-allocation of variables. In particular, given a move $m$, all block-allocated variables present at $m$ are carried over to every move $n$ justified by $m$. In addition, only P is allowed to allocate/deallocate such variables, or change their values.

\paragraph{\bf\em Just-P}
All variables allocated at $o$ survive in $p$. If $p$ is an answer then, in fact, the subsequence from $o$ to $p$ represents a whole {sub-}block, with $p$ closing the {sub-}block. Thus, $o$ and $p$ must have the same private variables.

\paragraph{\bf\em Just-O}
Each O-move inherits its private variables from its justifier move. Put otherwise, a block does not extend beyond the current P-view and we only store variables that are created by and are private to P (so $T$ cannot be larger that $\Sigma$).

\paragraph{\bf\em Prev-PQ}
P-questions can open or close private variables and thus alter the domain of the store. This process must obey the nesting of variables, as reflected in the order of names in stores. Therefore, 
variables are closed by removing their corresponding names from the right end of the store: $\Sigma\setminus T\substorE\Sigma$.
On the other hand, variables/names that survive comprise the left end of the new store: $\Sigma\remv(\Sigma\remv\Tau)\Substore\Tau$. 

In addition, any names that are added in the store must be fresh for the whole sequence (they represent fresh private variables). A final condition disallows variables to be closed if their block still contains open questions.

\paragraph{\bf\em Val-O}
Since variables are private to P, it is not possible for O to change their value: at each O-move $o^T$, the value of each $\na\in\dom(T)$ is the same as that of the last P-move $p^\Sigma$ such that $\na\in\dom(\Sigma)$.
\medskip

%
%
The above is a minimal collection of rules that we need to impose for block allocation. From them we can extract further properties for S-plays, whose proofs are delegated to Appendix~\ref{apx:plays}.
\begin{lemma}
{The following properties hold for S-plays $s$.}
\begin{itemize}
\item If $s=\cdots m^\Sigma a_{P}^\Tau\cdots$ then $\Sigma\remv\Tau\substorE\Sigma$ and $\Sigma\remv(\Sigma\remv\Tau)\Substore\Tau$ and
\begin{itemize}
    \item[\quad(a)] if $\na\in\nu(\Tau)$ then $\na\in\nu(\Sigma)$,
    \item[\quad(b)] if $\na\in\nu(\Sigma\remv\Tau)$ then $\na$ is closed in $s_{<a_P^T}$.
\end{itemize}({\em Prev-PA})
\item \nt{%
For any $\na$, we have 
$\pv{s}=s_1s_2s_3$, where
\begin{itemize}
  \item $\na\notin\nu(s_1)\cup\nu(s_3)$ and $\forall m^\Sigma\in s_2.\ \na\in\nu(\Sigma)$,
  \item if $s_2\neq\epsilon$ then its first element is the move introducing $\na$ in $s$.
\end{itemize}
({\em Block form})}
\item If $s=s_1o^\Sigma p^\Tau s_2$ with $\na\in \nu(\Sigma)\setminus\nu(\Tau)$ then $\na\notin\nu(s_2)$.\ ({\em Close})\qed
\end{itemize}
\end{lemma}

We now move on to strategies for block allocation.

\begin{definition}
An \boldemph{S-strategy} $\sigma$ on an arena $A$ is a non-empty prefix-closed set of S-plays from $A$ satisfying the first three of the following conditions.
An S-strategy is \boldemph{innocent} if it also satisfies the last condition.
\begin{itemize}
 \item If $s'\sim s\in\sigma$ then $s'\in\sigma$.\ ({\em Nominal Closure})
 \item If even-length $s\in\sigma$ and $sm^\Sigma\in\Splays{A}$ then $sm^\Sigma\in\sigma$.\ ({\em O-Closure})
 \item If even-length $sm_1^{\Sigma_1},sm_2^{\Sigma_2}\in\sigma$ then $sm_1^{\Sigma_1}\sim sm_2^{\Sigma_2}$.\ ({\em Determinacy})
\item If $s_1m^{\Sigma_1},s_2\in\sigma$ with $s_1,s_2$ odd-length and $\pv{s_1}=\pv{s_2}$ then there exists $s_2m^{\Sigma_2}\in\sigma$ with $\pv{s_1m^{\Sigma_1}}\sim \pv{s_2m^{\Sigma_2}}$.\ ({\em Innocence})
\end{itemize}
We  write $\sigma:A$ to denote that $\sigma$ is an $S$-strategy on $A$.
\end{definition}

Observe how  S-strategies
are defined in the same manner as ordinary strategies but
follow some additional conditions due to their involving of stores and names.

\begin{example}\label{ex:cell}
  For any base type $\beta$, consider the prearena $\sem{\vart\rarr\beta}\rarr\sem{\beta}$ given below, where we have indexed moves and type constructors to indicate provenance.
We use $\mread$ and $\mwrite{i}$ ($i\in\Z$) to refer to the question-moves
from $\sem{\var}$, and $i$ ($i\in\Z$) and $\mok$ for the corresponding answers.
\[
\begin{aligned}
&\sem{\vart\rarr\beta}\rarr\sem{\beta}\\
&=\ (\sem{\vart}\Rarr\sem{\beta})\to\sem{\beta}\\
&=\ (\sem{\vart}\Rarr_1\sem{\beta})\to_0\sem{\beta}
\\\\\\
\end{aligned}
\qquad\qquad
\begin{aligned}
\xymatrix@R=.5mm@C=2mm{
&q_0\ar@{-}[dd]\ar@{-}[dr] \\
&& a_0 \\
&q_1\ar@{-}[dd]\ar@{-}[ddl]\ar@{-}[dr] \\
&& a_1 \\
\mread\ar@{-}[dd]  & \mwrite{i}\ar@{-}[dd]
\\
~\\
i & \mok
}
\end{aligned}
\]
Let us define the S-strategy
$\mathsf{cell}_\beta: \sem{\vart\rarr\beta}\rarr\sem{\beta}$ as the
S-strategy containing all \nt{even-length prefixes of S-plays of the following form (where we use the Kleene star for move repetition).
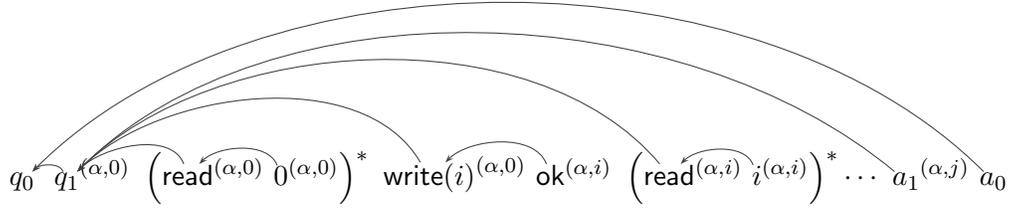
\[
\begin{aligned}
\\\\\\[2.5mm]
&\rnode{A}{q_0}\qwe \rnode{B}{q_1}^{(\na,0)}
\qweq
\left(\rnode{CA}{\mread}^{(\na,0)}\qweq\rnode{DA}{0}^{(\na,0)}\right)^*\justf{CA}{B}\justf{DA}{CA}
\qweq
\rnode{CB}{\mwrite{i}}^{(\na,0)}\qweq\rnode{DB}{\mok}^{(\na,i)}\justf[,angleA=120]{CB}{B}\justf{DB}{CB}
\qweq
\left(\rnode{CC}{\mread}^{(\na,i)}\qweq\rnode{DC}{i}^{(\na,i)}\right)^*\justf[,angleA=130]{CC}{B}\justf{DC}{CC}
\cdots\qweq
\rnode{C}{a_1}^{(\na,j)}\qweq\rnode{D}{a_0}\justf{B}{A}\justf[,angleA=140]{C}{B}\justf[,angleA=135]{D}{A}
\end{aligned}
\]}%
Note that, although the $\sf cell$ strategy is typically non-innocent, it is innocent in our framework, where private variables are explicit in game moves (in their stores).
\end{example}

\begin{example}
Let us consider the prearena $\sem{\comt\rarr\expt}\rarr\comt$, depicted as on the left below.
Had we used sets instead of lists for representing stores, the following ``S-strategy" (right below),
which represents incorrect overlap of scopes ($\alpha$ and $\beta$ are in scope of one another,
but at the same time have different scopes), would have been valid (and innocent).
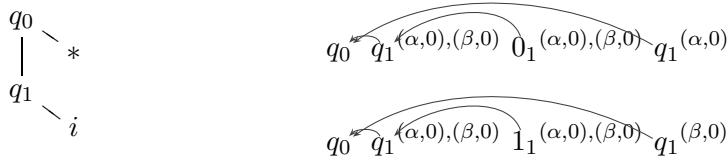
\[
\begin{aligned}
\xymatrix@R=.5mm@C=2mm{
q_0\ar@{-}[dd]\ar@{-}[dr] \\
& \ast \\
q_1\ar@{-}[dr] \\
& i
}
\end{aligned}
\qquad\qquad\qquad\qquad
\begin{aligned}
\\
&
\rnode{A}{q_0}\qwe\rnode{B}{q_1}^{(\na,0),(\nb,0)}\qweq\rnode{C}{0_1}^{(\na,0),(\nb,0)}\qweq\rnode{D}{q_1}^{(\na,0)}
\justf{B}{A}\justf{C}{B}\justf[,angleA=150,angleB=35]{D}{A}
\\\\
&
\rnode{A}{q_0}\qwe\rnode{B}{q_1}^{(\na,0),(\nb,0)}\qweq\rnode{C}{1_1}^{(\na,0),(\nb,0)}\qweq\rnode{D}{q_1}^{(\nb,0)}
\justf{B}{A}\justf{C}{B}\justf[,angleA=150,angleB=35]{D}{A}
\end{aligned}
\]
{However, the above is not a valid S-strategy in our language setting. Intuitively, it would correspond to a term that determines the scope of its variables on the fly, depending on the value (0 or 1) received after memory allocation.}
\end{example}

\subsection{Composing S-plays}

%
We next define composition of S-plays, following the approach of~\cite{HY97,Tze09}. Let us introduce some notation on stores. For a sequence of S-moves $s$ and stores $\Sigma,\Tau$, we write $\Sigma[s],\Sigma[\Tau]$ and $\Sigma+\Tau$ for the stores defined by: $\ee[s]=\ee[\Tau]\defn\ee$, $\ee+\Tau\defn\Tau$ and
\begin{align*}
((\na,i)::\Sigma)[s] &\defn 
\begin{cases}(\na,\Tau(\na))::(\Sigma[s]) &\text{if }s=s_1m^\Tau s_2\land\na\in\nu(\Tau)\setminus\nu(s_2)\\ (\na,i)::(\Sigma[s]) &\text{otherwise}
\end{cases}\\  
((\na,i)::\Sigma)[\Tau] &\defn
    \begin{cases}(\na,\Tau(\na))::(\Sigma[\Tau]) &\text{if }\na\in\nu(\Tau)\\ (\na,i)::(\Sigma[\Tau]) &\text{otherwise}\end{cases} \\
((\na,i)::\Sigma)+\Tau &\defn
\begin{cases}
(\na,i)::(\Sigma+\Tau) &\text{ if }\na\notin\nu(\Tau)\\ \text{undefined} &\text{ otherwise}
\end{cases}
\end{align*}
Moreover, we write $\st{s}$ for the store of the final S-move in $s$. 
\nt{For instance, by {\em Prev-PQ} and {\em Prev-PA}, if $o^\Sigma p^\Tau$ are consecutive inside an S-play then $\Tau=\Sigma[\Tau]\setminus(\Sigma\setminus\Tau)+(\Tau\setminus\Sigma)$.}

It will also be convenient 
to introduce the following store-constructor. For stores $\Sigma_0,\Sigma_1,\Sigma_2$ we define
\[ \nice(\Sigma_0,\Sigma_1,\Sigma_2) \defn \Sigma_0[\Sigma_2]\remv(\Sigma_1\remv\Sigma_2)+(\Sigma_2\remv\Sigma_1)\,.\]
Considering $\Sigma_1,\Sigma_2$ as \emph{consecutive},
the constructor first updates $\Sigma_0$ with values from $\Sigma_2$, removes those names that have been \emph{dropped} in $\Sigma_2$ and then adds those that have been \emph{introduced} in it.

\begin{definition}
Let $A,B,C$ be arenas and $s\in\Splays{A\arr B},t\in\Splays{B\arr C}$. We say that $s,t$ are \emph{compatible}, written $s\asymp t$, if $\erase{s}\rest B=\erase{t}\rest B$ and $\nu(s)\cap\nu(t)=\varnothing$. In such a case, we define their \emph{interaction}, $s\iseq t$, and their \emph{mix}, $s\mix t$, recursively as follows,
\begin{align*}
\ee\iseq\ee                 &\defn \ee
& \ee\mix\ee                          &\defn \ee \\
sm_A^\Sigma\iseq t            &\defn (s\iseq t)m_A^{sm_A^\Sigma\mix t}
& sn^\Tau m_{A(P)}^\Sigma\mix t                 &\defn \nice(sn^\Tau\mix t,\Tau,\Sigma) \\
&& sm_{A(O)}^\Sigma\mix t                 &\defn \tilde\Sigma[s\iseq t] \\
sm_B^\Sigma\iseq tm_B^{\Sigma'}    &\defn (s\iseq t)m_B^{sm_B^\Sigma\mix tm_B^{\Sigma'}}
& sn^\Tau m_{B(P)}^\Sigma\mix tm_{B(O)}^{\Sigma'}     &\defn \nice(sn^\Tau\mix t,\Tau,\Sigma) \\
&& s m_{B(O)}^{\Sigma'}\mix tn^\Tau m_{B(P)}^\Sigma     &\defn \nice(s\mix tn^\Tau,\Tau,\Sigma) \\
s\iseq tm_C^\Sigma            &\defn (s\iseq t)m_C^{s\mix tm_C^\Sigma}
& s\mix tn^\Tau m_{C(P)}^\Sigma                 &\defn \nice(s\mix tn^\Tau,\Tau,\Sigma) \\
&& s\mix tm_{C(O)}^\Sigma                 &\defn \tilde\Sigma[s\iseq t]
\end{align*}
where 
justification pointers in $s\iseq t$ are inherited from $s$ and $t$, and
$\tilde\Sigma$ is the store of $m_{A/C(O)}$'s justifier in $s\iseq t$. 
Note that $s\mix t$ is the store of the last S-move in $s\iseq t$. 
The \emph{composite} of $s$ and $t$ is
\[ s; t \defn (s\iseq t)\upharpoonright AC\,. \]
We moreover let
\[ \Sinter{A,B,C} \defn \{ s\iseq t\ |\ s\in\Splays{A\arr B}\land t\in\Splays{B\arr C}\land s\asymp t \} \]
be the set of \boldemph{S-interaction sequences} of $A,B,C$.
\end{definition}

\nt{
\begin{example}\label{ex:playcomp}
Let us 
demonstrate how to
compose S-plays
from the following strategies: 
\begin{align*}
\sigma &\defn\sem{\seq{}{\lambda x^\vart.\,x:={!x}+1;{!x}:\vart\to\int}}: 1\to(\sem{\vart}\Rarr\Z)\\
\tau &\defn\sem{\seq{f:\vart\to\int}{(\new{x}{fx})+\new{x}{fx}:\int}}: (\sem{\vart}\Rarr\Z)\rarr\Z
\end{align*}
We depict the two prearenas below (compared to Example~\ref{ex:cell}, in the prearena on the right we have replaced $q$ and $a$ with concrete moves $\ast$ and $i\in\Z$).
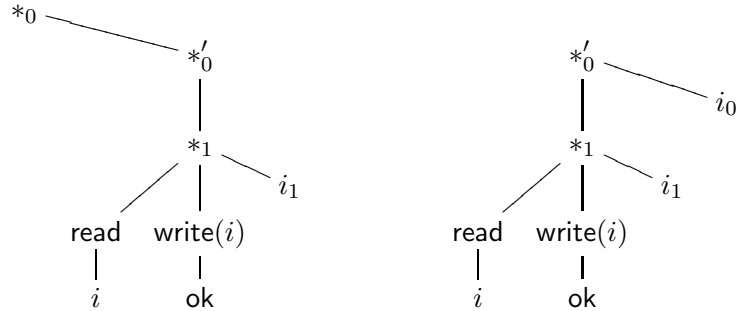
\[
\begin{aligned}
\xymatrix@R=.5mm@C=2mm{
\ast_0\ar@{-}[drr] 
\\
&& \ast_0'\ar@{-}[dd] 
&&&&\ast_0'\ar@{-}[dd]\ar@{-}[drr] 
\\
&&&&&&&& i_0
\\
&&\ast_1\ar@{-}[dd]\ar@{-}[ddl]\ar@{-}[dr] 
&&
&&\ast_1\ar@{-}[dd]\ar@{-}[ddl]\ar@{-}[dr] 
\\
&&& i_1
&\quad\qquad&&& i_1 
\\
&\mread\ar@{-}[dd]  & \mwrite{i}\ar@{-}[dd]
&&&
\mread\ar@{-}[dd]  & \mwrite{i}\ar@{-}[dd] 
\\
\\
&i & \mok
&&&
i & \mok
}
\end{aligned}
\]
The S-strategy $\sigma$ is given by even-length prefixes of S-plays of the form:
\[
\begin{aligned}\\[4mm]
&\rnode{A}{\ast_0}\qweq 
\rnode{AA}{\ast_0'}\qweq
\rnode{B}{\ast_1}
\qweq
\rnode{CA}{\mread}\qweq\rnode{DA}{i}
\qweq
\rnode{CB}{\mwrite{i{+}1}}\qweq\rnode{DB}{\mok}
\qweq
\rnode{CA}{\mread}\qweq\rnode{DA}{j}
\qweq
\rnode{C}{j_1}
\justf[,angleA=145]{C}{B}
\qweq
\rnode{XB}{\ast_1}
\qweq
\rnode{XCA}{\mread}\qweq\rnode{XDA}{i'}
\qweq
\rnode{XCB}{\mwrite{i'{+}1}}\qweq\rnode{XDB}{\mok}
\qweq
\rnode{CA}{\mread}\qweq\rnode{DA}{j'}
\qweq
\rnode{XC}{j'_1}
\qweq\cdots
\justf[,angleA=145]{XC}{XB}\justf[,angleA=140]{XB}{AA}
\end{aligned}
\]
where the $\sf read$'s and $\sf write$'s point to the last $\ast_1$ on their left (and
each other missing pointer is assumed to point to the next move on the left). 
On the other hand, the elements of $\tau$ are even-length prefixes of S-plays with the following pattern:
\begin{align*}
&\rnode{A}{\ast'_0}\qwe 
\rnode{B}{\ast_1^{(\na,0)}}
\qweq
\left(\rnode{CA}{\mread}^{(\na,0)}\qweq\rnode{DA}{0}^{(\na,0)}\right)^*
\qweq
\rnode{CB}{\mwrite{i}}^{(\na,0)}\qweq\rnode{DB}{\mok}^{(\na,i)}
\qweq
\left(\rnode{CC}{\mread}^{(\na,i)}\qweq\rnode{DC}{i}^{(\na,i)}\right)^*
\cdots\qweq
\rnode{C}{j_1}^{(\na,k)}
\\
&\quad\;\;
\rnode{B}{\ast_1^{(\na',0)}}
\qweq
\left(\rnode{CA}{\mread}^{(\na',0)}\qweq\rnode{DA}{0}^{(\na',0)}\right)^*
\qweq
\rnode{CB}{\mwrite{i'}}^{(\na',0)}\qweq\rnode{DB}{\mok}^{(\na',i')}
\qweq
\left(\rnode{CC}{\mread}^{(\na',i')}\qweq\rnode{DC}{i}^{(\na',i')}\right)^*
\cdots\qweq
\rnode{C}{j'_1}^{(\na',k')}\qweq
\rnode{D}{(j{+}j')_0}
\end{align*}
Observe that an S-play $s\in\sigma$ can only be composed with a $t\in\tau$ if it satisfies the read-write discipline of variables: each $\sf read$ move must be answered by the last $\sf write$ value, apart of the initial $\sf read$ that should be answered by 0. 
Moreover, $s$ must feature at most two moves $\ast_1$.
We therefore consider the following S-plays $s\asymp t$.
\begin{align*}
s&=\, \rnode{A}{\ast_0}\qweq 
\rnode{AA}{\ast_0'}\qweq
\rnode{B}{\ast_1}
\qweq
\rnode{CA}{\mread}\qweq\rnode{DA}{0}
\qweq
\rnode{CB}{\mwrite{1}}\qweq\rnode{DB}{\mok}
\qweq
\rnode{CC}{\mread}\qweq\rnode{DC}{1}
\qweq
\rnode{C}{1_1}
\qweq
\rnode{B}{\ast_1}
\qweq
\rnode{CA}{\mread}\qweq\rnode{DA}{0}
\qweq
\rnode{CB}{\mwrite{1}}\qweq\rnode{DB}{\mok}
\qweq
\rnode{CC}{\mread}\qweq\rnode{DC}{1}
\qweq
\rnode{C}{1_1}
\\
t&=\, \rnode{A}{\ast'_0}
\qweq \rnode{B}{\ast_1}^{(\na,0)}
\qweq
\rnode{CA}{\mread}^{(\na,0)}\qweq\rnode{DA}{0}^{(\na,0)}
\qweq
\rnode{CB}{\mwrite{1}}^{(\na,0)}\qweq\rnode{DB}{\mok}^{(\na,1)}
\qweq
\rnode{CC}{\mread}^{(\na,1)}\qweq\rnode{DC}{1}^{(\na,1)}
\qweq
\rnode{CC}{1_1}^{(\na,1)}
\\
&\qquad\qweq\; \rnode{B}{\ast_1}^{(\na',0)}
\qweq
\rnode{CA}{\mread}^{(\na',0)}\qweq\rnode{DA}{0}^{(\na',0)}
\qweq
\rnode{CB}{\mwrite{1}}^{(\na',0)}\qweq\rnode{DB}{\mok}^{(\na',1)}
\qweq
\rnode{CC}{\mread}^{(\na',1)}\qweq\rnode{DC}{1}^{(\na',1)}
\qweq
\rnode{CC}{1_1}^{(\na',1)}\qweq\rnode{DC}{2_0}
\end{align*}
By composing, we obtain:
\begin{align*}
s\|t&=\, 
\rnode{A}{\ast_0}
\qweq 
\rnode{A}{\ast'_0}
\qweq 
\rnode{B}{\ast_1}^{(\na,0)}
\qweq
\rnode{CA}{\mread}^{(\na,0)}\qweq\rnode{DA}{0}^{(\na,0)}
\qweq
\rnode{CB}{\mwrite{1}}^{(\na,0)}\qweq\rnode{DB}{\mok}^{(\na,1)}
\qweq
\rnode{CC}{\mread}^{(\na,1)}\qweq\rnode{DC}{1}^{(\na,1)}
\qweq
\rnode{CC}{1_1}^{(\na,1)}
\\
&\qquad\qweq\; \rnode{B}{\ast_1}^{(\na',0)}
\qweq
\rnode{CA}{\mread}^{(\na',0)}\qweq\rnode{DA}{0}^{(\na',0)}
\qweq
\rnode{CB}{\mwrite{1}}^{(\na',0)}\qweq\rnode{DB}{\mok}^{(\na',1)}
\qweq
\rnode{CC}{\mread}^{(\na',1)}\qweq\rnode{DC}{1}^{(\na',1)}
\qweq
\rnode{CC}{1_1}^{(\na',1)}\qweq\rnode{DC}{2_0}
\end{align*}
and $s;t=\ast_0\,2_0$.

As a side-note, let us observe that, had we considered a slightly different strategy to compose S-plays from $\sigma$ with:
\begin{align*}
\tau' &\defn\sem{\seq{f:\vart\to\int}{\new{x}{fx+fx}:\int}}: (\sem{\vart}\Rarr\Z)\rarr\Z
\end{align*}
we would only be able to compose variants of $s$ and $t$:
\begin{align*}
s'&=\, \rnode{A}{\ast_0}\qweq 
\rnode{AA}{\ast_0'}\qweq
\rnode{B}{\ast_1}
\qweq
\rnode{CA}{\mread}\qweq\rnode{DA}{0}
\qweq
\rnode{CB}{\mwrite{1}}\qweq\rnode{DB}{\mok}
\qweq
\rnode{CC}{\mread}\qweq\rnode{DC}{1}
\qweq
\rnode{C}{1_1}
\qweq
\rnode{B}{\ast_1}
\qweq
\rnode{CA}{\mread}\qweq\rnode{DA}{1}
\qweq
\rnode{CB}{\mwrite{2}}\qweq\rnode{DB}{\mok}
\qweq
\rnode{CC}{\mread}\qweq\rnode{DC}{2}
\qweq
\rnode{C}{2_1}
\\
t'&=\, \rnode{A}{\ast'_0}
\qweq \rnode{B}{\ast_1}^{(\na,0)}
\qweq
\rnode{CA}{\mread}^{(\na,0)}\qweq\rnode{DA}{0}^{(\na,0)}
\qweq
\rnode{CB}{\mwrite{1}}^{(\na,0)}\qweq\rnode{DB}{\mok}^{(\na,1)}
\qweq
\rnode{CC}{\mread}^{(\na,1)}\qweq\rnode{DC}{1}^{(\na,1)}
\qweq
\rnode{CC}{1_1}^{(\na,1)}
\\
&\qquad\qweq\; \rnode{B}{\ast_1}^{(\na,1)}
\qweq
\rnode{CA}{\mread}^{(\na,1)}\qweq\rnode{DA}{0}^{(\na,1)}
\qweq
\rnode{CB}{\mwrite{2}}^{(\na,1)}\qweq\rnode{DB}{\mok}^{(\na,2)}
\qweq
\rnode{CC}{\mread}^{(\na,2)}\qweq\rnode{DC}{2}^{(\na,2)}
\qweq
\rnode{CC}{2_1}^{(\na,2)}\qweq\rnode{DC}{3_0}
\end{align*}
and thus obtain $s';t'=\ast_0\,3_0$.
\end{example}
}

Given an interaction sequence $u$ and a move $m$ of $u$, we call $m$ a \boldemph{generalised P-move} if it is a P-move in $AB$ or in $BC$ \nt{(by which we mean a P-move in $A\to B$ or $B\to C$ respectively)}. We call $m$ an \boldemph{external O-move} if it is an O-move in $AC$.
The notions of P-view and O-view extend to interaction sequences as follows. Let $p$ be a generalised P-move and $o$ be an external O-move.
\begin{align*}
  \pv{o} &\defn o  & \ov{\ee} &\defn\ee \\
  \pv{up^\Sigma} &\defn \pv{u}\,p^\Sigma & \ov{uo^\Sigma} &\defn \ov{u}\,o^\Sigma \\
  \pv{u \rnode{A}{m^\Tau}\cdots\rnode{B}{o}^\Sigma\justf[,angleB=20]{B}{A}} &\defn \pv{u}\,m^\Tau o^\Sigma &
  \ov{u \rnode{A}{m^\Tau}\cdots\rnode{B}{p}^\Sigma\justf[,angleB=20]{B}{A}} &\defn \ov{u}\,m^\Tau p^\Sigma
\end{align*}
Observe that if $u$ ends in a P-move in $AB$ (resp.~$BC$) then $\ov{u}=\ov{u\rest AB}$ ($\ov{u\rest BC}$).

We now show that play-composition is well-defined. The result follows from the next two lemmas.
The first one is standard, while the proof of the second one is given in Appendix~\ref{apx:plays}.

\begin{lemma}[Zipper Lemma~\cite{HY97}]
If $s\in\Splays{A\arr B}$, $t\in\Splays{B\arr C}$ and $s\asymp t$ then either $s\,{\rest}B=t=\epsilon$,
or $s$ ends in $A$ and $t$ in $B$ (with a P-move in $BC$), 
or $s$ ends in $B$ (with a P-move in $AB$) and $t$ in $C$, 
or both $s$ and $t$ end in $B$ (with the same move).
\end{lemma}

\begin{lemma}\label{l:inter}
Suppose $s\in\Splays{A\arr B},t\in\Splays{B\arr C}$, $s\asymp t$ and $p$ a generalised P-move.
\begin{enumerate}\renewcommand{\theenumi}{\alph{enumi}}
  \item If $s\iseq t=un^\Tau p^\Sigma$ then $\nu(\Sigma\remv\Tau)\cap\nu(un^\Tau)=\varnothing$.
  \item For any $\na$, $\pv{s\iseq t}$ is in block-form: $\pv{s\iseq t}=u_1u_2u_3$ where $\na$ appears in every move of $u_2$, $\na\notin\nu(u_1)\cup\nu(u_3)$  and if $u_2\neq\ee$ then its first move introduces $\na$ in $s\iseq t$.
  \item If $s\iseq t=un^\Tau p^\Sigma$ and $\na\in\nu(\Tau\remv\Sigma)$ then $\na$ is closed in $un^\Tau$.
  \item If $s\iseq t=\cdots \rnode{A}{n^\Tau}\cdots \rnode{B}{m}^\Sigma\justf[,angleB=20]{B}{A}$ then $\Tau\Substore\Sigma$. If $m$ is an answer then $\Sigma\Substore\Tau$.
  \item If $s\iseq t=u_1n^\Tau p^\Sigma u_2$ and $\na\in\nu(\Tau\remv\Sigma)$ then $\na\notin\nu(u_2)$.
  \item If $s\iseq t=um^\Sigma$ with $m$ an O-move in $AC$ then, for any $\na\in\nu(\Sigma)$, the last appearance of $\na$ in $u$ occurs in $AC$.
  \item If $s\iseq t=um^\Sigma$ and $s=s'm^{\Sigma'}$ or $t=t'm^{\Sigma'}$ then $\Sigma'\substore\Sigma$ and $\Sigma[\Sigma']=\Sigma$. Moreover, $\Sigma[\st{s}]=\Sigma[\st{t}]=\Sigma$.
    \item If $s\iseq t=un^\Tau p^\Sigma$ then $\Tau\remv(\Tau\remv\Sigma)\Substore\Sigma$ and $\Tau\remv\Sigma\substorE\Tau$.
\end{enumerate}
\end{lemma}


\begin{proposition}[Compositionality]
S-play composition is well defined, that is, if $s\in\Splays{A\arr B},t\in\Splays{B\arr C}$ and $s\asymp t$ then $s;t\in\Splays{A\arr C}$.
\end{proposition}
\begin{proof}
We need to verify the 5 conditions of Definition~\ref{def:splay}. Init is straightforward. Just-P follows from part~(d) of the Lemma~\ref{l:inter}. Just-O follows directly from the definition of interaction sequences. Prev-PQ follows from parts (a,c,h) Lemma~\ref{l:inter}. Finally, Val-O follows from part~(f) of Lemma~\ref{l:inter} and the definition of interaction sequences.
\end{proof}

\subsection{Associativity}

We next show that composition of S-plays is associative. 
We first extend interactions to triples of S-plays. For $s\in\Splays{A\arr B},t\in\Splays{B\arr C},r\in\Splays{C\arr D}$ with $(s;t)\asymp r$,  $s\asymp(t;r)$ and $\nu(s)\cap\nu(r)=\varnothing$, we define $s\iseq t\iseq r$ and $s\mix t\mix r$ as follows,
\begin{align*}
  \ee\iseq\ee\iseq\ee &\defn \ee & \ee\mix\ee\mix\ee &\defn\ee \\
  sm^\Sigma_A\iseq t\iseq r &\defn (s\iseq t\iseq r)m_A^{sm^\Sigma_A\mix t\mix r} &
  sn^\Tau m^\Sigma_{A(P)}\mix t\mix r &\defn \nice(sn^\Tau\mix t\mix r,\Tau,\Sigma) \\ &&
  sm^\Sigma_{A(O)}\mix t\mix r &\defn \tilde\Sigma[s\mix t\mix r] \\
  sm^\Sigma_B\iseq tm^{\Sigma'}_B\iseq r &\defn (s\iseq t\iseq r)m_B^{sm^\Sigma_B\mix tm^{\Sigma'}_B\mix r} &
  sn^\Tau m^\Sigma_{B(P)}\mix tm^{\Sigma'}_{B(O)}\mix r &\defn \nice(sn^\Tau\mix t\mix r,\Tau,\Sigma) \\ &&
  sm^{\Sigma'}_{B(O)}\mix tn^\Tau m^\Sigma_{B(P)}\mix r &\defn \nice(s\mix tn^\Tau\mix r,\Tau,\Sigma) \\
  s\iseq tm^\Sigma_C\iseq rm^{\Sigma'}_C &\defn (s\iseq t\iseq r)m_C^{s\mix tm^\Sigma_C\mix rm^{\Sigma'}_C} &
  s\mix tn^\Tau m^\Sigma_{C(P)}\mix rm^{\Sigma'}_{C(O)} &\defn \nice(s\mix tn^\Tau\mix r,\Tau,\Sigma) \\ &&
  s\mix tm^{\Sigma'}_{C(O)}\mix rn^\Tau m^\Sigma_{C(P)} &\defn \nice(s\mix t\mix rn^\Tau,\Tau,\Sigma) \\
  s\iseq t\iseq rm^\Sigma_D &\defn (s\iseq t\iseq r)m_D^{s\mix t\mix rm^\Sigma_D} &
  s\mix t\mix rn^\Tau m^\Sigma_{D(P)} &\defn \nice(s\mix t\mix rn^\Tau,\Tau,\Sigma) \\ &&
  s\mix t\mix rm^\Sigma_{D(O)} &\defn \tilde\Sigma[s\mix t\mix r]
\end{align*}
where $\tilde\Sigma$ is the store of the move justifying $m_{A(O)}^\Sigma$ and $m_{D(O)}^\Sigma$ respectively.

The next lemma proves a form of associativity in $\nice$ that will be used in the proof of the next
proposition.  Its proof is delegated to Appendix~\ref{apx:plays}.

\begin{lemma}\label{l:nice}
Let $\Sigma_1,\Sigma_2,\Sigma_3,\Sigma_4,\Sigma_5$ be stores such that, for any $\na$,
\begin{enumerate}\renewcommand{\theenumi}{\rm\alph{enumi}}
  \item $\na\in\nu(\Sigma_5\remv\Sigma_4)\implies\na\notin\nu(\Sigma_1)\cup\nu(\Sigma_2)\cup\nu(\Sigma_3)$,
  \item $\na\in\nu(\Sigma_1)\cap\nu(\Sigma_4)\implies \na\in\nu(\Sigma_2)$.
\end{enumerate}
Then, $\nice(\Sigma_1,\Sigma_2,\nice(\Sigma_3,\Sigma_4,\Sigma_5))=\nice(\nice(\Sigma_1,\Sigma_2,\Sigma_3),\Sigma_4,\Sigma_5)$.
\end{lemma}

\begin{proposition}[Associativity]\label{p:assoc}
If $s_1\in\Splays{A_1\arr A_2},s_2\in\Splays{A_2\arr A_3}$ and $s_3\in\Splays{A_3\arr A_4}$ with $s_1;s_2\asymp s_3$ and $s_1\asymp s_2; s_3$ then:
\begin{itemize}
\item  $(s_1;s_2);s_3 = (s_1\iseq s_2\iseq s_3)\rest A_1A_4 = s_1;(s_2;s_3)$,
\item  $\nice((s_1;s_2)\mix s_3,\st{s_1;s_2},s_1\mix s_2) = s_1\mix s_2\mix s_3 = \nice(s_1\mix(s_2;s_3),\st{s_2;s_3},s_2\mix s_3)$.
\end{itemize}
\end{proposition}

\begin{proof}
By induction on $|s_1\iseq s_2\iseq s_3|$. The base case is trivial. We examine the following inductive cases; the rest are similar.\smallskip

\noindent$-$ $(s_1m_{A_1}^{\Sigma};s_2);s_3=((s_1;s_2);s_3)m_{A_1}^{\Sigma'}\overset{\text{IH}}{=}(s_1\iseq s_2\iseq s_3)m_{A_1}^{\Sigma'}\rest A_1A_4$ with $\Sigma'=(s_1m_{A_1}^{\Sigma};s_2)\mix s_3$. Moreover, as $\st{s_1m_{A_1}^{\Sigma};s_2}=s_1m_{A_1}^{\Sigma}\mix s_2$ and using Lemma~\ref{l:inter}(g),
\[
\nice((s_1;s_2m_{A_1}^{\Sigma});s_3,\st{s_1m_{A_1}^{\Sigma}; s_2},s_1m_{A_1}^{\Sigma}\mix s_2)={\Sigma'[s_1m_{A_1}^{\Sigma}\mix s_2]}=\Sigma'.
\]
We still need to show that $\Sigma'=s_1m_{A_1}^{\Sigma}\mix s_2\mix s_3$. If $m_{A_1}$ is an O-move this is straightforward. If a P-move and, say, $s_1=s_1'n^{\Tau}$ then $\Sigma'=\nice((s_1;s_2)\mix s_3,\st{s_1;s_2},\nice(s_1\mix s_2,\Tau,\Sigma))$ and
$s_1m_{A_1}^{\Sigma}\mix s_2\mix s_3=\nice(s_1\mix s_2\mix s_3,\Tau,\Sigma)\overset{\text{IH}}{=}
\nice(\nice((s_1;s_2)\mix s_3,\st{s_1;s_2},s_1\mix s_2),\Tau,\Sigma)$. These are equal since they satisfy the hypotheses of Lemma~\ref{l:nice}.
Moreover, 
\[s_1m_{A_1}^{\Sigma};(s_2;s_3)=(s_1;(s_2;s_3))m_{A_1}^{\Sigma''}\overset{\text{IH}}{=}(s_1\iseq
s_2\iseq s_3)m_{A_1}^{\Sigma''}\rest A_1A_4
\] with $\Sigma''=s_1m_{A_1}^\Sigma\mix(s_2;s_3)$. Note that
$\st{s_2;s_3}=s_2\mix s_3$, so it suffices to show that
$\Sigma''=s_1m_{A_1}^\Sigma\mix s_2\mix s_3$. If $m_{A_1}$ an O-move
then this is straightforward, otherwise 
\[\Sigma''=\nice(s_1\mix(s_2;s_3),\Tau,\Sigma)\overset{\text{IH}}{=}\nice(s_1\mix
s_2\mix s_3,\Tau,\Sigma)=s_1m_{A_1}^\Sigma\mix s_2\mix
s_3\,.\]\smallskip

\noindent$-$ $(s_1m_{A_2}^{\Sigma_1};s_2m_{A_2}^{\Sigma_2});s_3=(s_1;s_2);s_3\overset{\text{IH}}{=}(s_1\iseq s_2\iseq s_3)\rest A_1A_4=(s_1m_{A_2}^{\Sigma_1}\iseq s_2m_{A_2}^{\Sigma_2}\iseq s_3)\rest A_1A_4$. Assume WLOG that $m_{A_2}$ is a P-move in $A_1A_2$, so $\st{s_2;s_3}=s_2\mix s_3$, and suppose $s_1=s_1'n^\Tau$. Then,
\begin{align*}
\nice((&s_1m_{A_2}^{\Sigma_1};s_2m_{A_2}^{\Sigma_2})\mix s_3,\st{s_1m_{A_2}^{\Sigma_1};s_2m_{A_2}^{\Sigma_2}},s_1m_{A_2}^{\Sigma_1}\mix s_2m_{A_2}^{\Sigma_2})\\
&= \nice((s_1;s_2)\mix s_3,\st{s_1;s_2},\nice(s_1\mix s_2,\Tau,\Sigma_1)) \\
 s_1m_{A_2}^{\Sigma_1}\mix s_2m_{A_2}^{\Sigma_2}\mix s_3 &= \nice(s_1\mix s_2\mix s_3,\Tau,\Sigma_1)\overset{\text{IH}}{=}
\nice(\nice((s_1;s_2)\mix s_3,\st{s_1;s_2},s_1\mix s_2),\Tau,\Sigma_1)
\end{align*}
and equality follows from Lemma~\ref{l:nice}. Moreover,
\begin{align*}
& \nice(s_1m_{A_2}^{\Sigma_1}\mix(s_2m_{A_2}^{\Sigma_2};s_3),\st{s_2m_{A_2}^{\Sigma_2};s_3},s_2m_{A_2}^{\Sigma_2}\mix s_3)=
(s_1m_{A_2}^{\Sigma_1}\mix(s_2m_{A_2}^{\Sigma_2};s_3))[s_2m_{A_2}^{\Sigma_2}\mix s_3] \\
&\quad\overset{\text{lm.\,\ref{l:inter}(g)}}{=} s_1m_{A_2}^{\Sigma_1}\mix(s_2m_{A_2}^{\Sigma_2};s_3)= \nice(s_1\mix(s_2;s_3),\Tau,\Sigma_1)\\
& s_1m_{A_2}^{\Sigma_1}\mix s_2m_{A_2}^{\Sigma_2}\mix s_3 = \nice(s_1\mix s_2\mix s_3,\Tau,\Sigma_1)\overset{\text{IH}}{=}
\nice(\nice(s_1\mix(s_2;s_3),\st{s_2;s_3},s_2\mix s_3),\Tau,\Sigma_1)\\
&\quad=\nice((s_1\mix(s_2;s_3))[s_2\mix s_3],\Tau,\Sigma_1)\overset{\text{lm.\,\ref{l:inter}(g)}}{=}\nice(s_1\mix(s_2;s_3),\Tau,\Sigma_1)
\end{align*}
as required.
\end{proof}

\subsection{The categories \texorpdfstring{$\Scat$}{S} and \texorpdfstring{$\Scatinn$}{Sinn}}

We next show that S-strategies and arenas form a category, which we call $\mathcal{S}$, while innocent S-strategies form a wide subcategory of $\mathcal{S}$.
%
We start this section with a lemma on strong nominal sets. Recall that,
for a nominal set $X$ and $x,x'\in X$, we write $x\sim x'$ if there exists a permutation $\pi$ such that $x=\pi\actn x'$.

\begin{lemma}[Strong Support Lemma~\cite{Tze09}]
Let $X$ be a strong nominal set and let $x_i,y_i,z_i\in X$ with $\nu(y_i)\cap\nu(z_i)\subseteq\nu(x_i)$, for $i=1,2$.
Then, $(x_1,y_1)\sim(x_2,y_2)$ and $(x_1,z_1)\sim(x_2,z_2)$ imply $(x_1,y_1,z_1)\sim(x_2,y_2,z_2)$. \qed
\end{lemma}

We proceed to show compositionality of S-strategies. The \boldemph{interaction} of S-strategies $\sigma:A\arr B$ and $\tau:B\arr C$ is defined by:
\[ \sigma\iseq\tau\defn\{s\iseq t\ |\ s\in\sigma\land t\in\tau\land s\asymp t\} \]

First, some lemmas for determinacy.

\begin{lemma}
If $s_1\iseq t_1,s_2\iseq t_2\in\Sinter{A,B,C}$ then $s_1\iseq t_1=s_2\iseq t_2$ implies $s_1=s_2$ and $t_1=t_2$.
\nt{Consequently, if $s_1\iseq t_1\sim s_2\iseq t_2$ then $(s_1,t_1)\sim(s_2,t_2)$.}
\end{lemma}
\begin{proof}
The former part of the claim is shown by straightforward induction on the length of the interactions. For the latter part, if $s_1\iseq t_1\sim s_2\iseq t_2$ then there is some $\pi$ such that $s_1\iseq t_1=\pi\cdot( s_2\iseq t_2)=(\pi\cdot s_2)\iseq(\pi\cdot t_2)$. Hence, by the former part,  $(s_1,t_1)=(\pi\cdot s_2,\pi\cdot t_2)$, from which we obtain $(s_1,t_1)\sim(s_2,t_2)$.
\end{proof}

\begin{lemma}\label{l:det_inter}
If $\sigma:A\arr B$, $\tau:B\arr C$ are S-strategies and $u_1m_1^{\Sigma_1},u_2m_2^{\Sigma_2}\in\sigma\iseq\tau$ then $u_1\sim u_2$ and $m_1$ a generalised P-move imply $u_1m_1^{\Sigma_1}\sim u_2m_2^{\Sigma_2}$.
\end{lemma}
\begin{proof}
Let us assume that $u_im_i^{\Sigma_i}=s_i\iseq t_i=(s_i'\iseq t_i')\,m_i^{\Sigma_i}$. Then, by 
the previous lemma,
$(s_1',t_1')\sim(s_2',t_2')$. If, say, $m_1$ is a P-move in $AB$ then 
$m_2$ is also a P-move in $AB$ and
$s_1\sim s_2$, by Zipper Lemma and determinacy of $\sigma$, so $(s_1',m_1^{\Sigma_1'})\sim(s_2',m_2^{\Sigma_2'})$, where $s_i=s_i'm_i^{\Sigma_i'}$. Since $\nu(t_i')\cap\nu(m_i^{\Sigma_i'})=\varnothing$, we can apply the Strong Support Lemma to obtain $u_1m_1^{\Sigma_1}\sim u_2m_2^{\Sigma_2}$.
\end{proof}

\begin{lemma}\label{l:det_other}
Let $\sigma:A\arr B$, $\tau:B\arr C$ be S-strategies and $u_1,u_2\in\sigma\iseq\tau$ with $|u_1|\leq|u_2|$. Then, $u_1\rest AC=u_2\rest AC$ implies $u_1\sim\prefix u_2$.
\end{lemma}
\begin{proof}
By induction on $|u_1|$. The base case is trivial. 
Now suppose $u_1=u_1'm_1^{\Sigma_1}$, and let $u_2=u_2'u_2''$ with $u_2'$ being the greatest prefix of $u_2$ such that $u_1'\rest AC=u_2'\rest AC$. 
Then, by IH, $u_1'\sim\prefix u_2'$.
If $m_1$ a generalised P-move then, by previous lemma, $u_1\sim\prefix u_2$.
If $m_1$ an external O-move then $u_1'$ ends in a P-move in $AC$ and, by $u_1'\rest AC=u_2'\rest AC$ and Zipper Lemma, $u_2'$ must end in the same move. Thus, $u_1\rest AC\sim u_2\rest AC$ implies $u_1\sim\prefix u_2$.
\end{proof}

Now some lemmas for innocence. We say that a move $m$ in an interaction sequence $s\in\Sinter{A,B,C}$ is a \emph{generalised O-move} if it is an O-move in $AB$ or $BC$.

\begin{lemma}\label{l:proj}
If $s_1,s_2\in\Splays{A\arr B}$, $t_1,t_2\in\Splays{B\arr C}$ and $s_i\iseq t_i$ end in a generalised O-move in component $X$,
\begin{enumerate}\renewcommand{\theenumi}{\rm\alph{enumi}}
\item if $X=AB$ then $\pv{(s_1\iseq t_1)\rest AB}=\pv{(s_2\iseq t_2)\rest AB}\implies\pv{s_1}=\pv{s_2}$,
\item if $X=BC$ then $\pv{(s_1\iseq t_1)\rest BC}=\pv{(s_2\iseq t_2)\rest BC}\implies\pv{t_1}=\pv{t_2}$.
\end{enumerate}
\end{lemma}
\begin{proof}
We show (a) by induction on $|s_1|\geq1$, and (b) is proved similarly. The base case is obvious. If $s_1=s_1'n^{\Tau_1}s_1''m^{\Sigma_1}$ with $m$ an O-move in $AB$ justified by $n$ then $s_2=s_2'n^{\Tau_2}s_2''m^{\Sigma_2}$ and, by IH, $\pv{s_1'}=\pv{s_2'}$. We need to show that $\Tau_1=\Tau_2$ and
$\Sigma_1=\Sigma_2$, while we know that the stores of the corresponding moves in $s_i\iseq t_i$ are equal, say $\Sigma_1'=\Sigma_2'$ and $\Tau_1'=\Tau_2'$. We have that $\Tau_i'=\nice(\st{(s_i\iseq t_i)_{<n^{\Tau_i'}}},\st{s_i'},\Tau_i)$, hence $\Tau_1\remv\st{s_1'}=\Tau_2\remv\st{s_2'}$. By IH we have that $\st{s_1'}=\st{s_2'}$, so if $\na\in\nu(\Tau_1)\cap\nu(\st{s_1'})$ then, by Lemma~\ref{l:inter}(g), $\na\in\nu(\Tau_1')\cap\nu(\st{s_1'})$ so $\na\in\nu(\Tau_2')\cap\nu(\st{s_2'})$, $\therefore\na\in\nu(\Tau_2)\cap\nu(\st{s_2'})$, and viceversa. Moreover, $\Tau_i(\na)=\Tau'_i(\na)$ for each such $\na$, thus $\Tau_1=\Tau_2$. This also implies that $\nu(\Sigma_1)=\nu(\Sigma_2)$ and so, by Lemma~\ref{l:inter}(g), $\Sigma_1=\Sigma_2$.
\end{proof}

\begin{lemma}\label{l:inn_inter}
If $\sigma:A\arr B$, $\tau:B\arr C$ are innocent S-strategies then, if $u_1m^{\Sigma_1},u_2\in\sigma\iseq\tau$ with $u_i$ ending in a generalised O-move and $\pv{u_1}\sim\pv{u_2}$ then there exists $u_2m^{\Sigma_2}\in\sigma\iseq\tau$ such that $\pv{u_1m^{\Sigma_1}}\sim\pv{u_2m^{\Sigma_2}}$.
\end{lemma}
\begin{proof}
Suppose $u_1$ ends in an O-move in $A$\HY the other cases are shown similarly. Then, $u_1m^{\Sigma_1}=s_1 m^{\Sigma'_1}\iseq t_1$ and $u_2=s_2\iseq t_2$ for some relevant S-plays of $\sigma,\tau$. Moreover, $\pv{u_1}\sim\pv{u_2}$ implies, by Lemma~\ref{l:proj}, that $\pv{s_1}\sim\pv{s_2}$ and thus, by innocence, there exists $s_2m^{\Sigma_2'}\in\sigma$ such that $s_1m^{\Sigma_1'}\sim s_2m^{\Sigma_2'}$. In fact, we can pick a $\Sigma_2'$ such that $s_2m^{\Sigma_2'}\asymp t_2$. Let $s_2m^{\Sigma_2'}\iseq t_2=u_2m^{\Sigma_2}\in\sigma\iseq\tau$. We have that $(\pv{u_1},\pv{s_1})\sim(\pv{u_2},\pv{s_2})$ and $(\pv{s_1},m^{\Sigma_1'})\sim(\pv{s_2},m^{\Sigma_2'})$ and, moreover, $\nu(\pv{u_i})\cap\nu(m^{\Sigma_i'})\subseteq\nu(\pv{s_i})$ for $i=1,2$ thus, by Strong Support Lemma, $(\pv{u_1},\pv{s_1},m^{\Sigma_1'})\sim(\pv{u_2},\pv{s_2},m^{\Sigma_2'})$. This implies that $\pv{u_1m^{\Sigma_1}}\sim\pv{u_2m^{\Sigma_2}}$.
\end{proof}

\begin{lemma}\label{l:inn_other}
Let $\sigma:A\arr B$, $\tau:B\arr C$ be innocent S-strategies and let $u_1,u_2\in\sigma\iseq\tau$ with $|\pv{u_1}|\leq|\pv{u_2}|$. Then, $\pv{u_1\rest AC}=\pv{u_2\rest AC}$  implies $\pv{u_1}\sim\prefix\pv{u_2}$.
\end{lemma}
\begin{proof}
Let us write $u_i$ for $s_i\iseq t_i$.
We argue by induction on $|\pv{u_1}|+|\pv{u_2}|$. The base cases are obvious.  Now let $u_1=u_1'm^{\Sigma_1}$ with $m$ a B-move. We have that $|u_1'|\leq|u_2|$ and $\pv{u_1'\rest AC}=\pv{u_2\rest AC}$ so, by IH, $\pv{u_1'}\sim\prefix\pv{u_2}$. In particular, $u_2=u_2'm_2^{\Sigma_2}u_2''$ with $\pv{u_2'}\sim\pv{u_1'}$ so, by Lemma~\ref{l:inn_inter}, there exists $u_2'm_1^{\Sigma_2'}\in\sigma\iseq\tau$ with $\pv{u_1}\sim\pv{u_2'm_1^{\Sigma_2'}}$. Using Lemma~\ref{l:det_inter}, $\pv{u_2'}m_2^{\Sigma_2}\sim \pv{u_2'}m_1^{\Sigma_2'}\sim \pv{u_1}$.\\
The case of $m$ being a P-move in $AC$ is proved along the same lines. On the other hand, if $m$ is an O-move in $AC$ justified by $n$ then $u_1=v_1n^\Tau v_1'm^\Sigma$ and $u_2=v_2n^\Tau v_2'm^\Sigma v_2''$ and, by IH, $\pv{v_1n^\Tau}=\pi\actn\pv{v_2n^\Tau}$. In particular, $\pi$ fixes $\dom(\Tau)$ and therefore $\pi\actn m^\Sigma=m^\Sigma$, $\therefore\pv{v_1n^\Tau}m^\Sigma=\pi\actn(\pv{v_2n^\Tau}m^\Sigma)$.
\end{proof}

\begin{proposition}
If $\sigma:A\arr B$, $\tau:B\arr C$ are S-strategies then $\sigma;\tau:A\arr C$ is an S-strategy.
If $\sigma$ and $\tau$ are innocent, so is $\sigma;\tau$.
\end{proposition}
\begin{proof}
Prefix closure and nominal closure are easy to establish by noting that they hold at the level of interactions, for $\sigma\iseq\tau$.
For O-closure, suppose $v\in\sigma;\tau$ and $vm^\Sigma\in\Splays{A\arr C}$ with, say, $m$ an O-move in $A$. Then, $v=s;t$ for some $s\in\sigma$, $t\in\tau$ with $s$ ending in a P-move in $A$. Moreover, we can construct a (unique) store $\Sigma'$ such that $sm^{\Sigma'}\in\Splays{A\arr B}$, by means of the O-Just and O-Val conditions. Thus, $sm^{\Sigma'}\in\sigma$ and $vm^\Sigma\in\sigma;\tau$. \\
For determinacy, suppose even-length $vm_i^{\Sigma_i}\in\sigma;\tau$ and $vm_i^{\Sigma_i}=s_i;t_i$ with $s_i,t_i$ not both ending in $B$, for $i=1,2$.
Let $s_i\iseq t_i=(s_i'\iseq t_i')m_i^{\Sigma_i}$ and suppose WLOG that $|s_1'\iseq t_1'|\leq|s_2'\iseq t_2'|$. Then, by Lemma~\ref{l:det_other}, $s_1'\iseq t_1'\sim\prefix s_2'\iseq t_2'$ and therefore, by Lemma~\ref{l:det_inter}, $s_1\iseq t_1\sim\prefix s_2\iseq t_2$. In particular, $vm_1^{\Sigma_1}\sim vm_2^{\Sigma_2}$.
\\
For innocence, let $v_1m_1^{\Sigma_1},v_2\in\sigma;\tau$ with $\pv{v_1}=\pv{v_2}$ being of odd length, and $u_1m_1^{\Sigma_1},u_2\in\sigma\iseq\tau$ such that $v_1m_1^{\Sigma_1}=u_1m_1^{\Sigma_1}\rest AC$, $v_2=u_2\rest AC$.
By Lemma~\ref{l:inn_other}, 
either $\pv{u_2}\sim\prefix\pv{u_1}$
or $\pv{u_1}\sim\prefix\pv{u_2}$ and $\pv{u_1}\not\sim\pv{u_2}$.
In the latter case, by Lemmata~\ref{l:inn_inter} and~\ref{l:det_inter} we obtain $\pv{u_1m_1^{\Sigma_1}}\sim\prefix\pv{u_2}$, contradicting $\pv{v_1}=\pv{v_2}$. 
In the former case, let us assume $u_2$ is of maximum length such that $u_2\rest AC=v_2$ and $|\pv{u_2}|\leq|\pv{u_1}|$. Then, by Lemma~\ref{l:inn_inter}, there exists $u_2m_2^{\Sigma_2}\in\sigma\iseq\tau$ such that either $\pv{u_2m_2^{\Sigma_2}}\sim\prefix\pv{u_1}$ or $\pv{u_2m_2^{\Sigma_2}}\sim \pv{u_1m_1^{\Sigma_1}}$. The former case contradicts maximality of $u_2$, while the latter implies $v_2m_2^{\Sigma_2}\in\sigma;\tau$ and $\pv{v_1m_1^{\Sigma_1}}\sim\pv{v_2m_2^{\Sigma_2}}$.
\end{proof}

We proceed to construct a category of arenas and S-strategies. For each arena $A$, the identity S-strategy $\id_A:A\arr A$ is the \emph{copycat} S-strategy:
\[
\id_A = \makeset{ s\in\Splays{A\arr A}\ |\ \exists s'.\, s\prefix s'\land s'\rest A_l=s'\rest  A_r}
\]
where by $A_l$ and $A_r$ we denote the LHS and RHS arena $A$ of $A\arr A$ respectively.

\begin{proposition}
Let $\sigma_i:A_i\arr A_{i+1}$ for $i=1,2,3$. Then, $\sigma_i;\id_{A_{i+1}}=\id_{A_i};\sigma_i=\sigma_i$ and $\sigma_1;(\sigma_2;\sigma_3)=(\sigma_1;\sigma_2);\sigma_3$.
\end{proposition}
\begin{proof}
Composition with identities is standard. For associativity, if $s_1;(s_2;s_3)\in \sigma_1;(\sigma_2;\sigma_3)$ then there are $s_i'\sim s_i$ such that $\nu(s_{i_1}')\cap(\nu(s_{i_2})\cup\nu(s_{i_3})\cup\nu(s_{i_2}')\cup\nu(s_{i_3}'))=\varnothing$, for any distinct $i_1,i_2,i_3$. These S-plays satisfy the hypotheses of Proposition~\ref{p:assoc}, hence $s_1';(s_2';s_3')=(s_1';s_2');s_3'\in(\sigma_1;\sigma_2);\sigma_3$. Moreover, since $\nu(s_2,s_2')\cap\nu(s_3,s_3')=\varnothing$, $s_i\sim s_i'$ imply $(s_2,s_3)\sim(s_2',s_3')$ and therefore $s_2;s_3\sim s_2';s_3'$.
Since 
{$\nu(s_1,s_1')\cap\nu(s_2;s_3,s_2';s_3')=\emptyset$,}
we have $s_1;(s_2;s_3)\sim s_1';(s_2';s_3')$, hence $s_1;(s_2;s_3)\in(\sigma_1;\sigma_2);\sigma_3$. Other direction proved dually.
\end{proof}

We can now define our categories of games with stores. \nt{In the following definition we also make use of the observation that identity S-strategies are innocent.}

\begin{definition}
Let $\Scat$ be the category whose objects are arenas and, for each pair of arenas $A,B$, the morphisms are given by
$\Scat(A,B) = \makeset{\sigma:A\to B\ |\ \sigma\text{ an S-strategy}}$. Let $\Scatinn$ be the wide subcategory of $\Scat$ of innocent S-strategies.
\end{definition}

\subsection{The model of \texorpdfstring{$\iacbv$}{IAcbv}}

We next construct the model of $\iacbv$ in $\Scatinn$.
An innocent S-strategy $\sigma$ is specified by its \emph{view-function}, $\viewf(\sigma)$, defined as follows.
\[ 
\viewf(\sigma)\defn \{\pv{s}\ |\ s\in\sigma\land \mathsf{even}(|s|)\land s\not=\epsilon\} 
\]
Conversely, a \boldemph{preplay} is defined exactly like a (non-empty) S-play only that it does not necessarily satisfy the Val-O condition. Let us write $\preplays{A}$ for the set of preplays of $A$. Obviously, $\Splays{A}\subseteq\preplays{A}$. Moreover, if $s\in\Splays{A}$ then $\pv{s}\in\preplays{A}$. 

\nt{%
For instance, the cell strategy of Example~\ref{ex:cell} can be described as the least innocent S-strategy whose view-function contains the following preplays.
\[
\begin{aligned}\\[-2mm]
&\rnode{A}{q_0}\qwe \rnode{B}{q_1}^{(\na,0)}\qweq\rnode{C}{a_1}^{(\na,i)}\qweq\rnode{D}{a_0}\justf{B}{A}\justf{C}{B}\justf[,angleA=120]{D}{A}
\qquad
&\rnode{A}{q_0}\qwe \rnode{B}{q_1}^{(\na,0)}\qweq\rnode{C}{\mread}^{(\na,i)}\qweq\rnode{D}{i}^{(\na,i)}\justf{B}{A}\justf{C}{B}\justf{D}{C}
\qquad
&\rnode{A}{q_0}\qwe\rnode{B}{q_1}^{(\na,0)}\qweq\rnode{C}{\mwrite{j}}^{(\na,i)}\qweq\rnode{D}{\mok}^{(\na,j)}\justf{B}{A}\justf{C}{B}\justf{D}{C}
\end{aligned}
\]
We next make formal the connection between view-functions and innocent S-strategies.
}

A \boldemph{view-function} $f$ on $A$ is a subset of $\preplays{A}$ satisfying:
\begin{itemize}
 \item If $s\in f$ then $|s|$ is even and $\pv{s}=s$.\ ({\em  View})
 \item If $sn^\Tau m^\Sigma\in f$ and $s\not=\epsilon$ then $s\in f$.\ ({\em Even-Prefix Closure})
 \item If $s'\sim s\in f$ then $s'\in f$.\ ({\em Nominal Closure})
 \item If $sm_1^{\Sigma_1},sm_2^{\Sigma_2}\in f$ then $sm_1^{\Sigma_1}\sim sm_2^{\Sigma_2}$.\ ({\em Determinacy})
\end{itemize}
From a view-function $f$ we can derive an innocent S-strategy $\strat(f)$ by the following procedure. We set $\strat(f)\defn\bigcup_{i\in\omega}\strat_i(f)$, where
\begin{align*}
  \strat_{2i+1}(f) &\defn \{sm^\Sigma\in\Splays{A}\ |\ s\in\strat_{2i}(f) \} \\
  \strat_{2i+2}(f) &\defn \{sm^\Sigma\in\Splays{A}\ |\ s\in\strat_{2i+1}(f)\land \pv{sm^\Sigma}\in f\} 
\end{align*}
and $\strat_0(f)\defn\{\ee\}$.

\begin{lemma}
If $\sigma,f$ are an innocent S-strategy and a view-function respectively then $\viewf(\sigma),\strat(\sigma)$ are a view-function and an innocent S-strategy respectively. Moreover, $\strat(\viewf(\sigma))=\sigma$ and $\viewf(\strat(f))=f$. \qed
\end{lemma}

We can show that $\Scatinn$ exhibits the same kind of categorical structure as that obtained in~\cite{HY97} (in the context of call-by-value PCF),
which can be employed to model call-by-value higher-order computation with recursion.
In particular, let us call an S-strategy $\sigma:A\to B$ \boldemph{total} if for all $i_A\in I_A$ there is $i_Ai_B\in\sigma$. We write $\Scatinntot$ for the wide subcategory of $\Scatinn$ containing total innocent S-strategies.

For innocent S-strategies $\sigma:A\to B$ and $\tau:A\to C$, we define their \emph{left pairing} to be
$\ang{\sigma,\tau}_l= \strat(f)$, where $f$ is the view-function:
\begin{align*}
f\ &=\ \makeset{ s\in\preplays{A\to B\otimes C}\ |\ s\in \viewf(\sigma) \lland s\rest (B\otimes C)=\epsilon}
\\
&\quad\; \cup \makeset{ i_As_1s_2\in\preplays{A\to B\otimes C}\ |\ \exists i_B.i_As_1i_B\in\viewf(\sigma)\lland i_As_2\in\viewf(\tau)}
\\
&\quad\; \cup \makeset{ i_As_1s_2(i_B,i_C)s\in\preplays{A\to B\otimes C}\ |\ 
i_As_1i_Bs\in\viewf(\sigma)\lland i_As_2i_C\in\viewf(\tau)}
\\
&\quad\; \cup \makeset{ i_As_1s_2(i_B,i_C)s\in\preplays{A\to B\otimes C}\ |\ 
i_As_1i_B\in\viewf(\sigma)\lland i_As_2i_Cs\in\viewf(\tau)}
\end{align*}
We can show that left pairing yields a product in $\Scatinntot$ with the usual projections:
\[
\pi_1:A\otimes B\to A =
\makeset{s\in\Splays{A\otimes B\to A}\ |\ |s|\leq 1}\cup
\makeset{(i_A,i_B)i_As\in\Splays{A\otimes B\to A}\ |\ i_Ai_As\in\id_{A}}
\]
and dually for $\pi_2$.
Moreover, for every $A,B,C$, there is a bijection
\[
\Lambda:\Scatinn(A\otimes B,C)\overset{\cong}{\to}\Scatinntot(A,B\Rightarrow C)
\]
natural in $A,C$. In particular, for each innocent $\sigma:A\otimes B\to C$, $\Lambda(\sigma)=\strat(f)$, where
\[
f=
\makeset{i_A{*}\,i_Bs\in\preplays{A\to B\Rightarrow C}\ |\ (i_A,i_B)s\in\viewf(\sigma)}\,.
\]
The inverse of $\Lambda$ is defined is an analogous manner. We set $\ev_{A,B}=\Lambda^{-1}(\id_{A\Arr B})$.

Thus, the functional part of $\iacbv$ can be interpreted in $\Scatinn$ using the same constructions as in~\cite{HY97}.
Assignment, dereferencing and $\mkvar$ can in turn be modelled using
the relevant (store-free) innocent strategies of~\cite{AM97b}. 
Finally, the denotation of $\new{x}{M}$ is obtained by using $\mathsf{cell}_\beta$ of Example~\ref{ex:cell}.
Let us write $\ssem{\cdots}$ for the resultant semantic map.

\begin{proposition}\label{prop:model}
For any  $\iacbv$-term $\seq{\Gamma}{M:\theta}$,
$\ssem{\seq{\Gamma}{M:\theta}}$ is an innocent S-strategy. 
\end{proposition}
\begin{proof}
We present here the (inductive) constructions pertaining to variables, 
\begin{itemize}
\item
$\ssem{\seq{\Gamma}{\new{x}{M}:\beta}}=\sem{\Gamma}\xrightarrow{\Lambda(\ssem{M})}\sem{\vart}\Rarr\sem{\beta}\xrightarrow{\cell_\beta}\sem{\beta}
$
\item 
$\ssem{\seq{\Gamma}{M\aasg N:\comt}}=\sem{\Gamma}\xrightarrow{\ang{\ssem{M};\pi_2,\ssem{N}}_l}
(\Z\Rarr1)\otimes\Z\xrightarrow{\ev}1
$
\item
$\ssem{\seq{\Gamma}{!M}:\expt}=\sem{\Gamma}\xrightarrow{\ssem{M};\pi_1}
1\Rarr\Z\xrightarrow{\cong}(1\Rarr\Z)\otimes1\xrightarrow{\ev}\Z
$
\item
$\ssem{\seq{\Gamma}{\badvar{M}{N}:\vart}}=\sem{\Gamma}\xrightarrow{\ang{\ssem{M},\ssem{N}}_l}
\sem{\vart}
$
\end{itemize}
and refer to~\cite{HY97} for the functional constructions and the treatment of fixpoints.
\end{proof}

\cutout{With the soundness and definability results in place, we could now proceed in the familiar way to define a fully abstract model of $\iacbv$ via  taking the intrinsic quotient.
However, this would be somewhat counterproductive.
It turns out that $\rml$ is a conservative extension of $\iacbv$ (Corollary~\ref{cor:conservativity}),
so the (simpler) fully abstract model of $\rml$ from~\cite{AM97b}, based on knowing
strategies, is already fully abstract for $\iacbv$.
In fact, our model can be related to knowing strategies more precisely. Observe that by erasing
storage annotations in an innocent S-strategy $\sigma$ we obtain a knowing strategy,
which we call $\erase{\sigma}$ (determinism follows from the fact that stores in O-moves
are uniquely determined and from block-allocation). 
Let us write $\sem{\cdots}$
for the knowing strategy semantics (cast in~\cite{HY97}). 
}%
{Our model of $\iacbv$, based on innocent $S$-strategies, is closely related to one based on knowing strategies.
First observe that by erasing storage annotations in an innocent S-strategy $\sigma$ one obtains a knowing strategy
(determinacy follows from the fact that stores in O-moves are uniquely determined).
We shall refer to that knowing strategy by $\erase{\sigma}$.
Next note that the (simpler) fully abstract model of $\rml$ from~\cite{AM97b}, based on knowing strategies, also yields a model of $\iacbv$.
Let us write $\sem{\cdots}$
for this knowing-strategy semantics (cast in the Honda-Yoshida setting).
Then we have:}

\begin{lemma}
For any $\iacbv$-term $\seq{\Gamma}{M:\theta}$,
$\sem{\seq{\Gamma}{M:\theta}} = \erase{\ssem{\seq{\Gamma}{M:\theta}}}$.
\end{lemma}

\subsection{Block-innocent strategies}

\begin{definition}
A (knowing) strategy $\sigma:A$ is \boldemph{block-innocent} if $\sigma=\erase{\sigma'}$ for some innocent S-strategy $\sigma'$.
\end{definition}

\begin{example}
Let us revisit the two plays from the Introduction. The first one
indeed comes from an innocent S-strategy (we reveal the stores below).
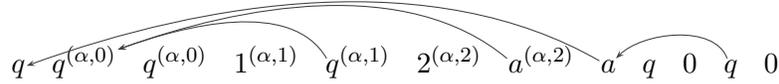
\[
\begin{array}{llllllllllll}
\\[3mm]
\rnode{A}{q}&
\rnode{B}{q^{(\na,0)}}&
\rnode{C}{q^{(\na,0)}}&
\rnode{D}{1^{(\na,1)}}&
\rnode{E}{q}^{(\na,1)}&
\rnode{F}{2^{(\na,2)}}&
\rnode{G}{a}^{(\na,2)}&
\rnode{H}{a}&
\rnode{I}{q}&
\rnode{J}{0}&
\rnode{K}{q}&
\rnode{L}{0}
\justf[,angleA=120,angleB=20]{E}{B}\justf[,angleA=140,angleB=20]{G}{B}\justf[,angleA=150,angleB=20]{H}{A}\justf{K}{H}
\end{array}
\]
For the second one to become innocent (in the setting with stores),
a store with variable $\alpha$, say, would need to be introduced in the second move, to justify the different responses in moves 4 and 6.
Then $\alpha$ must also occur in the seventh move by {\sc Just-O}, but
it must not occur in the eighth move by {\sc Just-P} (the $\mathit{PA}$ clause).
Hence, it will not be present in the ninth move by {\sc Just-O}.
Consequently, the last move 
is bound to break either {\sc Prev-PQ}(a) (if it contains $\alpha$) or
{\sc Just-P} (if it does not).
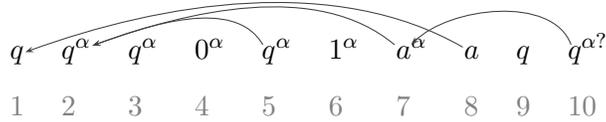
\[
\begin{array}{l@{\quad\;}l@{\quad\;}l@{\quad\;}l@{\quad\;}l@{\quad\;}l@{\quad\;}l@{\quad\;}l@{\quad\;}l@{\quad\;}l}
\\[3mm]
\rnode{A}{q}&
\rnode{B}{q^{\na}}&
\rnode{C}{q^{\na}}&
\rnode{D}{0^{\na}}&
\rnode{E}{q}^{\na}&
\rnode{F}{1^{\na}}&
\rnode{G}{a}^{\na}&
\rnode{H}{a}&
\rnode{I}{q}&
\rnode{K}{q}^{\na?}
\\[3mm]
\color{gray}
1&
\color{gray}2&\color{gray}3&\color{gray}4&\color{gray}5&\color{gray}6&\color{gray}7&\color{gray}8&\color{gray}9&\color{gray}10
\justf[,angleA=120,angleB=20]{E}{B}\justf[,angleA=140,angleB=20]{G}{B}\justf[,angleA=150,angleB=20]{H}{A}\justf{K}{G}
\end{array}
\]
\end{example}

\noindent The knowledge that strategies determined by $\iacbv$ are block-innocent
will be crucial in establishing a series of results in the following sections, where
we shall galvanise the correspondence by investigating 
full abstraction (Corollary~\ref{cor:conservativity}) and universality (Proposition~\ref{prop:univ}).
%
%



\section{Finitary definability and universality}

In this section we demonstrate that the game model of $\iacbv$ is complete when restricted to
\emph{finitary} or \emph{recursively presentable} innocent S-strategies. That is, every appropriately typed strategy of that kind is the denotation of some $\iacbv$ term. 
{Although finitary innocent S-strategies are subsumed by recursively presentable ones, the method of proving completeness in the latter case is much more involved and we therefore prove the two results separately. The two completeness results are called \emph{finitary definability} and \emph{universality} respectively.}

\subsection{Finitary definability}

We first formulate a decomposition lemma for innocent S-strategies which subsequently allows us to show the two results. The decomposition of innocent S-strategies follows the argument for call-by-value $\pcf$~\cite{HY97} except for the case in which the strategy replies to Opponent's (unique) initial move with a question that introduces a new name (case 8 in the lemma below).
Let us examine this case more closely, assuming
$\alpha$ to be the first variable from the non-empty store.
In order to decompose the S-strategy, say $\sigma$, consider any P-view $s$
in which $\alpha$ occurs in the second move $q_\alpha^{\Sigma_\alpha}$. It turns out
that $s$ must be of the form $q\, q_\alpha^{\Sigma_\alpha} s_\alpha s'$, where a move $m^\Sigma$ from $s$
contains $\alpha$ if, and only if, it is $q_\alpha^{\Sigma_\alpha}$ or in $s_\alpha$. In addition,
no justification pointers connect $s'$ to $q_\alpha^{\Sigma_\alpha} s_\alpha$, because of the
({\em Just-O}) and ({\em Close}) conditions.
This separation can be applied to decompose the view-function of $\sigma$. 
The $s_\alpha$ segments, put together as a single S-strategy,
can subsequently be dealt with in the style of factorisation arguments,
which remove $\alpha$ from moves
at the cost of an additional $\vart$-component.
Finally, to relate $s_\alpha$'s to the suitable $s'$ one can use numerical codes
for $q_\alpha^{\Sigma_\alpha} s_\alpha$. These ideas lie at the heart of the following result.

\cutout{
Innocent S-strategies can be \emph{decomposed} in a similar way
to the innocent strategies of~\cite{HY97}. There is one important exception, though,
which occurs when the second-move introduces a non-empty store
(our rules of play imply that the move must be a question).
Let $\alpha$ be the first variable from the non-empty store.
In order to decompose the strategy,  consider a P-view $s$
in which $\alpha$ occurs in the second move $q_\alpha$. It turns out
that $s= q q_\alpha s_\alpha s'$, where (the store of) a move $m^\Sigma$ from $s$
contains $\alpha$ if, and only if, it is $q_\alpha$ or in $s_\alpha$. In addition,
no justification pointers connect $s'$ to $q_\alpha s_\alpha$.
This separation can be applied to decompose the view-function of an innocent
strategy. The $s_\alpha$ parts, put together as a single S-strategy,
can subsequently be dealt with in the style of factorization arguments,
which remove $\alpha$ from moves
at the cost of an additional $\vart$-component.
Finally, to relate $s_\alpha$'s to the suitable $s'$ one can use numerical codes
for $q_\alpha s_\alpha$. These ideas lie at the heart of the following result.
By a finitary innocent $S$-strategy we mean
an innocent strategy whose view-function quotiented by
name-invariance is finite.
}

We fix a generic notation $\code{\_\!\_\,}$ for coding functions from enumerable sets to $\omega$. For example, $\code{i,j}$ encodes the pair $(i,j)$ as a number. There is an inherent abuse of notation in our coding notation which, nevertheless, we overlook for typographical economy. 
Moreover, we denote sequences $\theta_1,\cdots,\theta_n$ (where $n$ may be left implicit) as $\bar\theta$. In such cases, we may write $\bar\theta_i^j$ ($i\leq j$) for subsequences $\theta_i,\theta_{i+1},\cdots,\theta_j$.

\begin{lemma}[Decomposition Lemma (DL)] Let $\theta_1,\dots,\theta_m,\delta$ be types of $\iacbv$.
Each innocent S-strategy $\sigma:\sem{\theta_1}\otimes\cdots\otimes\sem{\theta_m}\arr\sem{\delta}$ can be decomposed as follows.
\begin{enumerate}[label=\bf\arabic*.]
  \item If $\theta_1,\dots,\theta_m=\overline{\theta}_1^{\,m'},\int,\overline{\theta}_{m'+2}^{\,m}$ with none of $\theta_1,\dots,\theta_{m'}$ being $\int$ then:
\begin{align*} 
\sigma &= \{(\overline{*}_1^{\,m'},i,\overline{q}_{m'+2}^{\,m})\,s\ |\ (\overline{*}_1^{\,m'},\overline{q}_{m'+2}^{\,m})\,s\in\tau_i\}\\
\text{where\ }
\tau_i &\defn \sem{\overline{\theta}_1^{\,m'}}\otimes\sem{\overline{\theta}_{m'+2}^{\,m}}\xrightarrow{\cong\,;\,\id\otimes i\otimes\id} \sem{\overline{\theta}_1^{\,m'}}\otimes \Z\otimes\sem{\overline{\theta}_{m'+2}^{\,m}}\xrightarrow{\sigma}\sem{\delta}
\end{align*}
  \item[$\bullet$\,\,] If none of $\theta_1,\dots,\theta_m$ is $\int$ then one of the following is the case.
  \begin{enumerate}[label=\bf\arabic*.,start=6]
    \item[$-$] $\overline{*}\,a\in\sigma$, in which case either:
    \begin{enumerate}[label=\bf\arabic*.,start=2]
      \item
           $\delta=\comt$ and $\sigma=\sem{\overline{\theta}}\xrightarrow{!}\one$ (the unique total S-strategy into $\one$),
      \item
           $\delta=\int$, $i\in\Z$ and $\sigma=\sem{\overline{\theta}}\xrightarrow{i}\Z$ (the unique total S-strategy into $\Z$ playing $i$),
      \item
           $\delta=\var$ and $\sigma=\ang{\sigma_1,\sigma_2}$ where $\sigma_i\defn\sigma;\pi_i$,
      \item
           $\delta=\delta'\arr \delta''$ and $\sigma=\Lambda(\sigma')$ where $\sigma'=\Lambda^{-1}(\sigma)$, that is:
      \[ \sigma':\sem{\overline{\theta}}\otimes\sem{\delta'}\arr\sem{\delta''}\defn\strat\{(\overline{*},q_{\delta'})\,s\ |\ \overline{*}\,a\,q_{\delta'}s\in\viewf(\sigma)\} \]
    \end{enumerate}
  \item
       $\overline{*}\,q\in\sigma$ with $q$ played in some $\theta_l=\var$, in which case $\sigma\cong\sigma'$ where 
\[ \sigma':\sem{\overline{\theta}_1^{\,l-1}}\otimes\sem{\overline{\theta}_{l+1}^{\,m}}\otimes(1\Arr\Z)\otimes(\Z\Arr1)\longrightarrow\sem{\delta}
\] is obtained from $\sigma$ by simply internally permuting and re-associating its initial moves.
  \item
       $\overline{*}\,q\in\sigma$ with $q$ played in some $\theta_l=\theta_{l}'\arr \theta_{l}''$, in which case
    \[ \sigma=\sem{\overline{\theta}}\xrightarrow{\ang{\Lambda(\sigma''),\pi_l,\sigma'}}(\sem{\theta_l''}\Arr \sem{\delta})\otimes(\sem{\theta_l'}\Arr\sem{\theta_l''})\otimes \sem{\theta_l'}\xrightarrow{\id\otimes\app;\,\app}\sem{\delta} \]
    where, taking $a$ to be $q$ (seen as an answer):
    \begin{align*}
      \sigma':\sem{\overline{\theta}}\arr\sem{\theta_l'}
&\defn\strat(\{\overline{*}\,a\}\cup\{\overline{*}\,a\,q_l'\,s\ |\ \overline{*}\,q\,q_l'\,s\in\viewf(\sigma)\land q_l'\in M_{\sem{\theta_l'}}\}) \\
      \sigma'':\sem{\overline{\theta}}\otimes\sem{\theta_l''}\arr \sem{\delta} &\defn \strat\{(\overline{*},q_l'')s\ |\ \overline{*}\,q\,a_l'' s\in\viewf(\sigma)\land q_l''=a_l'' \}
    \end{align*}
  \item
       $\overline{*}\,q^\Sigma\in\sigma$ with $\dom(\Sigma)=\na\cdots$, in which case
    \[ \sigma=\sem{\overline{\theta}}\xrightarrow{\ang{\Lambda(\sigma'),\id}}(\sem{\vart}\rarr\Z)\otimes\sem{\overline{\theta}}\xrightarrow{\cell\otimes\id}
    \Z\otimes\sem{\overline{\theta}}
    \xrightarrow{\sigma''}\sem{\delta}\]
    where:
    \begin{align*}
        \sigma':\sem{\vart}\otimes\sem{\overline{\theta}}\arr\Z &\defn\strat(\,\{\psi(sm^\Tau)\ |\ sm^\Tau\in\viewf(\sigma)\land\na\in\nu(\Tau)\} \\
        &\qquad\qquad\cup \{\psi(s)\code{s}\ |\ sm^\Tau\in\viewf(\sigma)\land\na\in\nu(\st{s}\remv\Tau)\}\,) \\
        \psi(o^{(\na,i)::\Tau} s)&\defn o^\Tau \rd^\Tau i^\Tau\psi(s) \\
        \psi(p^{(\na,i)::\Tau} s)&\defn \mwrite{i}^\Tau \ok^\Tau p^\Tau\psi(s) \\
        \sigma'':\Z\otimes\sem{\overline{\theta}}\arr\sem{\delta} &\defn \strat\{(\code{s},\overline{*})\,m^\Tau t\ |\ s\,m^\Tau t\in\viewf(\sigma)\land\na\in\nu(\st{s}\remv\Tau) \}
    \end{align*}
  \end{enumerate}
\end{enumerate}
\end{lemma}
\begin{proof}
Cases {\bf1}-{\bf7} are the standard ones that also occur for call-by-value $\pcf$~\cite{HY97}.
Case {\bf8} is the most interesting one. Here we exploit the fact that, once $\alpha$ occurs
in the second move of a P-view, it appears continuously (in the P-view) until it is dropped by Proponent.
Moreover, after $\alpha$ has been dropped, no move will ever have a justification pointer
to a move containing $\alpha$ (because of {\em Just-O} and {\em Close}). 
The $\sigma'$ strategy tracks the behaviour of $\sigma$
until $\alpha$ is dropped, at which point it returns the code of the current P-view.
$\sigma''$ in turn will take a code of such a P-view and will continue the play, 
as $\sigma$ would. Additionally, $\alpha$ is factored out  in $\sigma'$ through
an extra $\sem{\vart}$ arena, as in the factorisation argument of~\cite{AM97a}.
\end{proof}
\begin{definition}
We call an innocent S-strategy $\sigma$ \boldemph{finitary} if its view-function is finite modulo name-permutation, that is, if the set
\[
O(\viewf(\sigma)) = \makeset{\makeset{\pi\cdot s\ |\ \pi\in\permg}\ |\ s\in\viewf(s)}
\]
is finite. Accordingly, we call a block-innocent strategy finitary if the underlying innocent S-strategy is finitary.
\end{definition}
\begin{proposition}[Finitary definability]\label{prop:defin}
Let $\theta_1,\dots,\theta_m,\delta$ be types of $\iacbv$. For any finitary
innocent S-strategy $\sigma:\sem{\theta_1}\otimes\cdots\otimes\sem{\theta_m}\arr\sem{\delta}$ there exists a term $x_1:\theta_1,\dots,x_m:\theta_m\vdash M:\delta$ such that 
$\sigma=\ssem{\seq{x_1:\theta_1,\dots,x_m:\theta_m}{M}}$.
\end{proposition}
\begin{proof}
We rely on the decomposition lemma to reduce the suitably calculated size of the strategy.
The right measure is obtained by combining the size of the view-function quotiented by
name-permutation and the maximum number of names occurring in a single P-view.
It then suffices to establish that in  each case 
the reconstruction of the original strategy can be supported by the syntax.
For the first seven cases we can proceed as in~\cite{HY97}.
For the eighth case,
let $\seq{y:\vart, \Gamma}{M':\expt}$ and $\seq{x:\expt, \Gamma}{M'':\delta}$
be the terms obtained by IH for $\sigma'$ and $\sigma''$ respectively.
Then, in order to account for $\sigma$,  one can take $\letin{x=(\new{y}{M'})}{M''}$.
\end{proof}

\cutout{
\begin{proposition}[Finitary Definability and Universality]
\begin{asparaitem}
\item Any finitary innocent S-strategy
is $\iacbv$-definable.
\item Any recursively presentable innocent S-strategy is $\iacbv$-definable.
\end{asparaitem}
\end{proposition}}

\subsection{Universality}

We now proceed with the universality result.
In the rest of this section we closely follow the presentation of~\cite{AJM00}; the reader is referred thereto for a more detailed exposition of the background material.
Let us fix an enumeration of partial recursive functions such that $\phi_n$ is the $n$-th partial recursive function. 

The universality result concerns innocent S-strategies. Recall that S-strategies and their view-functions are saturated under name-permutations and, in fact, view-functions only become functions after nominal quotienting. To represent them we introduce an encoding scheme that is not dependent on names. 
Let us define a function $\eff$ which converts S-plays to plays in which moves are attached with lists of integers:
\[ 
\eff(so^\Sigma)\defn\eff(s)o^{\pi_2(\Sigma)}\,,\quad\eff(sm^\Sigma p^\Tau)\defn\eff(sm^\Sigma)p^{\pi_2(\Tau),|\Tau\remv\Sigma|}\,. 
\]
Thus, from an O-move $o^\Sigma$ we only keep the values stored in $\Sigma$, whereas in a P-move $p^\Tau$ we keep the values of $\Tau$ and a number indicating how many of the names of $\Tau$ are freshly introduced. Because of the conditions on stores that S-plays satisfy, $\eff$ maps two S-plays to the same encoding if, and only if, they are nominally equivalent. 
In the sequel we assume that S-plays are given using the encoding above.

For the rest of the section we assume that \nt{$\preplays{A}$} is recursively enumerable, which is clearly the case for denotable prearenas.

\begin{definition}
A subset of \nt{$\preplays{A}$} (for instance, a strategy or a view-function)
will be called \boldemph{recursively presentable} if it is a recursively enumerable subset of \nt{$\preplays{A}$}. A block-innocent strategy will be called  recursively presentable if the underlying
innocent S-strategy is recursively presentable.
\end{definition}
It follows that an innocent S-strategy $\sigma$ is recursively presentable if, and only if, its view-function is.
We therefore encode an innocent S-strategy $\sigma$ by $\code{\sigma}$, where the latter is the index $n$ such that $\viewf(\sigma)=\phi_n$ (with $\phi_n$ seen as a partial function from codes of P-views of S-plays to codes of P-moves).

\cutout{
We write $\recur$ for the wide subcategory of $\Scatinn$ containing  recursive S-strategies.
Above we used the fact that recursive S-strategies are closed under composition.
Since all our language constructors preserve recursiveness, $\recur$ is a sound model of $\iacbv$.
We therefore encode an innocent S-strategy $\sigma$ by $\code{\sigma}$ where the latter is seen indifferently as the index $n$ such that either $\sigma=W_n$ or $\viewf(\sigma)=\phi_n$. 
}

We want to show that any recursively presentable innocent S-strategy is definable by an $\iacbv$-term. The result will be proved by constructing a term that accepts the code of a given strategy, examines its initial behaviour, mimics it and, after a subsequent O-move, is ready to explore the relevant component of the decomposition. 
Observe that if we start from a strategy on $\sem{\theta_1,\cdots,\theta_k\vdash\theta}$ the decomposition will lead us to consider strategies on $\sem{\theta_1',\cdots,\theta_l'\vdash\theta'}$, where each $\theta_i'$ (as well as $\theta'$) is a subtype of some $\theta_j$ or $\theta$. Since the given strategy will in general be infinite, repeated applications of the Decomposition Lemma will mean that $l$ is unbounded. To keep track of the current component we will thus need to be able to represent unbounded lists of variables whose types are subtypes of $\theta_1,\cdots,\theta_k,\theta$. This issue is tackled next.

\paragraph{\bf List contexts} 
We say that a set of types $T$ is \boldemph{closed} if whenever $\theta\in T$ and $\delta$ is a subtype of $\theta$ then $\delta\in T$. 
For the rest of this section let us fix a closed finite set of types $T$ and an ordering of $T$, say $T=T_0,T_1,\dots,T_n$, such that 
$T_0=\comt$, $T_1=\int$, $T_2=\var$,
$T_3=\comt\to\expt$ and $T_4=\expt\to\comt$.

For each $i$, we encode lists of type $T_i$ as products $\int\times(\int\arr T_i)$. In particular, we use the notation
\[ z:\List{T_i},\,\Gamma \vdash M:\delta \]
as a shorthand for
\[ z^L:\int,\,z^R:\int\arr{T_i},\,\Gamma \vdash M:\delta \]
Thus, $z^L$ represents the length of the represented list. 
For each $1\leq i\leq z^L$, the value of the $i$-th element in the list is represented by $z^Ri$. The list can be shortened by simply `reducing' $z^L$. For example, for a term $z:\List{T_i},\,\Gamma \vdash M:\delta$ we can form
\[ z:\List{T_i},\,\Gamma \vdash \letin{z^L=z^L{-}1}{M}:\delta. \]
Note that, although the notation seems to suggest differently, the above is {unrelated} to variable assignment: it stands for
$(\lambda z^L.M)(z^L{-}1)$. 
A finer removal of a list element is executed as follows.
For a term $z:\List{T_i},\,\Gamma \vdash M:\delta$ and an index $j$, we define the term
\[  z:\List{T_i},\,\Gamma \vdash \remove{z}{j}{M}:\delta \]
to be
\[  
z:\List{T_i},\,\Gamma \vdash (\lambda z^L.\lambda z^R.\,M)(z^L{-}1)(\lambda x.\ \cond{x<j}{z^Rx}{z^R(x{+}1)}):\delta. 
\]
A list can be extended as follows.
For terms $z:\List{T_i},\,\Gamma \vdash M:\delta,\,N:T_i$ and an index $j$ we define the term
\[  
z:\List{T_i},\,\Gamma \vdash \add{z}{j}{N}{M}:\delta 
\]
to be:
\[  
z:\List{T_i},\,\Gamma \vdash (\lambda z^L.\lambda z^R.\,M)(z^L{+}1)(\lambda x.\ 
\cond{x<j}{z^Rx}{%
\cond{x=j}{N}{z^R(x{-}1)}}):\delta 
\]
Let us use the shorthands \nt{
\[
\Hextend{z}{N}{M}\qquad
\Textend{z}{N}{M}
\]}%
for $\add{z}{1}{N}{M}$ 
and $\add{z}{z_L{+}1}{N}{M}$ respectively
(that is, $\sf Hextend$ inserts at the head of lists and $\sf Textend$ at the tail).

We can define an (effective) indexing function $\indx$ which, for any sequence (possibly with repetitions) $\theta_1,\dots,\theta_m$ of types from $T$, returns a pair of numbers $(i,j)$ such that
$\theta_m=T_i$ and there are $j$ occurrences of $T_i$ in $\theta_1,\dots,\theta_m$.

Suppose now we have such a sequence $\overline \theta$ and a term $z_0:\List{T_0},\,z_1:\List{T_1},\,\dots,\, z_n:\List{T_n} \vdash M:\delta$. We can \emph{de-index} $M$ with respect to $\overline \theta$, obtaining the term $x_1:\theta_1,\dots,x_m:\theta_m \vdash \deindx{M}:\delta$, defined as
\[ \deindx{M} \defn \letin{\overline{z=\bot}}{(\Hextend{z_{l_m}}{x_m}{(\dots\ (\Hextend{z_{l_1}}{x_1}{M}))})}, \]
where 
\[ \letin{\overline{z=\bot}}{N}\defn \letin{z_{0}^L=0,z_{0}^R=\lambda x.\Omega}{(\dots\ (\letin{z_{n}^L=0,z_{n}^R=\lambda x.\Omega}{N}))} \]
and, for each $1\leq i\leq m$, $\theta_i=T_{l_i}$.
Note that the extensions above are executed from left to right so, in particular, $x_1$ will be related to $z_{l_1}^R\,1$.

\paragraph{\bf Universal terms}
Given a closed set of types $T$, the way we prove universality is by constructing for each $\delta\in T$ a \emph{universal term} $z_0:\List{T_0},\dots,z_n:\List{T_n}\vdash F_\delta:\int\arr\delta$ such that,
for every sequence $\theta_1,\dots,\theta_m$ from $T$ and recursively presentable S-strategy $\sigma:\sem{\theta_1}\otimes\cdots\otimes\sem{\theta_m}\arr\sem{\delta}$,
\[ \sigma = \ssem{\deindx (F_\delta\code{\sigma})}\,. \]
We first need to make sure that we can move inside the Decomposition Lemma effectively, i.e.\ that the passage from the code of the original strategy to the code of the components is effective and that the case which applies can also be computed from the index of the original strategy. Here is such a recursive version of the Decomposition Lemma (for types in $T$).\newpage

\begin{lemma}\label{l:DLDL}
There are partial recursive functions
\[ D,H:\omega\rightharpoonup\omega\text{ and }B:\omega\times\omega\rightharpoonup \omega \]
such that, for any $\theta_1,\dots,\theta_m,\delta\in T$ and recursively presentable S-strategy $\sigma:\sem{\theta_1}\otimes\cdots\otimes\sem{\theta_m}\arr\sem{\delta}$,
\begin{align*}
D\code{\sigma}&=\begin{cases}
\,i &\text{ if $\sigma$ falls within the $i$-th case of DL}\\ \bot &\text{ otherwise}
\end{cases}
\\
B(\code{\sigma},i) &= \begin{cases}
\code{\tau_i} &\text{ if $\sigma$ and $\tau_i$ are related as in first case of DL}\\
 \bot &\text{ otherwise}
\end{cases}\\
H\code{\sigma} &=\begin{cases}
\,i &\text{ if $\sigma,i$ are related as in third case of DL}\\
\code{\code{\sigma_1},\code{\sigma_2}} &\text{ if $\sigma,\sigma_1,\sigma_2$ are related as in fourth case of DL}\\
\code{\sigma'} &\text{ if $\sigma,\sigma'$ are related as in fifth case of DL}\\
\code{i,\code{\sigma'}} &\text{ if $\sigma,\sem{\theta_l},\sigma'$ are related as in sixth case of DL}\\
&\text{ and }
\indx(\theta_1,...\,,\theta_{l})=(2,i)
\\	
\code{i_1,i_2,\code{\sigma'},\code{\sigma''}} &\text{ if $\sigma,\sem{\theta_l},\sigma',\sigma''$ are related as in seventh case of DL}\\
&\text{ and }\indx(\theta_1,...\,,\theta_l)=(i_1,i_2)\\	
\code{\code{\sigma'},\code{\sigma''}} &\text{ if $\sigma,\sigma',\sigma''$ are related as in eighth case of DL}
\end{cases}
\end{align*}
\end{lemma}
\begin{proof}
{We assume that the type of $\sigma$ is represented in $\code{\sigma}$ and can be effectively decoded by $D,B$ and $H$.}
$D\code{\sigma}$ returns 1 if any of the $\theta_i$'s is $\int$, otherwise it applies $\phi_{\code{\sigma}}$ to the unique initial move of $\sem{\overline{\theta}}$ and returns the number corresponding to the result. For $B$, given $\code{\sigma},i$, membership in $\tau_i$ is checked as follows. For any (P-view) S-play $s$, we add $i$ to its initial move and check whether the resulting S-play is a member of $\sigma$. Thus we obtain $\phi_n$ such that $s\mapsto\phi_n(s,\code{\sigma},i)$ is the characteristic function of $\tau_i$. By an application of the S-m-n theorem we obtain $\code{\tau_i}$. For $H$ we argue along the same lines. 
\end{proof}

Since (call-by-value) $\pcf$ is Turing complete, there are closed $\pcf$-terms $\tilde D,\tilde H:\int\arr\int$ and $\tilde B:\int\arr\int\arr\int$ that represent each of the above functions with plays of the form \nt{$q\,{*}\,n\,f(n)$} or $q\,{*}\,m\,{*}\,n\,f(m,n)$.
The terms will be used inside the universal term, which will be constructed by mutual recursion (there are standard techniques to recast such definitions in $\pcf$).
Let us write $\theta=T_{l(\theta)}\arr T_{r(\theta)}$ whenever $\theta\in T$ is of arrow type (so $l,r:T\arr\{0,\dots,n\}$).\newpage

\begin{definition}
For each $\delta\in T$ we define  terms 
\[
z_0:\List{T_0},\dots,z_n:\List{T_n}\vdash F_\delta:\int\arr\delta
\]
by mutual recursion as follows.
\begin{align*}
  F_\delta \defn \lambda k^\int.\ &\mathsf{if}\ z_{1}^L\neq0\ \mathsf{then}\
  \letin{x=z_1^R1}{\remove{z_1}{1}{F_\delta(\tilde{B}\,k\,x)}}\\
  &\mathsf{else}\ 
    \mathsf{case}\ (\tilde{D}\,k)\ \mathsf{of} \\
  &\qquad2:\ \sskip\\
  &\qquad3:\ \tilde{H}k\\
  &\qquad4:\ \letin{\code{k_1,k_2}=\tilde{H}k}{\mkvar(F_{\comt\arr\int}\,k_1,F_{\int\arr\comt}\,k_2)}\\
  &\qquad5:\ \lambda y^{T_{l(\delta)}}.\,\Textend{z_{l(\delta)}}{y}{F_{T_{r(\delta)}}(\tilde{H}\,k)}\\
  &\qquad6:\ \letin{\code{i,k}=\tilde{H}k}{}\\ 
  &\qquad\qquad\Textend{z_3}{\lambda x^\comt.\,!(z_2^Ri)}{}\\[-1mm]
  &\qquad\qquad\;\; \Textend{z_4}{\lambda x^\int.\,(z_2^Ri)\aasg x}{\remove{z_2}{i}{F_\delta\,k}}\\
  &\qquad7:\ \letin{\code{i_1,i_2,k_1,k_2}=\tilde{H}k}{\mathsf{case}\ i_1\ \mathsf{of}}\\
  &\qquad\qquad1:\  \dots \\[-2.5mm]
  &\qquad\qquad\,\vdots \\[-1.5mm]
  &\qquad\qquad j:\ \Textend{z_{r(T_j)}}{(z_{j}^R\,i_2)(F_{T_{l(T_j)}}k_1)}{F_\delta\,k_2}\\[-2.5mm]
  &\qquad\qquad\,\vdots \\[-1.5mm]
  &\qquad\qquad n:\ \dots\\
  &\qquad8:\ \letin{\code{k_1,k_2}=\tilde{H}k}{}\\
  &\qquad\qquad\Hextend{z_1}{(\new{x}{\Hextend{z_4}{x}{F_\int\,k_1}})}{F_\delta\,k_2} \\
  &\qquad\mathsf{otherwise}:\ \Omega
\end{align*}
\end{definition}

\noindent The construction of $F_\delta$ follows closely the decomposition of $\sigma$ according to 
the Decomposition Lemma. In particular, on receiving $\code{\sigma}$, the term decides,
using the functions $B$ and $D$ of Lemma~\ref{l:DLDL}, to which branch of DL 
$\sigma$ can be matched.
Some branches decompose $\sigma$ into further strategies,
in which case $F_\delta$ will recursively call some $F_{\delta'}$ to simulate the rest of the strategy.
The use of lists in contexts guarantees that such a call is indeed recursive: $F$ is only parameterised by the output type $\delta'$, and each such $\delta'$ is in $T$.
\begin{proposition}[Universality]\label{prop:univ}
For every $\theta_1,\dots,\theta_m,\delta\in T$ and recursively presentable innocent S-strategy $\sigma:\sem{\theta_1}\otimes\cdots\otimes\sem{\theta_m}\arr\sem{\delta}$\,,\,
$\sigma = \ssem{\deindx (F_\delta\code{\sigma})}$\,.
\end{proposition} 
\begin{proof}
Suppose $F_\delta$ receives $\code{\sigma}$ in its input $k$, where $\sigma:\sem{\theta_{1}}\otimes\cdots\otimes\sem{\theta_m}\arr\sem\delta$. Then
\[ 
\mathsf{if}\ z_{1}^L\neq0\ \mathsf{then}\
  \letin{x=z_1^R1}{\remove{z_1}{1}{F_\delta(\tilde{B}\,k\,x)}}
%
\]
recognizes the first branch of the DL. Recall that $T_1=\int$, so $z_1:\List{\int}$, and therefore $z_{1}^L\neq0$ holds iff there is some $\theta_i=\int$. If this is so, then $F_\delta$ needs to return a term corresponding to the strategy instantiated with the leftmost element in the list $z_1$.  This is achieved by first applying $\tilde B$ to $k,(z_1^R1)$ to obtain $\code{\tau_i}$ 
and applying $F_\delta$ to it.

If $z_1^L=0$,  $F_\delta$ will call $\tilde D$ on $\code{\sigma}$, 
which will return the number of the case from DL  ({\bf 2}-{\bf 8}) that applies to $\sigma$.
Subsequently, $F_\delta$ will proceed to a case analysis. Below we examine two cases in detail. 
\begin{enumerate}[label={\bf\arabic*}:]
\item[{\bf 5}:] $\sigma$ is the currying of  $\sigma'$, so $F_\delta$ should return 
`$\lambda y.\, F\code{\sigma'(y)}$'. Now, $\sigma'$ is $F_{T_{r(\delta)}}\code{\sigma'}$, i.e.\ $F_{T_{r(\delta)}}(\tilde H\, k)$, where $\delta=T_{l(\delta)}\Arr T_{r(\delta)}$. In order to preserve typability, we need to add the abstracted variable to the context of $F_{T_{r(\delta)}}(\tilde H\, k)$, which is what $\mathsf{Textend}\,z_{l(\delta)}\,\mathsf{with}\,y$ achieves.
\item[{\bf 8}:] $\sigma$ introduces some fresh name and decomposes to $\sigma',\sigma''$
as in the Decomposition Lemma. Hence,  $F_\delta$ should return $`\letin{y=(\new{x}{F\code{\sigma'(x)}})}{F\code{\sigma''(y)}}$', which is exactly what the code achieves.
\end{enumerate}
The other cases are similar.
\end{proof}

\begin{remark}
It is worth noting that the universality result for innocent S-strategies
implies an analogous result for innocent strategies and $\pcf$.
Thanks to call-by-value, the result is actually sharper than the universality results of~\cite{AJM00,HO00}, which had
to be proved ``up to observational equivalence". This was due to the fact that partial
recursive functions could not always be represented in the canonical way (i.e.\ by
terms for which the corresponding strategy contained plays of the form $q\, q\, n\, f(n)$).
This is no longer the case under the call-by-value
regime, where each partially recursive function $f$ can be coded by a term
whose denotation will be the strategy based on plays of the shape $n\, f(n)$.
\end{remark}


\section{From omniscience to innocence}

In Section~\ref{sec:syntax} we introduced the three languages: $\pcfplus$,
$\iacbv$ and $\rml$, interpreted respectively by innocent, block-innocent and knowing 
 strategies. Let $A$ be a \nt{prearena}. We write $\inno{A}$, $\binno{A}$ and $\knowing{A}$
for the corresponding classes of (store-free) strategies  in $A$.
Obviously, $\inno{A}\subseteq \binno{A}\subseteq \knowing{A}$.
 Next we shall study type-theoretic conditions under which
 one kind of strategy collapses to another. Thanks to universality results,
 this corresponds to the existence of an equivalent program in a weaker language.
 \begin{thm}\label{lem:bk}
Let $A=\sem{\seq{\theta_1, \cdots, \theta_n}{\theta\rarr\theta'}}$.
Then $\binno{A}\subsetneq\knowing{A}$.
\end{thm}
\begin{proof}
Observe that there exist moves $q_0, a_0, q_1, a_1$ such that
$q_0\vdash_A a_0\vdash_A q_1\vdash_A a_1$ and consider
$\sigma=\makeset{\epsilon, q_0 a_0, q_0 a_0 q_1 a_1}$, i.e. $\sigma$
has no response at $q_0 a_0 q_1 a_1 q_1$. Then $\sigma\in \knowing{A}\setminus\binno{A}$.
It is worth remarking that a strategy of the above kind denotes the $\rml$-term
$\seq{}{\letin{v=\newc}{\lambda x^\comt. (\cond{!v}{\Omega}{v\aasg !v+1}}):\comt\rarr\comt}$.
\end{proof}
Theorem~\ref{lem:bk} confirms that, in general, block structure restricts expressivity.
However,  the next result shows this not to be the case for open terms of base type.
\begin{thm}
Let  $A=\sem{\seq{\theta_1, \cdots, \theta_n}{\beta}}$.
Then $\binno{A}=\knowing{A}$.
\end{thm}
\begin{proof}
Observe that any knowing strategy for $A$ becomes block-innocent if
in the second-move P introduces a store with one variable that keeps track of the history of
play (this is reminiscent of the factorization arguments in game semantics).
The variable should be removed from the store by P only when he plays an
answer to the initial question.
\end{proof}
By universality, we can conclude that each $\rml$-term of base type
is equivalent to an $\iacbv$-term. Since contexts used for
testing equivalence are exactly of this kind, we obtain the following corollaries.
The first one amounts to saying that $\rml$ is a conservative extension of $\iacbv$.
The second one states that block-structured contexts suffice to distinguish
terms that might use scope extrusion.

\begin{corollary}\label{cor:conservativity}
For any $\iacbv$-terms $\seq{\Gamma}{M_1,M_2}:\theta$ and $\rml$-terms $\seq{\Gamma}{N_1,N_2}:\theta$:
\begin{itemize}
\item $\seq{\Gamma}{M_1\eqmod{\rml}M_2}$
if, and only if, $\seq{\Gamma}{M_1\eqmod{\iacbv}M_2}$;
\item $\seq{\Gamma}{N_1\eqmod{\rml}N_2}$
if, and only if, $\seq{\Gamma}{N_1\eqmod{\iacbv}N_2}$.
\end{itemize}\qed
\end{corollary}
Now we investigate the boundary between block structure
and lack of state.
\begin{lemma}\label{lem:binno-inno}
Let $A$ be a \nt{prearena} such that each question enables an answer
\footnote{All denotable \nt{prearenas} enjoy this property.}.
The following conditions are equivalent.
\begin{enumerate}
\item $\binno{A}\subseteq\inno{A}$.
\item No O-question is enabled by a P-question:
$m\vdash_A q_O$ implies $\lambda_A(m)=PA$.
\item Store content of O-questions is trivial:
$s q_O^\Sigma\in\Splays{A}$ implies $\dom{\Sigma}=\emptyset$.
\end{enumerate}
\end{lemma}
\begin{proof}\hfill
\begin{itemize}[label=$(1\Rarr 2)$]
\item[$(1\Rarr 2)$]  We prove the contrapositive. Assume 
that there exists a P-question $q_P$ and an O-question $q_O$ 
such that $q_P\vdash_A q_O$.
Let $s$ be a chain of hereditary enablers of $q_P$ (starting from an initial move) 
augmented with pointers from non-initial moves to the respective preceding moves.
Then 
\[ \]
\[
\rnode{A}{s}\;\;\
\rnode{C}{\rnode{B}{q}_P^{\,(x,0)}}\
\rnode{E}{\rnode{D}{q}_O^{\,(x,0)}}\
\rnode{G}{\rnode{F}{a}_P^{\,(x,0)}}\
\rnode{I}{\rnode{H}{q}_O^{\,(x,1)}}\
\rnode{K}{\rnode{J}{q}_P^{\,(x,1)}}
\justf[,nodesep=1mm]{B}{A}\justf[,angleA=90,angleB=75,nodesep=1mm]{D}{B}
\justf[,angleA=90,angleB=75,nodesep=1mm]{F}{D}
\justf[,angleB=70,nodesep=1mm]{H}{B}
\justf[,angleA=90,angleB=75,nodesep=1mm]{J}{H}
\]
defines a block-innocent strategy that is not innocent.

\item[$(2\Rarr 3)$] Suppose no P-question enables an O-question in $A$
and let $s q_O^X\in\splay{A}$.
Then the sequence of hereditary justifiers of $q_O$ in $s$,
in order of their occurrence in $s$,  must have the form $(q_O a_P)^\ast$.
Consequently, none of the stores involved can be non-empty, so $X$
must be empty too.
\item[$(3\Rarr 1)$]
We observe that
$s_1 p^{X_p} s_2 o^{X_o}\in \splay{A}$, where $p$ justifies $o$,  
implies $X_o=X_p$. Note that, in presence of block innocence, this
implies innocence because the store content of O-moves 
can be reconstructed uniquely from the P-view.
Thus, it suffices to prove our observation correct.
\begin{itemize}
\item If $o$ is a question, we simply use our assumption: then we must have
$X_o=\emptyset$ and, because $\dom{X_o}=\dom{X_p}$, we can conclude
$X_p=\emptyset$.
\item If $o$ is an answer then $p$ must be a question. We claim that 
no store in $s_2$ can contain variables from $\dom{X_o}=\dom{X_p}$.
Suppose this is not the case and there is such an occurrence.
Then the earliest such occurrence must be part of an O-move.
This move cannot be a question due to our assumption, so 
it is an answer move. By the bracketing condition, this must be an 
answer that an earlier P-question played after $p$. Moreover, the store
accompanying that question must also contain a variable from $\dom{X_o}=\dom{X_p}$
contradicting our choice of the earliest occurrence.\qedhere
\end{itemize}
\end{itemize}
\end{proof}

\noindent Thanks to the following lemma we will be able to determine precisely at which 
types  block-innocence implies innocence.
\begin{lemma}
$\sem{\seq{\theta_1,\cdots,\theta_n}{\theta}}$ satisfies condition 2.
of Lemma~\ref{lem:binno-inno} iff $\ord{\theta_i}\le 1$ ($i=1,\cdots,n$)
and $\ord{\theta}\le 2$.
\end{lemma}
Consequently, second-order $\iacbv$-terms always have
purely functional equivalents.
Finally, we can pinpoint the types at which strategies are bound to be innocent:
it suffices to combine the previous findings.
\begin{thm}
Let $A=\sem{\seq{\theta_1,\cdots,\theta_n}{\theta}}$. Then 
$\knowing{A}=\inno{A}$ iff $\ord{\theta_i}\le 1$ ($i=1,\cdots,n$) and $\ord{\theta}=0$.
\end{thm}
In the next section we demonstrate that the gap in expressivity between $\knowing{A}$
and $\binno{A}$ also bears practical consequences.  The undecidable equivalence
problem for second-order finitary $\rml$ becomes decidable in second-order finitary $\iacbv$
(as well as at some third-order types).



\section{Decidability of a finitary fragment of \texorpdfstring{$\iacbv$}{IAcbv}}

In order to prove program equivalence decidable, we restrict the base
datatype of integers to the finite segment $\makeset{0,\cdots,N}$ ($N>0$)
and replace recursive definitions ($\fix{M}$) with looping ($\while{M}{N}$).
Let us call the resultant language $\ialoop$.
Our decidability result will hold for a subset $\iatwo$ of $\ialoop$, in which
type order is restricted. $\iatwo$ will reside inside the third-order fragment of $\ialoop$
and contain its second-order fragment. Note that the second-order fragment of 
similarly restricted $\rml$ is known be undecidable (even without loops)~\cite{Mur04b}.

The decidability of program equivalence in $\iatwo$ will be shown by translating terms
to regular languages representing the corresponding \emph{block-innocent} strategies.
We stress that we are \emph{not} going to work with the induced S-plays.
Nevertheless, the translation will rely crucially on insights
gleaned from the semantics with explicit stores. In particular, we shall take advantage of the uniformity inherent in block innocence
to represent only subsets of the strategies in order to overcome technical problems presented by pointers.
We discuss the issue next.

\subsection{Pointer-related issues}

Pointers from answer-moves need not be represented at all, because they are uniquely reconstructible through the well-bracketing condition.
However, this need not apply to pointers from questions. The most obvious way to represent them is to decorate moves
with integers that encode the distance from the target in some way. Unfortunately, there are scenarios in which the distance can grow arbitrarily.

\cutout{
Next we analyse two typing scenarios that look hopeless from the point of view of encoding
pointers, since the distance from the pointer can grow arbitrarily. In the first case, thanks to
block-innocence, we will be able to overcome the difficulties. The other case
must remain a challenge for future work (or an undecidability result). On the basis of 
our discussion we shall subsequently introduce the type system of $\iatwo$.
}
Consider, for instance,  the \nt{prearena} $A=\sem{\seq{\theta}{\theta_1\rarr\ldots\rarr \theta_k\rarr\beta}}$.
Due to the presence of the $k$ arrows on the right-hand side we obtain chains of enablers
$q_0\vdash a_0 \vdash \cdots \vdash q_k\vdash a_k$, where $q_0$ is initial and 
each $q_i$ ($i=1,\cdots,k$) is initial in $\sem{\theta_i}$. We shall call the moves \emph{spinal}.
Observe that plays in $A$ can have the following shape
\[
q_0\cdots a_0 q_1 \cdots a_1 q_1\cdots a_1 q_1\cdots a_1 q_2
\]
and any of the occurrences of $a_1$ could be used to justify $q_2$, thus creating several different options for justifying $q_2$.
If we consider S-plays for $A$, Definition~\ref{def:splay} implies that none of the moves $q_i,a_i$ will ever carry a non-empty store.
Consequently, whenever a play of the above kind comes from a block-innocent strategy, its behaviour in  the $q_1\cdots a_1$ segments
will not depend on that in the other $q_1\cdots a_1$ segments. Thus, in order to explore exhaustively the range of behaviours offered by a block-innocent strategy
(so as to compare them reliably), it suffices to restrict the number of $q_1$'s  to $1$.
Next, under the assumption that $q_1$ occurs only once, one can repeat the same argument for $q_2$ to conclude
that a single occurrence of $q_2$ will suffice, and so on.
Altogether this yields the following lemma. Note that, due to Visibility,  
insisting on the presence of a unique copy of $q_1,\cdots, q_k$ in a play amounts to 
asking that each $q_i$ be preceded by $a_{i-1}$.
\begin{lemma}
Call a play \emph{spinal} if each spinal question $q_i$ ($0< i \le k$) occurring in it 
is the immediate successor of $a_{i-1}$.
Let $\spinal{A}$ be the set of spinal plays of $A$.
Let $\sigma,\tau:A$ be block-innocent strategies.
Then $\sigma\cap \spinal{A} = \tau\cap\spinal{A}$ implies $\sigma=\tau$.\qed
\end{lemma}
Hence, for the purpose of checking program equivalence, it suffices to compare
the induced sets of \emph{spinal} complete plays. Moreover, the pointer-related problems discussed above will not arise.

Now that we have dealt with one challenge, let us introduce another one, which
cannot be overcome so easily. Consider the \nt{prearena} 
$\sem{\seq{(\theta_1\rarr\theta_2\rarr\theta_3)\rarr\theta_4}{\theta}}$ and the enabling 
sequence $q_0\vdash q_1 \vdash q_2 \vdash a_2 \vdash q_3$ it contains. Now
consider the plays $q_0 q_1 (q_2 a_2)^j q_3$, where $j\ge 0$. Again, to represent
the pointer from $q_3$ to one of the $j$ occurrences of $a_2$, one would need an unbounded number of indices.
This time it is not sufficient to restrict $j$ to $1$, because the behaviour need not be uniform
after each $q_2$ (this is because in the setting with stores a non-empty store
can be introduced as soon as in the second move $q_1$). To see that the concern is real,
consider the term 
$
\seq{f:(\comt\rarr\comt\rarr\comt)\rarr\comt}{\new{x}{f(\lambda y^\comt.\cdots\lambda z^\comt.\cdots):\comt}}
$,
where $(\cdots)$ contain some code inspecting and changing the value of $x$.

This leads us to introduce $\iatwo$ via a type system that will not generate the configuration
just discussed. Another restriction is to omit third-order types in the context, as
they lead beyond the realm of regular languages 
(cf. $\seq{f:((\comt\rarr\comt)\rarr\comt)\rarr\comt}{f(\lambda g^{\comt\rarr\comt}.g()}$).
Since $\vart$ leads to identical problems as $\comt\rarr\comt$, we restrict its use
accordingly.

\subsection{\texorpdfstring{$\iatwo$}{IA2+}}

\begin{definition}
$\iatwo$ consists of $\ialoop$-terms whose typing derivations rely solely on typing
judgments of the shape
$
\seq{x_1:\ctype_1,\cdots,x_n:\ctype_n}{M:\ttype}
$,
where $\ctype$ and $\ttype$ are defined by the grammar below.
\[\begin{array}{rll}
\ctype &::=\quad &\beta \quad|\quad \vart \quad|\quad \beta\rarr\ctype \quad|\quad \vart\rarr\ctype \quad|\quad (\beta\rarr\beta)\rarr\ctype\\
\ttype  &::= &\beta\quad|\quad\vart \quad|\quad \ctype\rarr\ttype 
\end{array}\]
\end{definition}\medskip

\noindent A lot of pointers from questions become uniquely determined in strategies
representing $\iatwo$ terms, namely, all pointers from any O-questions and all
pointers from $P$-questions to O-questions.
\begin{lemma}\label{lem:pointers}
Let $A=\sem{\seq{\ctype_1,\cdots,\ctype_n}{\ttype}}$ and $s_1$, $s_2$ be spinal plays of $A$
that are equal after all pointers from O-questions and all pointers from P-questions
to O-questions have been erased. Then $s_1=s_2$.
\end{lemma}
\begin{proof}
Observe that whenever a P-question 
is enabled by an O-question in the \nt{prearenas} under consideration, the O-question must be spinal. Hence,
because both $s_1$ and $s_2$ are spinal, all such O-questions will
occur only once, so pointers from P-questions to O-questions are 
uniquely reconstructible.

Now let us consider O-questions. Observe that, due to restrictions on the type
system of $\iatwo$, whenever an O-question is justified by a P-answer, both will be
spinal. Hence only one copy of each can occur in a spinal position, making pointer
reconstruction unambiguous. Finally, we tackle the case of O-questions justified
by P-questions.

\begin{itemize}
\item If $q$ comes from a type of the context then, due to the shape of types involved,
any sequence of hereditary enablers of $q$ must be of the form 
$q'  q_1' a_1'\cdots q_j' a_j' q_{j+1}' q$, where $q'$
is initial and each of the moves listed enables the following one.
If a move $m$ enables $q$ hereditarily, let us define its degree
as the distance from the initial move in the sequence above
(this definition is independent of the actual choice of the chain
of enablers; the degree of $q_i'$ is $2i-1$,  that of $a_i'$ is $2i$).

By induction on move-degree we show that in any O-view
only one move of a given degree can be present, if at all.
By visibility this makes pointer reconstruction unambiguous. 
\begin{itemize}
\item By definition of a play $q'$ can occur in an O-view only once.
Whenever $q_1'$ is present in an O-view, it must preceded by an initial move,
so its  position in an O-view is uniquely determined (always second).
\item Since $q_i'$ occurs only once in an O-view, so does $a_i'$ (questions can be
answered only once). Hence, $q_{i+1}'$ can also occur only once, because each
occurrence must be preceded by $a_i'$.
\end{itemize}
Consequently, any move of degree $2j+1$ ($q_{j+1}'$) 
can only occur once in an O-view and thus the pointer from $q$
can be reconstructed uniquely.
\item If $q$ originated from the type on the right-hand side of the typing
judgment, we can repeat the reasoning above. The only difference is
that the sequences of enablers are now of the form
\[
q_0 a_0 \cdots q_k a_k q' q_1' a_1' \cdots q_j' a_j' q_{j+1}' q
\]
where $q_i, a_i$ ($i=0,\cdots,k$) are spinal. Then the base case
of the induction ($q'$) follows from the fact that we are dealing with
spinal plays.\qedhere
\end{itemize}
\end{proof}

\noindent Thus, the only pointers that need to be accounted for are those from P-questions
to O-answers. Here is the simplest scenario illustrating that they can be ambiguous.
Consider the terms 
\[
\seq{f:\comt\rarr\comt\rarr\comt}{\letin{g_1=f()}{(\letin{g_2=f()}{g_i()})}:\comt}
\]
where $i=1,2$.
They lead to the following plays, respectively for $i=1$ and $i=2$, which are equal up to 
pointers from P-questions to O-answers.

\medskip

\[
\rnode{A}{q_0}\qwe\rnode{B}{q_1}\qwe\rnode{C}{a_1}\qwe\rnode{D}{q_1}\qwe\rnode{E}{a_1}\qwe\rnode{F}{q_2}
\justf{B}{A}\justf{C}{B}\justf{D}{A}\justf{E}{D}\justf{F}{C}\qquad
\rnode{A}{q_0}\qwe\rnode{B}{q_1}\qwe\rnode{C}{a_1}\qwe\rnode{D}{q_1}\qwe\rnode{E}{a_1}\qwe\rnode{F}{q_2}
\justf{B}{A}\justf{C}{B}\justf{D}{A}\justf{E}{D}\justf{F}{E}
\]
\begin{remark}{
In the conference version of the paper~\cite{MT10} we suggested that justification pointers of the above kind be represented
with numerical indices encoding the target of the pointer inside the current P-view.
More precisely, one could enumerate (starting from $0$) all question-enabling O-answers in the P-view. Then pointers from P-questions
to O-answers could be encoded by decorating the P-question with the index of the O-answer. The plays above 
could then be coded as $q_0 q_1 a_1 q_1 a_1 q_2^0$ and $q_0 q_1 a_1 q_1 a_1 q_2^1$ respectively.
Unfortunately, there exists terms that  generate plays where such indices would be unbounded, such as
\[
{\letin{g_1=f()}{(\while{h()}{\letin{g_2=f()}{g_2()});g_1():\comt}}}
\]
where $f:\comt\rarr\comt\rarr\comt$ and $h:\comt\rarr\expt$. Because the number of loop iterations is unrestricted, the number needed to represent the
justification pointer corresponding to the rightmost occurrence of $g_1()$ cannot be bounded. Consequently, the representation scheme proposed in~\cite{MT10}
would lead to an infinite alphabet. Next we show that this problem can be overcome, though.
}
\end{remark}
The above-mentioned defect can be patched with the help of a different representation scheme based on annotating targets and sources of 
justification pointers, with $\circ$ and $\bullet$ respectively. We shall use the same two symbols for each pointer.
This can lead to ambiguities if many pointers are represented at the same time in a single string.  However, to avoid that, we are going to use
multiple strings to represent a single play. More precisely, these will be strings corresponding to the underlying sequence of moves
in which each pointer may or may not be represented, i.e. 
plays featuring $n$ pointers from P-questions to O-answers will be represented by $2^n$ encoded strings. For example, to represent

\medskip

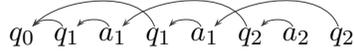
\[
\rnode{A}{q_0}\qwe\rnode{B}{q_1}\qwe\rnode{C}{a_1}\qwe\rnode{D}{q_1}\qwe\rnode{E}{a_1}\qwe\rnode{F}{q_2}\qwe \rnode{G}{a_2} \qwe \rnode{H}{q_2}
\justf{B}{A}\justf{C}{B}\justf{D}{A}\justf{E}{D}\justf{F}{C}\justf{G}{F}\justf{H}{E}
\]
we shall use the following four strings
\[\begin{array}{ccc}
q_0 \, q_1\, a_1\, q_1\, a_1\, q_2\, a_2\, q_2 &\qquad & q_0\,  q_1\, a_1^\circ\, q_1\, a_1\, q_2^\bullet\, a_2\, q_2\\
q_0\,  q_1\, a_1\, q_1\, a_1^\circ\, q_2\, a_2\, q_2^\bullet && q_0\,  q_1\, a_1^\circ\, q_1\, a_1^\circ\, q_2^\bullet\, a_2\, q_2^\bullet
\end{array}\]
Note that the last string represents pointers ambiguously, though the second and third strings identify them uniquely. 
We could have achieved the same effect using strings  in which only one pointer is represented~\cite{HMO11} but, from the technical point of view, it is easier to include 
all possible pointer/no-pointer combinations.

\subsection{Regular-language interpretation}

In order to translate $\ialoop$-terms into regular languages representing their game semantics, we 
we restrict our translation to terms in a canonical shape, to be defined next. 
Any $\ialoop$-term can be converted effectively to such a form and the conversion
preserves denotation.

The canonical forms are defined by the following grammar. We use types as 
superscripts, whenever we want to highlight the type of an identifier ($u,v,x,y,z$ 
range over identifier names). Note that the only identifiers in canonical form are those of base type, 
represented by  $x^\beta$ below.
\[\begin{array}{ll}
\can ::=   & ()  \,\,\,|\,\,\, i \,\,\,|\,\,\, x^\beta \,\,\,|\,\,\, x^\beta\oplus y^\beta \,\,\,|\,\,\, \cond{x^\beta}{\can}{\can} \,\,\,|\,\,\, 
x^\vart\aasg y^\expt \,\,\,|\,\,\, !x^\vart \,\,\,|\,\,\,\lambda x^\theta.\can \,\,\,|\,\,\,\\
 & \badvar{\lambda x^\comt.\can}{\lambda y^\expt.\can} \,\,\,|\,\,\,  \new{x^\vart}{\can} \,\,\,|\,\,\, \while{\can}{\can}  \,\,\,|\,\,\, \letin{x^\beta=\can}{\can} \,\,\,|\,\,\,\\ 
& \letin{x= z y^\beta}{\can} \,\,\,|\,\,\, 
\letin{x= z\, \badvar{\lambda u^\comt.\can}{\lambda v^\expt.\can}}{\can} \,\,\,|\,\,\,
\letin{x=z(\lambda x^\theta.\can)}{\can}
\end{array}\]
\begin{lemma}\label{lem:canon}
Let $\seq{\Gamma}{M:\theta}$ be an $\ialoop$-term. There is
an $\ialoop$-term $\seq{\Gamma}{N:\theta}$ in canonical form, effectively
constructible from $M$, such that $\sem{\seq{\Gamma}{M}}=\sem{\seq{\Gamma}{N}}$.
\end{lemma}
\begin{proof}
$N$ can be obtained via a series of $\eta$-expansions, $\beta$-reductions and
commuting conversions involving $\mathsf{let}$ and $\mathsf{if}$. We present a detailed argument in Appendix~\ref{apx:canon}.
\end{proof}
A useful feature of the canonical form is that the problems with pointers 
can be related to the syntactic shape: they concern references to $\mathsf{let}$-bound identifiers $x^\theta$
such that $\theta$ is \emph{not} a base type (i.e. $\theta=\vart$ or $\theta$ is a function type).
\cutout{The representation scheme for pointers corresponds then to enumerating such $\mathsf{let}$ bindings
along branches of the syntactic tree of the canonical form (using $0$ for topmost bindings).}
Below we state our representability theorem for $\iatwo$-terms.
The definition of $\clg{A}_M$ is actually too generous, as we shall only need $\circ,\bullet$ to
decorate P-questions enabled by O-answers.
\begin{proposition}\label{prop:regular}
Suppose $\seq{\Gamma}{M:\theta}$ is an $\iatwo$-term. Let $\clg{A}_M=M_A+(M_A\times\{\circ,\bullet\})$, 
where $A=\sem{\seq{\Gamma}{\theta}}$.
Let $\comp{\seq{\Gamma}{M}}$ be the set of non-empty spinal complete plays from $\sem{\seq{\Gamma}{M:\theta}}$.
Then $\comp{\seq{\Gamma}{M}}$ can be represented as a regular language over a \emph{finite} subset of $\clg{A}_M$.
\end{proposition}
\begin{proof}
$\comp{\seq{\Gamma}{M}}$ can be decomposed as $\sum_{i\in I_A}(i\,\, \comp{\seq{\Gamma}{M}}^i)$. Obviously $\comp{\seq{\Gamma}{M}}$
is regular if, and only if, any $\comp{\seq{\Gamma}{M}}^i$ is regular ($i\in I_A$). Hence, it suffices to show
that $\comp{\seq{\Gamma}{M}}^i$ is regular for any relevant $i$.

As already discussed, the regular-language representations, which we shall also refer to by $\compi{\seq{\Gamma}{M}}{i}$,
will consist of plays in which individual pointers may or may not be represented. However, because all possibilities are covered, 
this will yield a faithful representation of the induced complete plays.
We proceed by induction on the structure of canonical forms. Let $i_\Gamma$ range over $I_{\sem{\Gamma}}$. 
\begin{itemize}
\item $\compi{\seq{\Gamma}{()}}{i_\Gamma}=\star$
\item $\compi{\seq{\Gamma}{j}}{i_\Gamma}=j$
\item $\compi{\seq{\Gamma,x:\beta}{x:\beta}}{(i_\Gamma,j_x)}=j$
\item $\compi{\seq{\Gamma,x:\expt,y:\expt}{x\oplus y}}{(i_\Gamma, j_x, k_y)}=j\oplus k$
\item $\compi{\seq{\Gamma,x:\expt}{\cond{x}{M_1}{M_0}}}{(i_\Gamma,j_x)}=\compi{\seq{\Gamma,x}{M_{h(j)}}}{(i_\Gamma,j_x)}$, 
where $h(j)=\left\{\renewcommand\arraystretch{.7}\begin{array}{lcl} 0 &\,& j=0\\ 1 && j>0\end{array}\right.$
\item $\compi{\seq{\Gamma,x:\vart,y:\expt}{x\aasg y}}{(i_\Gamma,\star_x,j_y)}=\mwrite{j}_x\mok_x\star$
\item $\compi{\seq{\Gamma,x:\vart}{!x}}{(i_\Gamma,\star_x)}=\mread_x (\sum_{j=0}^{N} j_x j)$
\item $\compi{\seq{\Gamma}{\badvar{\lambda x^\comt.M_1}{\lambda y^\expt.M_2}}}{i_\Gamma}=
\star (\epsilon+\mread\,\compi{\seq{\Gamma,x}{M_1}}{(i_\Gamma,\star_x)}+\sum_{j=0}^N \mwrite{j}\, \compi{\seq{\Gamma,y}{M_2}}{(i_\Gamma,j_y)}[\mok/\star])$

\item $\compi{\seq{\Gamma}{\lambda x^\theta.M}}{i_\Gamma}=\star(\epsilon+\sum_{i\in I_{\sem{\theta}}} i\,\,\compi{\seq{\Gamma,x}{M}}{(i_\Gamma,i_x)}[m_x'/m_x]  )$

\item $\compi{\seq{\Gamma}{\new{x}{M}}}{i_\Gamma} = (\compi{\seq{\Gamma,x}{M}}{(i_\Gamma,\star_x)}\, \cap\, \clg{C}') [\epsilon/m_x]$ where, \nt{writing $\|$ for the shuffle operator on strings,} 
\[
\clg{C'}= (\clg{A}'\setminus (M_{\sem{\vart}})_x)^\ast\quad ||\quad ((\mread_x\,0_x)^\ast (\sum_{j=0}^N \mwrite{j}_x\,\mok_x\,(\mread_x\,j_x)^\ast  )^\ast)
\]
and $\clg{A}'$ is the finite alphabet used to represent $\seq{\Gamma,x:\vart}{M}$.
\end{itemize}
The substitution $[m_x'/m_x]$ highlights the fact that the moves associated with $x$ have to be bijectively relabelled, because
the copy of $\theta$ moved from the left- to the right-hand side of the context.
$[\epsilon/m_x]$ stands for erasure of all moves associated with $x$. Obviously, these (homomorphic) operations preserve regularity. 
Note that the clause for $\lambda x^\theta.M$ is correct because we consider spinal plays only.

For $\seq{\Gamma}{M:\beta}$ we sometimes refer
to components determined by the following decomposition.
\[
\compi{\seq{\Gamma}{M:\beta}}{i_\Gamma}=\sum_{j\in I_{\sem{\beta}}} (\compi{\seq{\Gamma}{M}}{i_\Gamma,j}\,\,j)
\]
The components $\compi{\seq{\Gamma}{M}}{i_\Gamma,j}$ can be extracted from $\compi{\seq{\Gamma}{M:\beta}}{i_\Gamma}$
by applying operations preserving regularity (intersection, erasure), so the latter is regular iff each of the former is.
\begin{itemize}
\item $\compi{\seq{\Gamma}{\while{M}{N}}}{i_\Gamma}= (\sum_{j=1}^N  \compi{\seq{\Gamma}{M}}{i_\Gamma,j} 
\,\compi{\seq{\Gamma}{N}}{i_\Gamma,\star}  )^\ast \,\compi{\seq{\Gamma}{M}}{i_\Gamma,0}$
\item $\compi{\seq{\Gamma}{\letin{x^\beta=M}{N}}}{i_\Gamma}=\sum_{j\in I_{\sem{\beta}}} \compi{\seq{\Gamma}{M}}{i_\Gamma,j}  
\compi{\seq{\Gamma,x}{N}}{(i_\Gamma,j_x)}$
\end{itemize}

The remaining cases are those of $\mathsf{let}$-bindings of the form $\letin{x=z(\cdots)}{\cdots}$.
First we explain some notation used throughout.
Consider the following context $\Gamma,z:\theta'\rarr\theta,x:\theta$.
We shall refer to moves contributed by $x:\theta$ with $m_x$.
If we want to range solely over O- or P-moves from the component, we use $o_x$ and $p_x$ respectively.
Moreover, we use $m_{z,x}, o_{z,x}, p_{z,x}$ to refer to copies of $m_x, o_x, p_x$ in the $z:\theta'\rarr\theta$ component.
The most common operation performed using this notation will be the relabelling of $m_x$ to $m_{z,x}$.
If $\theta$ is a function type, then there is a unique P-question $q_x$ enabled by the initial move $\star_x$.
Whenever we have a separate substitution rule for $q_x$, the rule for $m_x$ or $p_x$ will not apply to $q_x$.
In most cases we will want to substitute $q_{z,x}$ or $q_{z,x}^\bullet$ for $q_x$. In the latter case, $q_{z,x}^\bullet$ stands for $q_{z,x}$
annotated to represent the source of a pointer to a source move represented by $\star_{z,x}^\circ$.

First we tackle cases where the bound value is of function type, i.e. those related to possibly ambiguous
pointer reconstruction. Note that we include plays with and without pointer representations.
\begin{itemize}
\item $\seq{\Gamma,z:\beta\rarr(\theta_1\rarr\theta_2), y:\beta}{\letin{x=z y}{N}}$
given $\seq{\Gamma,z, y, x:\theta_1\rarr\theta_2}{N}$
\[\begin{array}{rcl}
\compi{\seq{\Gamma,z,y}{\letin{\cdots}{\cdots}}}{(i_\Gamma,\star_z,j_y)} &=& 
j_z \star_{z,x}^\circ  \compi{\seq{\Gamma,z,y,x}{N}}{(i_\Gamma,\star_z,j_y,\star_x)}[q_{z,x}^\bullet/q_x,\,\, m_{z,x}/m_x]\\
&+& j_z \star_{z,x}  \compi{\seq{\Gamma,z,y,x}{N}}{(i_\Gamma,\star_z,j_y,\star_x)}[q_{z,x}/q_x,\,\, m_{z,x}/m_x])
\end{array}\]
\item $\seq{\Gamma,z:(\beta_1\rarr\beta_2)\rarr(\theta_1\rarr\theta_2)}{\letin{x=z(\lambda y^{\beta_1}.M)}{N}}$
given $\seq{\Gamma,z, y:\beta_1}{M:\beta_2}$ and $\seq{\Gamma,z,x:\theta_1\rarr\theta_2}{N}$
\[\begin{array}{rcl}
\compi{\seq{\Gamma,z}{\letin{\cdots}{\cdots}}}{(i_\Gamma,\star_z)} &=&
q_z\, \clg{C}' \star_{z,x}^\circ \compi{\seq{\Gamma,z,x}{N}}{(i_\Gamma,\star_z,\star_x)}[q_{z,x}^\bullet\clg{C}'/q_x,\,\, p_{z,x}\clg{C}'/p_x,\,\, o_{z,x}/o_x]\\
&+& q_z\, \clg{C}' \star_{z,x} \compi{\seq{\Gamma,z,x}{N}}{(i_\Gamma,\star_z,\star_x)}[q_{z,x}\clg{C}'/q_x,\,\, p_{z,x}\clg{C}'/p_x,\,\, o_{z,x}/o_x]
\end{array}
\]
where $\clg{C}'=(\sum_{i\in I_{\sem{\beta_1}}} i_z\, (\sum_{j\in I_{\sem{\beta_2}}}\compi{\seq{\Gamma,z,y}{M}}{(i_\Gamma,\star_z,i_y),j} j_z) )^\ast$ 
\item $\seq{\Gamma,z:\vart\rarr(\theta_1\rarr\theta_2)}{\letin{x=z\,\badvar{\lambda u^{\comt}.M_1}{\lambda v^\expt.M_2}}{N}}$
given $\seq{\Gamma,z, u}{M_1}$, as well as $\seq{\Gamma,z, v}{M_2}$ and $\seq{\Gamma,z,x:\theta_1\rarr\theta_2}{N}$
\[\begin{array}{rcl}
\compi{\seq{\Gamma,z}{\letin{\cdots}{\cdots}}}{(i_\Gamma,\star_z)} &=& 
q_z\, \clg{C}' \star_{z,x}^\circ  \compi{\seq{\Gamma,z,x}{N}}{(i_\Gamma,\star_z,\star_x)}[q_{z,x}^\bullet\clg{C}'/q_x,\,\, p_{z,x}\clg{C}'/p_x,\,\, o_{z,x}/o_x] \\
&+& 
q_z\, \clg{C}' \star_{z,x}  \compi{\seq{\Gamma,z,x}{N}}{(i_\Gamma,\star_z,\star_x)}[q_{z,x}\clg{C}'/q_x,\,\, p_{z,x}\clg{C}'/p_x,\,\, o_{z,x}/o_x],
\end{array}\]
where $\clg{C}'=(\mread_z\, (\sum_{j=0}^N\compi{\seq{\Gamma,z,u}{M_1}}{(i_\Gamma,\star_z,\star_u),j} \,\,j_z) 
+ \sum_{j=0}^N\mwrite{j}_z\,\, \compi{\seq{\Gamma,z,v}{M_2}}{(i_\Gamma,\star_z,j_v),\star} \,\, \mok_z
)^\ast$ 
\end{itemize}
To understand the second formula (the third case is analogous) observe that, after $\star_{z,x}$ has been played,
plays for $\letin{\cdots}{\cdots}$ are plays from $N$ interleaved with possible detours to $\lambda x^{\beta_1}.M$:
such a detour can be triggered by  $i_z$ from $\beta_1$ each time the second move ($q_z$) is O-visible. 
Moreover, provided $q_z$ is O-visible,
such a detour can also take place between $q_z$ and $\star_{z,x}$.
The following auxiliary lemma will help us analyze when detours can occur.
\begin{lemma}
Let $A$ be a \nt{prearena}, $ s\in P_A$ be non-empty and  $P_i$ --- the set of P-moves enabled by 
the initial move of $s$. Let $s'$ be a prefix of $s$ containing at least two moves.
Then the O-view of $s'$ contains exactly one move from $P_i$.
\end{lemma}
\begin{proof}
Any non-initial move must be either in $P_i$ or hereditarily enabled by a move from $P_i$.
By visibility the O-view of $s'$ any prefix of $s$ must thus contain
a move from $P_i$. Because moves from $P_i$ are enabled by the initial move, they are
always the second moves in O-views. Hence, no two moves from $P_i$ can occur in the same O-view.
\end{proof}
Let us apply the lemma to the denotation of $M$. Because $\beta_2$ is a base type
it follows that at any time non-trivial O-views will contain a move from $\Gamma$ enabled
by the initial move or the answer from $\beta_2$ (which completes the play).
Returning to the play for $\letin{\cdots}{\cdots}$, this means that during a detour the second move
$q_z$ will be hidden from O-view until the detour is completed, i.e.
a single detour has to be completed
before the next one begins. Hence, $\clg{C'}$ has the form $(\cdots)^\ast$.

Also by the lemma above, once the play after $\star_{z,x}$ progresses, the second move $q_z$ will be O-visible if, and only if,
$q_{z,x}$ (which $\star_{z,x}$ enables) is. Thus detours will be possible exactly after $P$ 
plays a $P$-move hereditarily justified by $\star_{z,x}$, which corresponds to $P$ playing
a $P$-move hereditarily justified by $q_x$ in $N$ (hence the substitutions $p_{z,x} \clg{C'}/ p_x$ restricted
to $P$-moves). A special case is then that of $q_{z,x}$ which, as a question enabled by an answer, should be represented both with and without a pointer.

The above three cases cover all scenarios (that can arise in $\iatwo$) in which $z$'s type 
is of the form $\theta\rarr(\theta_1\rarr\theta_2)$.
The cases where $z:\theta\rarr\vart$ are analogous except that one needs to use $q_x$ to range
over $\mread_x, \mwrite{0}_x, \cdots, \mwrite{N}_x$ rather than the single move enabled by $\star_x$.
It remains to consider cases of $z:\theta\rarr\beta$. The bound values are of base type, so  no new pointer indices need to be
introduced.
\begin{itemize}
\item $\seq{\Gamma,z:\beta'\rarr\beta, y:\beta'}{\letin{x=z y}{N}}$
given $\seq{\Gamma,z, y, x:\beta}{N}$
\[
\compi{\seq{\Gamma,z,y}{\letin{}{}}}{(i_\Gamma,\star_z, j_y)} = j_z\,  (\sum_{k\in I_{\sem{\beta}}}k_z\, 
\compi{\seq{\Gamma,z,x,y}{N}}{(i_\Gamma,\star_z,j_y,k_x)})
\]
\item $\seq{\Gamma,z:(\beta_1\rarr\beta_2)\rarr\beta}{\letin{x=z(\lambda y^{\beta_1}.M)}{N}}$
given $\seq{\Gamma,z, y:\beta_1}{M:\beta_2}$ and $\seq{\Gamma,z,x:\beta}{N}$
\[
\compi{\seq{\Gamma,z}{\letin{}{}}}{(i_\Gamma,\star_z)} = q_z\, \clg{C}'  (\sum_{k\in I_{\sem{\beta}}} k_z\,\compi{\seq{\Gamma,z,x}{N}}{(i_\Gamma,\star_z,k_x)})
\]
where $\clg{C}'=(\sum_{i\in I_{\sem{\beta_1}}} i_z\, (\sum_{j\in I_{\sem{\beta_2}}}\compi{\seq{\Gamma,z,y}{M}}{(i_\Gamma,\star_z,i_y),j} j_z) )^\ast$ 
\item $\seq{\Gamma,z:\vart\rarr\beta}{\letin{x=z\,\badvar{\lambda u^{\comt}.M_1}{\lambda v^\expt.M_2}}{N}}$
given $\seq{\Gamma,z, u}{M_1}$, $\seq{\Gamma,z, v}{M_2}$ and $\seq{\Gamma,z,x}{N}$
\[
\compi{\seq{\Gamma,z}{\letin{}{}}}{(i_\Gamma,\star_z)} = q_z\, \clg{C}'  (\sum_{k\in I_{\sem{\beta}}} k_z\, \compi{\seq{\Gamma,z,x}{N}}{(i_\Gamma,\star_z,k_x)})
\]
where $\clg{C}'=(\mread_z\, (\sum_{j=0}^N\compi{\seq{\Gamma,z,u}{M_1}}{(i_\Gamma,\star_z,\star_u),j} \,\,j_z) 
+ \sum_{j=0}^N\mwrite{j}_z\,\, \compi{\seq{\Gamma,z,v}{M_2}}{(i_\Gamma,\star_z,j_v),\star} \,\, \mok_z
)^\ast$.\qedhere
\end{itemize}
\end{proof}
\begin{theorem}
Program equivalence of $\iatwo$-terms is decidable.
\end{theorem}
We remark that adding dynamic memory allocation in the form of $\newc$ to $\iatwo$,
or its second-order sublanguage, results in undecidability~\cite{Mur04b}. Hence,
at second order, block structure is ``strictly weaker" than scope extrusion.

\section{Summary}


In this paper we have introduced the notion of block-innocence that
has been linked with call-by-value Idealized Algol in a sequence of 
results. Thanks to the faithfulness
of block-innocence, we could investigate the interplay between
type theory, functional computation and stateful computation 
with block structure and dynamic allocation respectively. 
We have also shown a new decidability result for a carefully
designed fragment of $\iacbv$. Its extension to product types
poses no particular difficulty. In fact, it suffices to follow the way we have tackled
the $\vart$ type, which is itself a product type.
The result thus extends those from~\cite{Ghi01} and is a step forward
towards a full classification of decidable fragments of $\iacbv$:
the language $\iatwo$ we considered features all second-order types
and some third-order types, while finitary $\iacbv$ is known to be undecidable
at order $5$~\cite{Mur03b}. Interestingly, $\iatwo$ features restrictions that are 
compatible with the use of higher-order types in PASCAL~\cite{Mit02},
in which procedure parameters cannot be procedures with procedure parameters.
An interesting topic for future work would be to characterise the uniformity inherent in block-innocence in more
abstract, possibly category-theoretic, terms.

\section*{Acknowledgement}

We would like to thank the anonymous referees for their constructive comments and feedback.

\bibliographystyle{plain}
\bibliography{my}

\newcommand{\noopsort}[1]{} \newcommand{\printfirst}[2]{#1}
  \newcommand{\singleletter}[1]{#1} \newcommand{\switchargs}[2]{#2#1}
\begin{thebibliography}{10}

\bibitem{AJM00}
S.~Abramsky, R.~Jagadeesan, and P.~Malacaria.
\newblock Full abstraction for {PCF}.
\newblock {\em Information and Computation}, 163:409--470, 2000.

\bibitem{AM97b}
S.~Abramsky and G.~McCusker.
\newblock Call-by-value games.
\newblock In {\em Proceedings of CSL}, volume 1414 of {\em Lecture Notes in
  Computer Science}, pages 1--17. Springer-Verlag, 1997.

\bibitem{AM97a}
S.~Abramsky and G.~McCusker.
\newblock Linearity, sharing and state: a fully abstract game semantics for
  {I}dealized {A}lgol with active expressions.
\newblock In P.~W. O'Hearn and R.~D. Tennent, editors, {\em Algol-like
  languages}, pages 297--329. Birkha\"user, 1997.

\bibitem{AM99a}
S.~Abramsky and G.~McCusker.
\newblock Full abstraction for {Idealized Algol} with passive expressions.
\newblock {\em Theoretical Computer Science}, 227:3--42, 1999.

\bibitem{CBHMO15}
C.~Cotton{-}Barratt, D.~Hopkins, A.~S. Murawski, and C.{-}H.~L. Ong.
\newblock Fragments of {ML} decidable by nested data class memory automata.
\newblock In {\em Proceedings of FOSSACS'15}, volume 9034 of {\em Lecture Notes
  in Computer Science}, pages 249--263. Springer, 2015.

\bibitem{GP02}
M.~J. Gabbay and A.~M. Pitts.
\newblock A new approach to abstract syntax with variable binding.
\newblock {\em Formal Aspects of Computing}, 13:341--363, 2002.

\bibitem{Ghi01}
D.~R. Ghica.
\newblock Regular-language semantics for a call-by-value programming language.
\newblock In {\em Proceedings of MFPS}, volume~45 of {\em Electronic Notes in
  Computer Science}. Elsevier, 2001.

\bibitem{HY97}
K.~Honda and N.~Yoshida.
\newblock Game-theoretic analysis of call-by-value computation.
\newblock {\em Theoretical Computer Science}, 221(1--2):393--456, 1999.

\bibitem{HMO11}
D.~Hopkins, A.~S. Murawski, and C.-H.~L. Ong.
\newblock A fragment of {ML} decidable by visibly pushdown automata.
\newblock In {\em Proceedings of ICALP}, volume 6756 of {\em Lecture Notes in
  Computer Science}, pages 149--161. Springer, 2011.

\bibitem{HO00}
J.~M.~E. Hyland and C.-H.~L. Ong.
\newblock {On Full Abstraction for PCF: I. Models, observables and the full
  abstraction problem, II. Dialogue games and innocent strategies, III. A fully
  abstract and universal game model}.
\newblock {\em Information and Computation}, 163(2):285--408, 2000.

\bibitem{MTH90}
R.~Milner, M.~Tofte, and R.~Harper.
\newblock {\em The Definition of Standard ML}.
\newblock The MIT Press, Cambridge, Massachussetts, 1990.

\bibitem{Mit02}
J.~C. Mitchell.
\newblock {\em Concepts in programming languages}.
\newblock Cambridge University Press, 2002.

\bibitem{Mog91}
E.~Moggi.
\newblock Notions of computation and monads.
\newblock {\em Information and Computation}, 93:55--92, 1991.

\bibitem{Mur03b}
A.~S. Murawski.
\newblock About the undecidability of program equivalence in finitary languages
  with state.
\newblock {\em ACM Transactions on Computational Logic}, 6(4):701--726, 2005.

\bibitem{Mur04b}
A.~S. Murawski.
\newblock Functions with local state: regularity and undecidability.
\newblock {\em Theoretical Computer Science}, 338(1/3):315--349, 2005.

\bibitem{MOW05}
A.~S. Murawski, C.-H.~L. Ong, and I.~Walukiewicz.
\newblock {I}dealized {A}lgol with ground recursion and {DPDA} equivalence.
\newblock In {\em Proceedings of ICALP}, volume 3580 of {\em Lecture Notes in
  Computer Science}, pages 917--929. Springer, 2005.

\bibitem{MT10}
A.~S. Murawski and N.~Tzevelekos.
\newblock Block structure vs scope extrusion: between innocence and
  omniscience.
\newblock In {\em Proceedings of FOSSACS}, volume 6014 of {\em Lecture Notes in
  Computer Science}, pages 33--47. Springer-Verlag, 2010.

\bibitem{Ole85}
F.~Oles.
\newblock Type algebras, functor categories and block structure.
\newblock In M.~Nivat and J.~C. Reynolds, editors, {\em Algebraic Methods in
  Semantics}, pages 543--573. Cambridge University Press, 1985.

\bibitem{Ong02}
C.-H.~L. Ong.
\newblock Observational equivalence of 3rd-order {I}dealized {A}lgol is
  decidable.
\newblock In {\em Proceedings of IEEE Symposium on Logic in Computer Science},
  pages 245--256. Computer Society Press, 2002.

\bibitem{PS93}
A.~M. Pitts and I.~Stark.
\newblock On the observable properties of higher order functions that
  dynamically create local names, or: What's new?
\newblock In {\em Proc. 18th Int. Symp. on Math. Foundations of Computer
  Science}, pages 122--141. Springer-Verlag, 1993.
\newblock LNCS Vol.~711.

\bibitem{Plo77}
G.~D. Plotkin.
\newblock {LCF} considered as a programming language.
\newblock {\em Theoretical Computer Science}, 5:223--255, 1977.

\bibitem{Pow06}
J.~Power.
\newblock Semantics for local computational effects.
\newblock {\em Electr. Notes Theor. Comput. Sci.}, 158:355--371, 2006.

\bibitem{Rey81}
J.~C. Reynolds.
\newblock The essence of {A}lgol.
\newblock In J.~W. de~Bakker and J.C. van Vliet, editors, {\em Algorithmic
  Languages}, pages 345--372. North Holland, 1981.

\bibitem{Sta95}
I.~D.~B. Stark.
\newblock {\em Names and Higher-Order Functions}.
\newblock PhD thesis, University of Cambridge Computing Laboratory, 1995.
\newblock {T}echnical Report No. 363.

\bibitem{Tze09}
N.~Tzevelekos.
\newblock Full abstraction for nominal general references.
\newblock {\em Logical Methods in Computer Science}, 5(3), 2009.

\end{thebibliography}

\appendix
\section*{Appendix}
\section{S-plays}\label{apx:plays}
In this section we use the term ``move" and ``S-move" interchangeably.

\begin{lemma}[Prev-PA]\label{l:PrevPA}
If $s=\cdots m^\Sigma a_P^\Tau\cdots$ is an S-play then, for any $\na$,
  \begin{enumerate}[label=\({\alph*}]
    \item if $\na\in\nu(\Tau)$ then $\na\in\nu(\Sigma)$,
    \item if $\na\in\nu(\Sigma\remv\Tau)$ then $\na$ is closed in $s_{<a_P^T}$.
  \end{enumerate}
Moreover, $\Tau\Substore\Sigma$ and therefore $\Sigma\remv(\Sigma\remv\Tau)\Substore\Tau$ and $\Sigma\remv\Tau\substorE\Sigma$.
\end{lemma}

\begin{proof}
Let $s=s_1\,\rnode{A}{q_0^{\Sigma_0}}s_2\,\rnode{B}{a ^\Tau}\justf{B}{A}\cdots$. As $q_0^{\Sigma_0}$ is the pending-Q in $s_1q_0^{\Sigma_0}s_2$, we have that $s$ is in fact of the form:

\[ s_1\,\rnode{A}{q_0^{\Sigma_0}}\rnode{BA}{q_{1}^{\Tau_1}}\cdots\rnode{BB}{a_{1}^{\Sigma_1}}
\rnode{CA}{q_{2}^{\Tau_2}}\cdots\rnode{CB}{a_{2}^{\Sigma_2}}
\cdots\rnode{DA}{q_{j}^{\Tau_j}}\cdots\rnode{DB}{a_{j}^{\Sigma_j}}\rnode{B}{a^\Tau}\justf{BB}{BA}\justf{CB}{CA}\justf{DB}{DA}
\justf[,angleA=155,angleB=25]{B}{A} \]
For (a), by Just-P we have that $\na\in\nu(\Sigma_0)$ and therefore $\na$ is not closed in $s_1m_0^{\Sigma_0}s_2$. Hence, by Prev-PQ(b), $\na\in\nu(\Tau_1)$ and thus, by Just-O, $\na\in\nu(\Sigma_1)$. Repeating this argument $j$ times we obtain $\na\in\nu(\Sigma_j)$, i.e.~$\na\in\nu(\Sigma)$.\\
For (b), let $\na\in\nu(\Sigma)=\nu(\Sigma_j)$ be open in $s_{<a^\Tau}$. We claim that then $\na\in\nu(\Sigma_0)$ and therefore $\na\in\nu(\Tau)$ by Just-P. We have that $\na\in\nu(\Tau_j)$, by Just-O. Moreover, since there are no open questions in ${q_{j}^{\Tau_j}}\cdots{a_{j}^{\Sigma_j}}$, we have that $\na$ is open in $s_{<{q_{j}^{{\Tau}}}^{{}_{\!_j}}}$, so $\na\in\nu(s_{<{q_{j}^{{\Tau}}}^{{}_{\!_j}}})$ and thus, by Prev-PQ(a), $\na\in\dom(\Sigma_{j-1})$. Applying this argument $j$ times we obtain $\na\in\nu(\Sigma_0)$.\\
Finally, we show by induction that $\Sigma_0\Substore\Sigma_i$, for all $0\leq i\leq j$. For the inductive step, we have that $\Sigma_i\remv(\Sigma_i\remv\Tau_{i+1})\Substore\Tau_{i+1}\Substore\Sigma_{i+1}$. By IH, $\Sigma_0\Substore\Sigma_i$. Moreover, if $\na\in\nu(\Sigma_i\remv\Tau_{i+1})$ then $\na\notin\nu(\Sigma_0)$, by Prev-PQ(b), so $\Sigma_0\Substore\Sigma_i\remv(\Sigma_i\remv\Tau_{i+1})$, $\therefore\Sigma_0\Substore\Sigma_{i+1}$. Thus, $\Tau\Substore\Sigma_0\Substore\Sigma_j=\Sigma$. 
\end{proof}

\begin{lemma}[Block Form]\label{lem:block}
If $s$ is an S-play then $\pv{s}$ is in \emph{block-form}: for any $\na$, we have $\pv{s}=s_1s_2s_3$, where
\begin{itemize}
  \item $\na\notin\nu(s_1)\cup\nu(s_3)$ and $\forall m^\Sigma\in s_2.\ \na\in\nu(\Sigma)$,
  \item if $s_2\neq\epsilon$ then its first element is the move introducing $\na$ in $s$.
\end{itemize}
\end{lemma}

\begin{proof}
We do induction on $|s|$, with the cases of $|s|\leq 1 $ being trivial. If $s=s'o^\Sigma$ then, by IH, $\pv{s}=s''p^\Tau o^\Sigma$, some $s''p^\Tau$ in block-form, and $\nu(\Tau)=\nu(\Sigma)$, which imply that $\pv{s}$ is in block-form.
If $s=s'p^\Sigma$ then $\pv{s}=\pv{s'}p^\Sigma$ with $\pv{s'}=s_1s_2s_3$ in block-form by IH. If $\na\notin\dom(\Sigma)$ then OK. Otherwise, if $\na$ does not appear in the last move in $s_1s_2s_3$ then, by Prev,
$p^\Sigma$ in fact introduces $\na$ in $s$ and therefore $\pv{s}$ has block-form.
\end{proof}

\begin{lemma}\label{l:oviews}
Let $s=s_1o^\Sigma p^\Tau s_2$ be an S-play with $\na\in\nu(\Sigma)\setminus\nu(\Tau)$. Then, for any $s_2'\prefix s_2$, $\na\notin\nu(\ov{s_1o^\Sigma p^\Tau s_2'})$.
\end{lemma}
\begin{proof}
We do induction on $|s_2'|$. For the base case, by the previous lemma we have that $\pv{s_1o^\Sigma}$ has block-form; in particular, it ends in a block of moves which contains the move introducing $\na$ in $s$, and all moves in the block contain $\na$ in their stores. Hence, since the justifier of $p^\Tau$, say $o'^{\Sigma'}$, occurs in $\pv{s_1o^\Sigma}$ and $\na\notin\nu(\Sigma')$ (by Just and the fact that $\na\notin\nu(\Tau)$), we have that $o'^{\Sigma'}$ occurs in $s$ before the move introducing $\na$ in it and therefore $\na\notin\nu(\ov{s_1o^\Sigma p^\Tau})$.
Now, if $s_2'=s_2''o'^{\Sigma'}$ then we need to show that $\na\notin\nu(\ov{s_1o^\Sigma p^\Tau s_2''}o'^{\Sigma'})$ given by IH that $\na\notin\nu(\ov{s_1o^\Sigma p^\Tau s_2''})$, which immediately follows from Just and Visibility. Finally, if $s_2'=s_2''p'^{\Tau'}$ then, by IH, $\na\notin\nu(s_2'')$ and therefore, by Prev, $\na\notin\nu(\Tau')$. Let $p'^{\Tau'}$ be justified by some $o'^{\Sigma'}$ in $s$. If $o'^{\Sigma'}$ occurs in $s_2''$ then our argument follows directly from the IH. Otherwise, arguing as before, $o'^{\Sigma'}$ occurs in $s$ before the introduction of $\na$ and therefore $\na\notin\nu(\ov{s_1o^\Sigma p^\Tau s_2''p'^{\Tau'}})$.
\end{proof}

\begin{corollary}[Close]
If $s=s_1o^\Sigma p^\Tau s_2$ is an S-play with $\na\in \nu(\Sigma)\setminus\nu(\Tau)$ then $\na\notin\nu(s_2)$.\qed
\end{corollary}

\begin{proof}[Proof of Lemma~\ref{l:inter}.]
For (a), assuming WLOG $p$ is a P-move in $AB$ and taking $s=s'n^{\Tau'}p^{\Sigma'}$, by definition of the interaction and Prev we have that if $\na\in\nu(\Sigma\remv\Tau)$ then $\na\in\nu(\Sigma'\remv\Tau')$, hence $\na\notin\nu(s'n^{\Tau'})$ and therefore $\na\notin\nu(un^\Tau)$.
\\
For (b), we do induction on $|s\iseq t|$, base case is trivial. Now, if $s\iseq t=u p^\Sigma$ with $p$ a generalised P-move then $\pv{s\iseq t}=\pv{u}p^\Sigma$ and, by IH, $\pv{u}$ has block form, say $u_1u_2u_3$. If $\na\notin\nu(\Sigma)$ then OK. Otherwise, if $\na$ does not appear in the last move of $\pv{u}$ then $p^\Sigma$ is in fact the move introducing $\na$ in $s\iseq t$, so OK. Otherwise, $u_3=\ee$ and therefore $\pv{s\iseq t}$ in block form. If $s\iseq t=u o^\Tau$ with $o$ an O-move in $AC$ then $\pv{s\iseq t}=u'o^\Tau$ for some $u'$ in block-form, and the last move in $u'$ has domain $\dom(\Tau)$. This implies that $s\iseq t$ is in block-form.
\\
For (c), assuming WLOG $p$ is a P-move in $AB$ and taking $s=s'n^{\Tau'}p^{\Sigma'}$, by definition of the interaction we have $\na\in\nu(\Tau'\remv\Sigma')$ so, by Prev, $\na$ is closed in $s'n^{\Tau'}$. Now suppose $\na$ is open in $un^{\Tau}$, that is, there is an open question $q_1^{\Sigma_1}$ in $un^{\Tau}$ with $\na\in\nu(\Sigma_1)$. Then, if $q_0^{\Sigma_0}$ is the pending question of $un^{\Tau}$ then $\na\in\nu(\Sigma_0)$: $\pv{un^{\Tau}}$ has block-form, by (b), and $q_0^{\Sigma_0}$ appears in it, thus if $\na\notin\nu(\Sigma_0)$ then $q_0^{\Sigma_0}$ would precede the move introducing $\na$ in $un^{\Tau}$ and hence it would precede $q_1^{\Sigma_1}$ too. As the interaction ends in an O-move in $AB$, $q_0^{\Sigma_0}$ is the pending question of $un^\Tau\rest AB$. Let $q_0^{\Sigma'_0}$ be the move in $s$ corresponding to $q_0^{\Sigma_0}$. $q_0^{\Sigma'_0}$ is the pending question in $s'n^{\Tau'}$ and therefore it appears in $\pv{s'n^{\Tau'}}$. Moreover, since $\na$ is closed in $s'n^{\Tau'}$, $\na\notin\nu(\Sigma'_0)$, which means that $q_0^{\Sigma'_0}$ occurs in $\pv{s'n^{\Tau'}}$ before the move introducing $\na$ in $s$. But then $\na\notin\nu(\Sigma_0)$, a contradiction.
\\
For (d), we do induction on $|s\iseq t|$, the base case being trivial. If $m^\Sigma$ is an O-move in $AC$ then the claim is obvious. So assume WLOG $m^\Sigma$ is a P-move in $AB$ and consider $\pv{s\iseq t\rest AB}_{\geq n^\Tau}$, and take two consecutive moves $m_1^{\Sigma_1}m_2^{\Sigma_2}$ in it, and assume $m_2^{\Sigma_2}$ be a P-move in $AB$. Then these are consecutive also in $s\iseq t\rest AB$ and hence, by switching, also in $s\iseq t$. Let $\na\in\nu(\Sigma_1\remv\Sigma_2)$. By definition of interaction, this dropping of $\na$ happens in $s$ too, at the respective consecutive moves  $m_1^{\Sigma_1'}m_2^{\Sigma_2'}$, assuming $s=\cdots n^{\Tau'}\cdots m_1^{\Sigma_1'}m_2^{\Sigma_2'}\cdots m^{\Sigma'}$. We can see that, as $\pv{s}$ is in block-form, $\na\notin\nu(\Sigma')$ and therefore, by Just, $\na\notin\nu(\Tau')$. More than that, $n^{\Tau'}$ occurs in $s$ before the move introducing $\na$ in $s$, thus $\na\notin\nu(\Tau)$. We therefore have $\Sigma_1\remv(\Sigma_1\remv\Sigma_2)\Substore\Sigma_2$ and $\nu(\Sigma_1\remv\Sigma_2)\cap\nu(\Tau)=\varnothing$. On the other hand, if $m_2^{\Sigma_2}$ is an O-move in $AB$ then $m_1^{\Sigma_1}$ is its justifier and therefore, by IH, $\Sigma_1\Substore\Sigma_2$. Thus, for every move $m''^{\Sigma''}$ in $\pv{s\iseq t\rest AB}_{\geq n^\Tau}$, we have $\Tau\Substore\Sigma''$ and hence $\Tau\Substore\Sigma$. Finally, if $m^\Sigma$ is an answer then, since $\pv{s\iseq t\rest AB}$ satisfies well-bracketing, we have that $\pv{s\iseq t\rest AB}_{\geq n^\Tau}=n^{\Tau_0}\rnode{BA}{q_{1}^{\Sigma_1}}\rnode{BB}{a_{1}^{\Tau_1}}
\cdots\rnode{DA}{q_{j}^{\Sigma_j}}\rnode{DB}{a_{j}^{\Tau_j}}m^{\Sigma_{j+1}}$
with $\dom(\Sigma_i)=\dom(\Tau_i)$ for each $1\leq i\leq j$ by IH. Assuming $s_{\geq
n^{\Tau'}}=n^{\Tau_0'}\rnode{BA}{q_{1}^{\Sigma_1'}}\cdots\rnode{BB}{a_{1}^{\Tau_1'}}\cdots\rnode{DA}{q_{j}^{\Sigma_j'}}
\cdots\rnode{DB}{a_{j}^{\Tau_j'}}m^{\Sigma_{j+1}'}$, if $\na\in\nu(\Sigma_{i+1}\remv\Tau_i)$ for some $i$ then $\na\in\nu(\Sigma_{i+1}'\remv\Tau_i')$ and therefore $\na\notin\nu(\Tau'_0)=\nu(\Sigma'_{j+1})$. But the latter implies that $\na\in\nu(\Tau_{i'}'\remv\Sigma_{i'+1}')$ for some $i<i'\leq j$ and therefore $\na\notin\nu(\Sigma_{i'+1})$. Thus, $\nu(\Sigma\remv\Tau)=\varnothing$.
\\
For (e), we replay the proof of lemma~\ref{l:oviews}, that is, we show that, for every $u_2'\prefix u_2$, $\na\notin\nu(\ov{u_1n^\Tau p^\Sigma u_2'})$. We do induction on $|u_2'|$. For the base case, by (b) we have that $\pv{u_1n^\Tau}$ has block-form; in particular, it ends in a block of moves which contains the move introducing $\na$ in $u$, and all moves in the block contain $\na$ in their stores. Hence, since the justifier of $p^\Sigma$, say $n'^{\Tau'}$, occurs in $\pv{u_1n^\Tau}$ and $a\notin\nu(\Tau')$ (by (d) and the fact that $\na\notin\nu(\Sigma)$), we have that $n'^{\Tau'}$ occurs in $u$ before the move introducing $\na$ in it and therefore $\na\notin\nu(\ov{u_1n^\Tau p^\Sigma})$.
Now, if $u_2'=u_2''o'^{\Tau'}$ (an O-move in $AC$) then we need to show that $\na\notin\nu(\ov{u_1m^\Tau p^\Sigma u_2''}o'^{\Tau'})$ given by IH that $\na\notin\nu(\ov{u_1m^\Tau p^\Sigma u_2''})$, which immediately follows from Visibility. Finally, if $u_2'=u_2''p'^{\Sigma'}$ (a generalised P-move) then, by IH, $\na\notin\nu(u_2'')$ and therefore, by (a), $\na\notin\nu(\Sigma')$. Let $p'^{\Sigma'}$ be justified by some $m'^{\Tau'}$ in $u$. If $m'^{\Tau'}$ occurs in $u_2''$ then our argument follows directly from the IH. Otherwise, arguing as before, $m'^{\Tau'}$ occurs in $u$ before the introduction of $\na$ and therefore $\na\notin\nu(\ov{u_1m^\Tau p^\Sigma u_2''p'^{\Sigma'}})$.
\\
For (f), suppose $\na$ appears in a $B$-move $n^{\Tau}$ of $u$. By (e), and because $\na$ reappears in $m^\Sigma$, we have that $\na$ also appears in the move following $n^{\Tau}$ in $u$. Applying this reasoning repeatedly, we obtain that $\na$ appears in $u\rest AC$ after $n^{\Tau}$.
\\
For (g), we do induction on $|s\iseq t|$, the base case being trivial. Now assume $s=s'm^{\Sigma'}$. If $m$ is an O-move then $\Sigma'\substore\Sigma$ follows from (d) and from the IH applied to the subsequence ending in the justifier of $m^\Sigma$. Moreover, if $\na\in\nu(\Sigma')$ then $\Sigma'(\na)$ is determined by the last appearance of $\na$ in $s'$, say $n^{\Tau'}$. By IH, the corresponding move of $s'\iseq t$ has store $\Tau$, with $\Tau(\na)=\Tau'(\na)$. But then, by inspection of the definition of interaction, any move in $s\iseq t$ occurring after $n^\Tau$ does not change the value of $\na$, hence $\Sigma(\na)=\Tau(\na)=\Sigma'(\na)$.
If $m$ is a P-move  preceded by $n^{\Tau'}$ then, using the IH, we have
\[ \Sigma'=\Tau'[\Sigma']\remv(\Tau'\remv\Sigma')+(\Sigma'\remv\Tau')\substore (sn^{\Tau'}\mix t)[\Sigma']\remv(\Tau'\remv\Sigma')+(\Sigma'\remv\Tau')=\Sigma\,.\]
From the above, and using again the IH, we obtain also that $\Sigma[\Sigma']=\Sigma$. Moreover, if $\na\in\nu(\st{t})$ then, by IH, if the last $B$-move of $s\iseq t$ has store $\Sigma_B$ then $\Sigma_B(\na)=\st{t}(\na)$ and the value of $\na$ cannot be changed by subsequent $A$-moves. Hence, if  $\na\in\nu(\Sigma)$ then $\Sigma(\na)=\st{t}(\na)$.
The case of $t=t'm^{\Sigma'}$ is treated dually.
\\
For (h), the last two moves come from the same sequence, say $s$, so let $s=s'n^{\Tau'}p^{\Sigma'}$. We have that
\begin{align*}
\Tau\remv\Sigma &= (sn^{\Tau'}\mix t)\remv((sn^{\Tau'}\mix t)[\Sigma']\remv(\Tau'\remv\Sigma')+(\Sigma'\remv\Tau')) \\
 &= (sn^{\Tau'}\mix t)\remv((sn^{\Tau'}\mix t)[\Sigma']\remv(\Tau'\remv\Sigma'))=(\Tau'\remv\Sigma')[sn^{\Tau'}\mix t]
\end{align*}
where the last equality holds because $\Tau'\remv\Sigma'\substore\Tau'\substore T=sn^{\Tau'}\mix t$, by (g). Hence, $\Tau\remv(\Tau\remv\Sigma)=\Tau\remv(\Tau'\remv\Sigma')\Substore\Sigma$.
We still need to show that $\Tau'\remv\Sigma'\substorE\Tau$.
Let $\na_1,\na_2$ be consecutive names in $\dom(\Tau)$ such that $\na_1\in\nu(\Tau'\remv\Sigma')$. Then, the point of introduction of $\na_2$ in $un^\Tau$ does not precede that of $\na_1$, say $q_1^{\Sigma_1}$. Now, by (c), $\na_1$ is closed in $un^\Tau$ so, using also (b), the latter has the form $u'\rnode{BA}{q_{1}^{\Sigma_1}}\cdots\rnode{BB}{a_{1}^{\Tau_1}}
\rnode{CA}{q_{2}^{\Sigma_2}}\cdots\rnode{CB}{a_{2}^{\Tau_2}}
\cdots\rnode{DA}{q_{j}^{\Sigma_j}}\cdots\rnode{DB}{a_{j}^{\Tau_j}}\justf{BB}{BA}\justf{CB}{CA}\justf{DB}{DA}$\,.
Thus, by (d), the point of introduction of $\na_2$ has to be one of the $q_i^{\Sigma_i}$'s. This implies that $\na_1,\na_2\in\nu(s)$ and therefore $\na_1,\na_2$ are consecutive also in $\Tau'$. But then $\na_1\in\nu(\Tau'\remv\Sigma')$ implies $\na_2\in\nu(\Tau'\remv\Sigma')$ too.
\end{proof}

\begin{proof}[Proof of lemma~\ref{l:nice}.]
Let us write $\Sigma_{ij}$ for $\Sigma_i\remv\Sigma_j$. Then, the LHS is:
\begin{align*} L
&=\Sigma_1[\nice(\Sigma_3,\Sigma_4,\Sigma_5)]\remv(\Sigma_2\remv\nice(\Sigma_3,\Sigma_4,\Sigma_5))+\nice(\Sigma_3,\Sigma_4,\Sigma_5)\remv\Sigma_2 \\
&=\Sigma_1[\nice(\Sigma_3,\Sigma_4,\Sigma_5)]\remv(\Sigma_2\remv(\Sigma_3[\Sigma_5]\remv\Sigma_{45}+\Sigma_{54}))+(\Sigma_3[\Sigma_5]\remv\Sigma_{45}
+\Sigma_{54})\remv\Sigma_2 \\
&=\Sigma_1[\nice(\Sigma_3,\Sigma_4,\Sigma_5)]\remv(\Sigma_2\remv(\Sigma_3\remv\Sigma_{45}))+(\Sigma_3[\Sigma_5]\remv\Sigma_{45})\remv\Sigma_2
+\Sigma_{54}
\end{align*}
where the last equality holds because of (a). The first constituent above is:
\[
L_1 = \Sigma_1[\Sigma_3[\Sigma_5]\remv\Sigma_{45}+\Sigma_{54}]\remv(\Sigma_2\remv(\Sigma_3\remv\Sigma_{45}))
= \Sigma_1[\Sigma_3[\Sigma_5]\remv\Sigma_{45}]\remv(\Sigma_2\remv(\Sigma_3\remv\Sigma_{45}))
\]
On the other hand, the RHS is:
\begin{align*}R
&=\nice(\Sigma_1,\Sigma_2,\Sigma_3)[\Sigma_5]\remv\Sigma_{45}+\Sigma_{54}
=(\Sigma_1[\Sigma_3]\remv\Sigma_{23}+\Sigma_{32})[\Sigma_5]\remv\Sigma_{45}+\Sigma_{54} \\
&=((\Sigma_1[\Sigma_3]\remv\Sigma_{23})\remv\Sigma_{45})[\Sigma_5]+(\Sigma_{32}\remv\Sigma_{45})[\Sigma_5]+\Sigma_{54} \\
&=((\Sigma_1[\Sigma_3]\remv\Sigma_{23})\remv\Sigma_{45})[\Sigma_5]+(\Sigma_{3}[\Sigma_5]\remv\Sigma_{45})\remv\Sigma_2+\Sigma_{54}
\end{align*}
The first constituent above is:
\begin{align*}
R_1 &= ((\Sigma_1[\Sigma_3\remv\Sigma_{45}]\remv\Sigma_{23})\remv\Sigma_{45})[\Sigma_5]
= ((\Sigma_1[\Sigma_3[\Sigma_5]\remv\Sigma_{45}]\remv\Sigma_{23})\remv\Sigma_{45})[\Sigma_5] \\
&= (\Sigma_1[\Sigma_3[\Sigma_5]\remv\Sigma_{45}]\remv\Sigma_{23})\remv\Sigma_{45}
\end{align*}
For the last equality, let $\na$ by in the domain of the resulting store. Then $\na\in\nu(\Sigma_5)\cap\nu(\Sigma_1)$, $\therefore\na\in\nu(\Sigma_4)\cap\nu(\Sigma_1)$, so $\na\in\nu(\Sigma_2)$ by (b). Moreover, $\na\notin\nu(\Sigma_{23})$, $\therefore\na\in\nu(\Sigma_3)$. Now, let us write $\Sigma$ for $\Sigma_1[\Sigma_3[\Sigma_5]\remv\Sigma_{45}]$. We need to show that
\[ \Sigma\remv(\Sigma_2\remv(\Sigma_3\remv\Sigma_{45})) = (\Sigma\remv(\Sigma_{2}\remv\Sigma_{3}))\remv\Sigma_{45} \]
and, in fact, it suffices to show that these stores, say $\Sigma_L$ and $\Sigma^R$, have the same domain. By elementary computation,
\begin{itemize}
\item $\na\in\nu(\Sigma_L)$ iff $(\na\in\nu(\Sigma)\land\na\notin\nu(\Sigma_2))\lor (\na\in\nu(\Sigma)\land\na\in\nu(\Sigma_3)\land\na\notin\nu(\Sigma_{45}))$,
\item $\na\in\nu(\Sigma^R)$ iff $(\na\in\nu(\Sigma)\land\na\notin\nu(\Sigma_2)\land\na\notin\nu(\Sigma_{45}))\lor (\na\in\nu(\Sigma)\land\na\in\nu(\Sigma_3)\land\na\notin\nu(\Sigma_{45}))$.
\end{itemize}
But note now that $\na\in\nu(\Sigma)\land\na\notin\nu(\Sigma_2)$ implies that $\na\notin\nu(\Sigma_{4})$, by (b), so $\na\notin\nu(\Sigma_{45})$.
\end{proof}


\section{Proof of Lemma~\ref{lem:canon}\label{apx:canon}}

We prove two auxiliary results first, which are special cases of Lemma~\ref{lem:canon}.

\begin{lemma}\label{lem:var}
Any identifier $x^\theta$ satisfies Lemma~\ref{lem:canon}. Moreover, the canonical 
form is of the form $\lambda y^\theta_1.\can$ when $\theta\equiv\theta_1\rarr\theta_2$
and of the shape  $\badvar{\lambda y^\comt.\can}{\lambda z^\expt.\can}$ if $\theta\equiv\vart$.
\end{lemma}
\begin{proof} 
Induction with respect to type structure. If $\theta$ is a base type, $x^\theta$ is already in canonical
form.  $x^\vart$ can be converted to one using the rule
\[
x^\vart  \lrarr  \badvar{\lambda u^\expt. x\aasg u}{\lambda v^\comt.!x}
\]
For $\theta\equiv\theta_1\rarr\theta_2$ we use the rule
\[
x^{\theta_1\rarr\theta_2} \lrarr \lambda z^{\theta_1}. \letin{v^{\theta_2}=xz^{\theta_1}}{v} 
\]
and appeal to the inductive hypothesis for $z^{\theta_1}$ and $v^{\theta_2}$.
\end{proof}

\begin{lemma}\label{lem:let}
Suppose $C_1, C_2$ are canonical forms. Then $\letin{y^\theta=C_1}{C_2}$, if typable,
satisfies Lemma~\ref{lem:canon}.
\end{lemma}
\proof
Induction with respect to type structure. If $\theta$ is a base type, the term is already
in canonical form. If $\theta$ is not a base type, $C_1$ can take one of the following
three shapes: $\badvar{\lambda x^\comt.\can}{\lambda y^\expt.\can}$,  $\lambda x_1^{\theta_1}.\can$,
$\cond{x^\beta}{\can}{\can}$ or $\letin{\cdots}{\can}$.

We first focus on the first two of them to which the remaining two cases will be reduced later.
\begin{itemize}
\item Suppose $C_1\equiv\badvar{\lambda x_1^\comt.C_{11}}{\lambda x_2^\expt.C_{12}}$.
Then $\theta\equiv\vart$. Since $C_2$ is in canonical form, $y$ can only
occur in it as part of a canonical subterm of the form $y^\vart\aasg z^\expt$ or $!y$.
Hence, after substitution for $y$, we will obtain non-canonical subterms of the shape
$\badvar{\lambda x_1^\comt.C_{11}}{\lambda x_2^\expt.C_{12}}\aasg z$ and
$!(\badvar{\lambda x_1^\comt.C_{11}}{\lambda x_2^\expt.C_{12}})$.
Using the rules
\renewcommand\arraystretch{1}
\[\begin{array}{rcl}
! \badvar{\lambda u^\comt.D_1}{\lambda v^\expt.D_2} &\!\!\!\lrarr\!\!\!& D_1[()/u]\\
\badvar{\lambda u^\comt.D_1}{\lambda v^\expt.D_2}:= z &\!\!\!\lrarr\!\!\!& D_2[z/v]
\end{array}\]
we can easily convert them (and thus the whole term) to canonical form.
\item Suppose $C_1\equiv\lambda x_1^{\theta_1}.C_3$ and $\theta\equiv\theta_1\rarr\theta_2$.
Let us substitute $C_1$ for the rightmost occurrence of $y$ in $C_2$.
This will create a non-canonical subterm in $C_2$ of the form
$\letin{x^{\theta_2}=(\lambda x_1^{\theta_1}.C_3) C_4}{C_5}\equiv \letin{x^{\theta_2}=(\letin{x_1^{\theta_1}=C_4}{C_3})}{C_5}$.
By inductive hypothesis for $\theta_1$, $\letin{x_1^{\theta_1}=C_4}{C_3}$ can be converted to canonical form, say, $C_6$.
Consequently, the non-canonical subterm $\letin{x^{\theta_2}=(\lambda x_1^{\theta_1}.C_3) C_4}{C_5}$
 can be transformed into the form $\letin{x^{\theta_2}=C_6}{C_5}$, which --- by inductive hypothesis for $\theta_2$ --- can 
 also be converted to canonical form. Thus, we have shown how to recover canonical forms after substitution for the 
 rightmost occurrence of $y$. Because of the choice of the rightmost occurrence, the transformation does not
 involve terms containing other occurrences of $y$, so it will also decrease their overall number in $C_2$ by one.
Consequently, by repeated substitution for rightmost occurrences one can eventually arrive at a canonical form 
 for $\letin{y^\theta=(\lambda x_1^{\theta_1}.C_3)}{C_2}$.
 \end{itemize}
For the remaining two cases it suffices to take advantage of the following conversions
before referring to the two cases above.
\renewcommand\arraystretch{1}
\[\begin{array}{rcl}
\letin{y=(\cond{x}{D_1}{D_0})}{E}&\!\!\!\lrarr\!\!\! & \cond{x}{(\letin{y=D_1}{E})}{(\letin{y=D_0}{E})}\\
\letin{y=(\letin{x=D}{E})}{F} &\!\!\!\lrarr\!\!\!
                                                     &\letin{x=D}{(\letin{y=E}{F})}
\rlap{\hbox to100 pt{\hfill\qEd}}
\end{array}\]\smallskip

\noindent
Now we are ready to prove Lemma~\ref{lem:canon} by induction on term structure.
The base cases of $()$, $i$ are trivial. That of $x^\theta$ follows from Lemma~\ref{lem:var}.

The following inductive cases follow directly from the inductive
hypothesis: $\new{x}{M}$, $\lambda x^\theta.M$, $\while{M}{M}$.
The cases of $M_1\oplus M_2$ and $\cond{M}{N_1}{N_0}$
are only slightly more difficult. After invoking the inductive
hypothesis one needs to apply the rules given below.
Note that in this case all the $\mathsf{let}$-bindings are of base type.
\renewcommand\arraystretch{1}
\[\begin{array}{rcl}
D_1\oplus D_2 &\!\!\!\lrarr\!\!\!& \letin{x=D_1}{(\letin{y=D_2}{(x\oplus y)})}\\
\cond{D}{D_1}{D_0}&\!\!\!\lrarr\!\!\! &\letin{x=D}{(\cond{x}{D_1}{D_0})}
\end{array}\]
For $!M$ and $M\aasg N$ we take advantage of the fact that
a canonical form of type $\vart$ can only take three shapes:
$\badvar{\lambda x^\comt.\can}{\lambda y^\expt.\can}$, 
$\cond{x^\beta}{\can}{\can}$ or $\letin{\cdots}{\can}$.
An appeal to the inductive hypothesis for $M$ and $N$ and the conversions
given below will then yield the canonical forms for $!M$ and $M\aasg N$.
\[\begin{array}{rcl}
! \badvar{\lambda u^\comt.D_1}{\lambda v^\expt.D_2} &\!\!\!\lrarr\!\!\!& D_1[()/u]\\
\badvar{\lambda u^\comt.D_1}{\lambda v^\expt.D_2}:= E &\!\!\!\lrarr\!\!\!& \letin{x^\expt=E}{D_2[x/v]}\\
!(\cond{x}{D_1}{D_0}) &\!\!\!\lrarr\!\!\! & \cond{x}{!D_1}{!D_0}\\
(\cond{x}{D_1}{D_0} )\aasg D &\!\!\!\lrarr\!\!\! & \cond{x}{(D_1\aasg D)}{(D_0\aasg D)}\\
!(\letin{x=D}{E}) &\!\!\!\lrarr\!\!\! & \letin{x=D}{!E}\\
 (\letin{x=D}{E})\aasg F &\!\!\!\lrarr\!\!\! & \letin{x=D}{(E\aasg F)} 
\end{array}\]

To convert $\badvar{M}{N}$ to canonical form we observe that
a canonical form of function type can take the following shapes:
$\lambda x^\theta.\can$, $\cond{x^\beta}{\can}{\can}$ or $\letin{\cdots}{\can}$.
Hence, by appealing to the inductive hypothesis and then repeately applying
the rules below we will arrive at a canonical form.
\[\begin{array}{rcl}
\!\!\badvar{\cond{x}{D_1}{D_0}}{E} &\!\!\!\lrarr\!\!\! & \cond{x}{\badvar{D_1}{E}}{\badvar{D_0}{E}}\\
\!\!\badvar{\letin{x=M}{D}}{E} &\!\!\!\lrarr\!\!\! & \letin{x=M}{\badvar{D}{E}}\\
\!\!\badvar{\lambda u^\comt\!. D}{\cond{x}{E_1}{E_0}} &\!\!\!\lrarr\!\!\! & \cond{x}{\badvar{\lambda u^\comt\!.D}{E_1}}{\badvar{\lambda u^\comt\!.D}{E_0}}\\
\!\!\badvar{\lambda u^\comt\!. D}{\letin{x=M}{E}}&\!\!\!\lrarr\!\!\! & \letin{x=M}{\badvar{\lambda u^\comt\!. D}{E}}
\end{array}\]

Finally, we handle application $M N$. First we apply the inductive hypothesis to both terms.
Then we use the rules below to reveal the $\lambda$-abstraction inside the canonical form of $M$.
\[\begin{array}{rcl}
(\cond{x}{D_1}{D_0}) E &\!\!\!\lrarr\!\!\! &\cond{x}{(D_1 E)}{(D_2 E)}\\
(\letin{x=D}{E})F &\!\!\!\lrarr\!\!\! &\letin{x=D}{(E F)}
\end{array}\]
Now it suffices to be able to deal with terms of the form $(\lambda x^\theta.C_1) C_2\equiv\letin{x=C_2}{C_1}$,
and this is exactly what Lemma~\ref{lem:let} does.

All our transformation preserve denotations: the proofs are simple exercises in the use of Moggi's monadic
approach~\cite{Mog91} to modelling call-by-value languages (the store-free game model is an instance of the monadic framework).\qed


\end{document}